\documentclass[12pt,a4paper]{article}
\usepackage{hyperref}
\usepackage{amsthm,amsmath,amssymb}
\usepackage{verbatim}
\usepackage{graphicx}
\usepackage{url}
\usepackage{xcolor}
\usepackage{float}
\usepackage[T1]{fontenc}
\usepackage{cite}

\usepackage[headsepline,nouppercase]{scrpage2}
\clearscrheadfoot
\ihead{\headmark}
\ohead{\pagemark}
\automark{section}
\setheadsepline{0.1pt}

\theoremstyle{theorem}
\newtheorem{theorem}{Theorem}[section]
\newtheorem{lemma}[theorem]{Lemma}
\newtheorem{proposition}[theorem]{Proposition}
\newtheorem{corollary}[theorem]{Corollary}

\theoremstyle{definition}
\newtheorem{definition}[theorem]{Definition}

\theoremstyle{remark}

\numberwithin{equation}{section}

\newcommand{\bin}{{\rm bin}}
\newcommand{\IN}{\mathbb{N}}
\newcommand{\IZ}{\mathbb{Z}}

\newcommand{\subf}{\mathrm{sf}}
\newcommand{\kxs}{\mathrm{K4\times S5}}
\newcommand{\sxs}{\mathrm{S4\times S5}}

\newcommand{\SfourSfive}{\mathrm{S4\times S5}}
\newcommand{\ssl}{\mathrm{SSL}}
\newcommand{\SSL}{\mathrm{SSL}}
\newcommand{\EXPSPACE}{\mathrm{EXPSPACE}}
\newcommand{\ESPACE}{\mathrm{ESPACE}}
\newcommand{\ALOGTIME}{\mathrm{ALOGTIME}}
\newcommand{\LOGSPACE}{\mathrm{LOGSPACE}}
\newcommand{\main}{\mathrm{main}}

\makeindex

\title
{$\EXPSPACE$-Completeness of the Logics $\kxs$ and $\sxs$ and the Logic of Subset Spaces, \\
Part 2: $\EXPSPACE$-Hardness}

\author{Peter Hertling and Gisela Krommes\\
Fakult\"at f\"ur Informatik \\
Universit\"at der Bundeswehr M\"unchen \\
85577 Neubiberg, Germany \\[2mm]
Email: peter.hertling@unibw.de, gisela.krommes@unibw.de}

\date{\today}

\begin{document}

\maketitle

\begin{abstract}
It is known that the satisfiability problems of the product logics $\kxs$ and $\sxs$ are $\mathrm{NEXPTIME}$-hard and that the satisfiability problem of the logic $\ssl$ of subset spaces is $\mathrm{PSPACE}$-hard. 
We improve these lower bounds for the complexity of these problems by showing that all three problems are $\EXPSPACE$-hard under logspace reduction.
In another paper we show that these problems are in $\ESPACE$. This shows that all three problems are $\EXPSPACE$-complete.
\end{abstract}

\bigskip

\noindent{\bf Keywords:}
bimodal product logics, subset space logic, satisfiability problem, complexity theory, $\EXPSPACE$-completeness

\maketitle

\section{Introduction}
In this article we are concerned with the complexity of the bimodal product logics $\kxs$ and $\sxs$
and with the subset space logic $\ssl$, a bimodal logic as well.
To the best of our knowledge, the complexity of $\kxs$, of $\sxs$, and of $\ssl$ were open problems.
The main results of this article can be summarized in the following theorem.

\begin{theorem}
	The logics	$\kxs$, $\sxs$, and $\ssl$ are $\EXPSPACE$-hard under logspace reduction.
\end{theorem}

Actually, we are considering the satisfiability problems of these three logics, and we are going to show that the satisfiability problems of these logics are $\EXPSPACE$-hard. Of course, this assertion is equivalent to the theorem above because $\EXPSPACE$ is closed under complements. This paper is a continuation of the paper \cite{HK2019-1}, in which we show that these problems are in $\ESPACE$. Both results together imply the following theorem.

\begin{theorem}
	The logics	$\kxs$, $\sxs$, and $\ssl$ are $\EXPSPACE$-complete under logspace reduction.
\end{theorem}

Let us recap the history of the questions and results concerning the complexity of these problems. The following text is almost identical with a corresponding text in ~\cite{HK2019-1}.
In \cite[Question 5.3(i)]{marx1999complexity} Marx posed the question what the complexity of the bimodal logic $\sxs$ is. 
This question is restated and extended to the logic $\kxs$ in~\cite[Problem 6.67, Page 334]{Kurucz2003}.
There it is also stated that
``M. Marx conjectures that these logics are also EXPSPACE-complete''.
That it is desirable to know the complexity of $\ssl$ and similar logics is mentioned by Parikh, Moss, and Steinsvold in \cite[Page 30]{parikh2007topology}
and by Heinemann in \cite[Page 153]{heinemann2016augmenting} and in \cite[Page 513]{heinemann2016subset}. 

For the complexity of the satisfiability problems of the logics $\kxs$ and $\sxs$ the best upper bound known is $\mathrm{N2EXPTIME}$ \cite[Theorem 5.28]{Kurucz2003}, that is, they can be solved by a nondeterministic Turing machine working in doubly exponential time. The best lower bound known for the satisfiability problems of these two logics is $\mathrm{NEXPTIME}$-hardness \cite[Theorem 5.42]{Kurucz2003}; compare also \cite[Table 6.3, Page 340]{Kurucz2003}.
It is known as well that for any $\ssl$-satisfiable formula there exists a cross axiom model of at most doubly exponential size~\cite[Section 2.3]{Dabrowski1992}. This shows that the complexity of the satisfiability problem of $\ssl$ is in $\mathrm{N2EXPTIME}$ as well.  The best lower bound known for $\ssl$ is $\mathrm{PSPACE}$-hardness ~\cite{krommes2003,krommes2003new}.

In this paper we improve the lower bounds $\mathrm{NEXPTIME}$-hardness resp. $\mathrm{PSPACE}$-hardness for the satisfiability problems of these three logics to $\EXPSPACE$-hardness. In \cite{HK2019-1} we show a matching upper bound by showing that these problems are in $\ESPACE$. This shows that they are $\EXPSPACE$-complete. Thus, Marx's conjecture for $\kxs$ and $\sxs$ stated above is true.

The main part of the paper is the $\EXPSPACE$-hardness proof of the satisfiability problem of the logic $\ssl$ in Section~\ref{section:ATMs-SSL}. In order to show this, we shall use Alternating Turing Machines~\cite{Chandra:1981:ALT:322234.322243}.
In this respect, we follow the example of Lange and Lutz~\cite{lange20052} who used Alternating Turing Machines in order to establish a sharp lower bound for the complexity of a certain dynamic logic.
As any language in $\EXPSPACE$ is recognized by an Alternating Turing Machine (ATM) working in exponential time, it is sufficient to show that any language recognized by an Alternating Turing Machine working in exponential time can be reduced in logarithmic space to the satisfiability problem of $\ssl$. We will present such a reduction in Section~\ref{section:ATMs-SSL}. For this purpose we will construct an $\ssl$ formula that describes the computation of an exponential time bounded Alternating Turing Machine. 
In Section~\ref{section:preparations} we shall introduce Alternating Turing Machines. In that section we will also introduce certain formulas that we shall call `shared variables' that we use in order to overcome the problem that ordinary propositional variables are persistent (see Subsection~\ref{subsection:shared-variables}) in $\ssl$. Their usage is illustrated by an implementation of a binary counter in $\ssl$. The $\EXPSPACE$-hardness of the satisfiability problems of $\kxs$ and of $\sxs$ is then shown by reductions. In Section~\ref{section:SSL-S4xS5} we show that the satisfiability problem of $\ssl$ can be reduced in logarithmic space to the satisfiability problem of $\sxs$. And in Section~\ref{section:S4xS5-K4xS5} we show that the satisfiability problem of $\sxs$ can be reduced in logarithmic space to the satisfiability problem of $\kxs$. These reductions are much easier than the reduction of a language recognized by an Alternating Turing Machine working in exponential time to the satisfiability problem of $\ssl$ presented in Section~\ref{section:ATMs-SSL}.

As there may be some interest in a direct proof of the $\EXPSPACE$-hardness of the satisfiability problem of the logic $\sxs$, in an appendix we present such a proof by presenting a direct logspace reduction of any language recognized by an Alternating Turing Machine working in exponential time to the satisfiability problem of $\sxs$. Although the overall structure of this proof is rather similar to the structure of the reduction of ATMs to the satisfiability problem of $\ssl$, there are some important differences. For illustration purposes, we also present an implementation of a binary counter in $\sxs$ in the appendix.

\section{Notations and Preliminaries}
\label{section:preliminaries}

This paper is a continuation of \cite{HK2019-1}. We are going to use the same terminology as in that paper. In order not to repeat the definition of a lot of basic notions we would like to ask the reader to consult the first sections of \cite{HK2019-1} for the needed basic notions from complexity (see the end of the introduction of \cite{HK2019-1}), for the syntax of bimodal formulas (Subsection 2.1 in \cite{HK2019-1}), for $\kxs$- and $\sxs$-product models and -commutator models as well as for cross axiom models (Subsection 2.2 in \cite{HK2019-1}), for the notions of $X$-satisfiability of bimodal formulas, for $X\in\{\kxs,\sxs,\ssl\}$ (Subsection 2.3 in \cite{HK2019-1}), and for some basic notions and observations concerning transitive relations and equivalence relations, in particular for the definition of the relation $\stackrel{\Diamond}{\to}^{\stackrel{L}{\to}}$ induced by a relation $\stackrel{\Diamond}{\to}$ on the equivalence classes with respect to an equivalence relation $\stackrel{L}{\to}$ (Subsection 3.1 in \cite{HK2019-1}). 

As in \cite{HK2019-1} the $\stackrel{L}{\to}$-equivalence class of a point in an $X\times\mathrm{S5}$-commutator model, for $X\in\{\mathrm{K4},\mathrm{S4}\}$, or a cross axiom model will be called the \emph{cloud} of that point.
Finally, for reducing one language to another one we use the logarithmic space bounded reduction as in \cite{Papadimitriou1994}.

\section{Preparations for the Reduction of Alternating Turing Machines to $\ssl$}
\label{section:preparations}

A string $w$ is an element of the language $L(M)$ recognized by an Alternating Turing Machine $M$ iff there exists a so-called accepting tree of $M$ on input $w$. In such a tree each node represents a configuration of $M$. Our idea is to construct a formula depending on $w$ such that the models of this formula have the tree structure of an accepting tree of $M$ on input $w$, where now the nodes of the tree are clouds such that the formulas satisfied in some cloud describe a configuration of $M$. In this way the induced relation $\stackrel{\Diamond}{\to}^{\stackrel{L}{\to}}$ between the clouds serves as a simple temporal operator.

In the following subsection we describe how information is stored and transmitted in a model of such a formula. In particular we introduce certain formulas that we call {\em shared variables} that can transmit information in the $\stackrel{L}{\to}$-direction and by which we can overcome the problem in the logic $\SSL$ that all propositional variables are persistent.
As a first application of this, in Subsection~\ref{subsection:counter} we demonstrate how one can implement a binary counter in the logic $\SSL$. 
In Subsection~\ref{subsection:ATM} we recall the definition of Alternating Turing Machines.

\subsection{Shared Variables}
\label{subsection:shared-variables}

We have to make sure that various kinds of information are stored in a suitable way in any model of the fomula.
We also need to copy and transmit various bits of information, 
both in the $\stackrel{\Diamond}{\to}$-direction (we always depict this as the vertical direction)
as well as in the $\stackrel{L}{\to}$-direction (we always depict this as the horizontal direction).
This will be done by two kinds of formulas.
\begin{itemize}
\item
On the one hand, we need formulas that have the same truth value in the vertical ($\stackrel{\Diamond}{\to}$) direction but can change their truth values in the horizontal ($\stackrel{L}{\to}$) direction. In the case of the logic $\SSL$, for this purpose we can simply use propositional variables as they are persistent anyway. 
\item
On the other hand,  we need formulas that have the same truth value in the horizontal ($\stackrel{L}{\to}$) direction but can change their truth values in the vertical ($\stackrel{\Diamond}{\to}$) direction. Such formulas will be called {\em shared variables} and will be defined now.
\end{itemize}

\begin{definition}[Shared Variables]
\label{def: SV}
For $i\in\IN$ let $A_i$ be special propositional variables, and
let $B$ be another special propositional variable $B$, different from all $A_i$. 
The {\em shared variables $\alpha_i$} are defined as follows:
\[\alpha_i := L(A_i \wedge \Box L B) .\]
\end{definition}
Note that 
\[    \neg\alpha_i \equiv K(\neg A_i \vee \Diamond K \neg B) . \]
See Figure \ref{figure:SV model} for a model of a single shared variable $\alpha_i$ (in the figure we have omitted the index $i$) in $\ssl$ changing its value from $1$ to $0$ and back from $0$ to $1$. In this model the information is stored at the white points which we call \emph{information points}. The gray points are \emph{auxiliary points} that ensure that we obtain a model for the shared variables. Note that the information points differ from the auxiliary points in the value of the propositional variable $B$, which is true at all information points, independent of the value of $\alpha$ stored there, and false at all auxiliary points. Thus, the value of the propositional variable $B$ allows us to distinguish between the information points and the auxiliary points.
\begin{figure}[h]
   \begin{center}
	\includegraphics[width=0.7\linewidth]{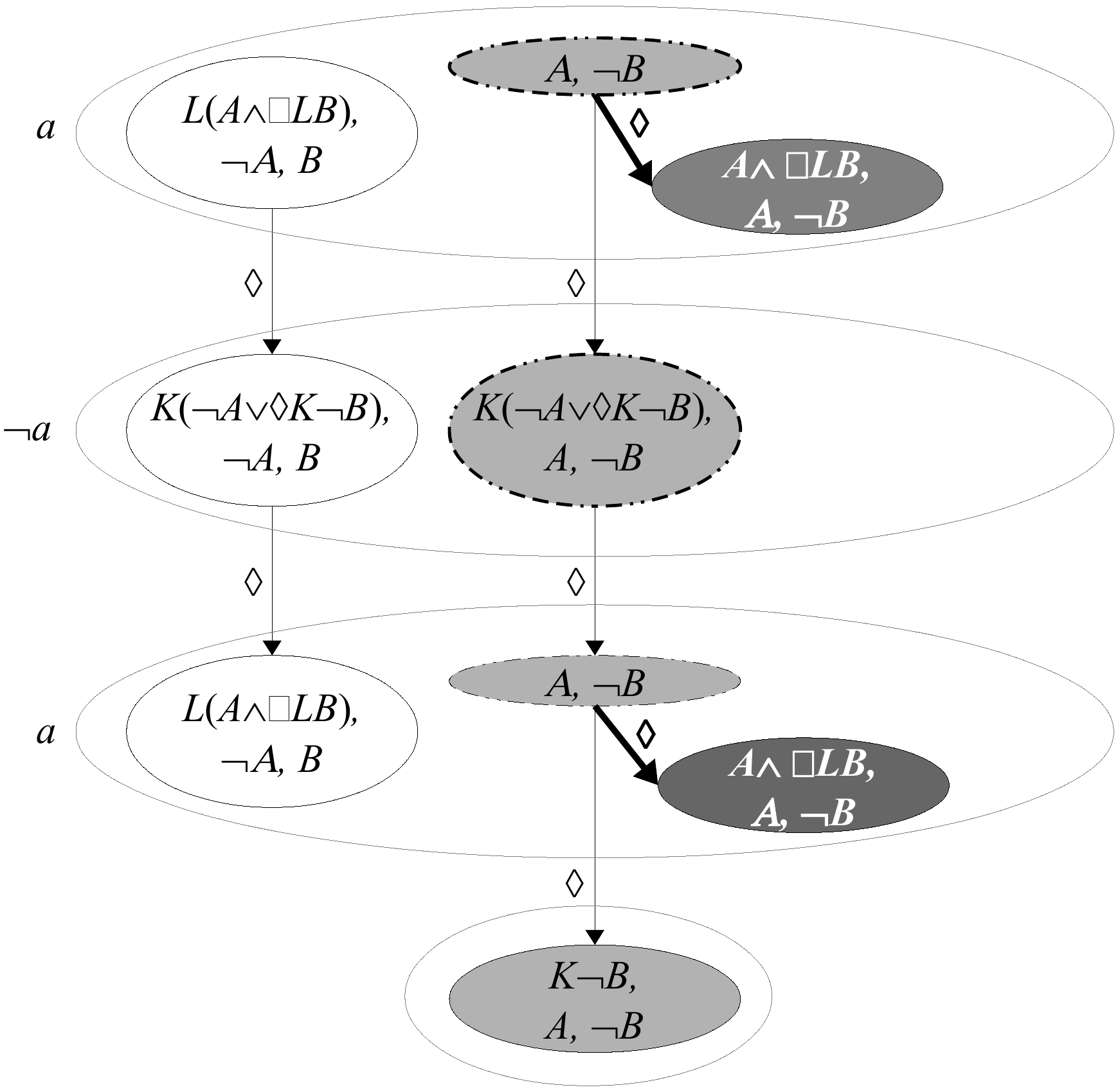}	
	\end{center}
   \caption{Cross axiom model of $\alpha\wedge\Diamond(\neg\alpha\wedge\Diamond\alpha)$.}
   \label{figure:SV model}
\end{figure}

Although the shared variables $\alpha_i$ are formulas we are going to use them as if they were variables. The propositional variables $A_i$ and the propositional variable $B$ that are used in their definition will not be used in any other way.
As a first example of the application of shared variables, in the following section we demonstrate how, using shared variables, one can implement $n$-bit binary counters in $\SSL$. Binary counters are going to play a key role in the simulation of Alternating Turing Machines in $\ssl$.

\subsection{Binary Counters in $\ssl$}
\label{subsection:counter}

Fix some natural number $n\geq 1$. We wish to implement in $\ssl$ a binary $n$-bit counter that counts from $0$ to $2^n-1$. That means, we wish to construct an $\ssl$-satisfiable formula with the property
that any model of it contains a sequence of pairwise distinct points
$p_0,\ldots,p_{2^n-1}$ such that, for each $i\in\{0,\ldots,2^n-1\}$,
at the point $p_i$ the number $i$ is stored in binary form in a certain way.
To describe the implementation of the counter we first introduce some notation.

\begin{itemize}
	\item 
	For a natural number $i$, we define the finite set $\mathrm{Ones}(i) \subseteq\IN$ by
	\[ \sum_{k \in \mathrm{Ones}(i)} 2^k = i, \]
	that is, $\mathrm{Ones}(i)$ is the set of the positions of ones in the binary representation of $i$ (where the positions are counted from the right starting with $0$).
	\item
	We will also need the bits $b_k(i)\in\{0,1\}$ of the binary representation of $i$, for $i,k\in\IN$. They are defined by
	\[ b_k(i) := \begin{cases}
	1 & \text{ if } k \in \mathrm{Ones}(i) , \\
	0 & \text{ if } k \not\in \mathrm{Ones}(i) .
	\end{cases}\]
	\item
	For natural numbers $i,n$ with $n > 0$ and $i\leq 2^n-1$ the
	{\em binary representation of length $n$ of $i$} is the string
	\[ \mathrm{bin}_n(i):= b_{n-1}(i),\ldots,b_0(i).  \]
\end{itemize}
Table \ref{table:abbrev1} lists expressions that we use as abbreviations of formulas.

\begin{table}
		\caption{Some (partially numerical) abbreviations for logical formulas,
			where $\underline{F}=(F_{l-1},\ldots,F_0)$ and $\underline{G}=(G_{l-1},\ldots,G_0)$
			are vectors of formulas.
			As usual, an empty conjunction like $\bigwedge_{h=0}^{-1} F_h$
			can be replaced by any propositional formula that is true always.}
		\label{table:abbrev1}
		\renewcommand{\arraystretch}{1.5}
\begin{center}
		\begin{tabular}{|l|l|l|}
			\hline
			For & the following &is an abbreviation\\[-3mm]
			&expression & of the following formula  \\[-1mm]
			\hline\hline
			$l\geq 1, k\geq -1$ & $(\underline{F}=\underline{G},>k)$ & $\bigwedge_{h=k+1}^{l-1} (F_h \leftrightarrow G_h)$ \\ \hline
			$l\geq 1$ & $(\underline{F}=\underline{G})$ & $(\underline{F}=\underline{G},>-1)$ \\ \hline
			$l\geq 1$, $0\leq i < 2^l$ 
			& $(\underline{F}=\mathrm{bin}_l(i))$ 
			& $\bigwedge_{k \in \mathrm{Ones}(i)}\! F_k \wedge$ \\
			&& $\bigwedge_{k \in \{0,\ldots,l-1\}\setminus\mathrm{Ones}(i)} \neg F_k$  \\ \hline
			$l\geq 1$, $0\leq k < l$ & $\mathrm{rightmost\_zero}(\underline{F},k)$ 
			& $\neg F_k \wedge \bigwedge_{h=0}^{k-1}  F_h$ \\  \hline
			$l\geq 1$, $0\leq k < l$ & $\mathrm{rightmost\_one}(\underline{F},k)$  
			& $F_k \wedge \bigwedge_{h=0}^{k-1} \neg F_h$ \\  \hline 			
		\end{tabular}
\end{center}
\end{table}

The idea of the construction is as follows. 
\begin{itemize}
	\item
	We store the counter values in a vector $\underline{\alpha}:=\alpha_{n-1},\ldots,\alpha_0$ of shared variables. To this end we embed the sequence $p_0,\ldots,p_{2^n-1}$ of points in a sequence of clouds $\mathit{C}_0,\ldots,\mathit{C}_{2^{n-1}}$ such that the cloud $\mathit{C}_i$ contains the point $p_i$ and such that the vector $\underline{\alpha}$ of shared variables satisfied at $p_i$ (and hence at all points in $\mathit{C}_i$) encodes the number $i$.
	\item
	Let $i\leq 2^n-1$ be the number encoded by $\underline{\alpha}$. If $\underline{\alpha}$ contains no $0$ then $i$ has reached its highest posible value, the number $2^n-1$. Otherwise let $k$ be the position of the rightmost $0$. We determine the position $k$ with the aid of the formula $\mathrm{rightmost\_zero}(\alpha,k)$. In order to increment the counter we have to keep all $\alpha_j$ at positions $j>k$ unchanged and to switch all $\alpha_j$ at positions $j\leq k$. We do this in two steps:
	\begin{enumerate}
		\item 
		First me make an $\stackrel{L}{\to}$-step from the point $p_i$ to a point $p'_i$ where we store the number $i+1$ in a vector $\underline{X}:=X_{n-1},\ldots,X_0$ of usual propositional variables by demanding that
		$$p'_i\models (\underline{X}=\underline{\alpha},>k) \wedge 
		\mathrm{rightmost\_one}(\underline{X},k).$$
		\item 
		Then we make a $\stackrel{\Diamond}{\to}$-step from the point $p'_i$ to a 
		point $p_{i+1}$ in the cloud $\mathit{C}_{i+1}$ and demand that 
		$$p_{i+1}\models (\underline{X}=\underline{\alpha}).$$
		Note that in $\ssl$ the value of $\underline{X}$ is copied from $p'_i$ to its $\stackrel{\Diamond}{\to}$-successor $p_{i+1}$ since in $\ssl$ propositional variables are persistent. 
	\end{enumerate}
 	Altogether we demand that for the number $k$
		$$p_i\models L\bigl((\underline{X}=\underline{\alpha},>k)
			\wedge \mathrm{rightmost\_one}(\underline{X},k) \wedge
			\Diamond(\underline{X}=\underline{\alpha}) \bigr).$$
	\item
	Additionally we need a formula to ensure that the starting value is $0$, that is we demand 
	$$p_0\models(\underline{\alpha}=\mathrm{bin}_n(0)).$$
\end{itemize}

We now define the counter formula, for $n>0$.
Remember that
$\underline{\alpha}$ is a vector $(\alpha_{n-1},\ldots,\alpha_0)$
of formulas $\alpha_i$ where $\alpha_i$ is defined by
$\alpha_i := L(A_i \wedge \Box L B)$; compare Definition~\ref{def: SV}.
\begin{eqnarray*}
	\mathrm{\mathrm{counter}}_{\ssl,n}  &:= & B \wedge \ (\underline{\alpha}=\mathrm{bin}_n(0)) 
	\wedge K \Box \Biggl( \bigwedge_{k=0}^{n-1} \Biggl( 
		(B\wedge \mathrm{rightmost\_zero}(\underline{\alpha},k))
		\rightarrow \\
	&& L \biggl(B \wedge (\underline{X}=\underline{\alpha},>k) 
		\wedge \mathrm{rightmost\_one}(\underline{X},k) \wedge \Diamond(\underline{X}=\underline{\alpha})  \biggr) \Biggr) \Biggr).
\end{eqnarray*}

\begin{proposition}
	\label{prop: counter}
	\begin{enumerate}
		\item
		For all $n\in\IN\setminus \{0\}$, the formula $\mathrm{counter}_{\ssl,n}$ is $\ssl$-satisfiable.
		\item 
		For all $n\in\IN\setminus \{0\}$, for every cross axiom model of $\mathrm{counter}_{\ssl,n}$ and for every point $p_0$ in this model with $p_0 \models \mathrm{counter}_{\ssl,n}$
		there exist a sequence of $2^{n}-1$ points $p_1,p_2,\ldots,p_{2^n-1}$
		and a sequence of $2^{n}-1$ points $p_0',p_1',\ldots,p'_{2^n-2}$
		such that 
		\begin{itemize}			
			\item
			for $0\leq i\leq 2^n-1$, \quad 
			$p_i\models (\underline{\alpha}=\mathrm{bin}_n(i))$,
			\item
			for $0\leq i\leq 2^n-2$, \quad $p_i\stackrel{L}{\to}p'_{i}$ and $p'_i\stackrel{\Diamond}{\to}p_{i+1}$ and 
		   $p'_{i}\models (\underline{X}=\mathrm{bin}_n(i+1))$.
		\end{itemize}
	\end{enumerate}
\end{proposition}

\begin{proof}
	For the following let us fix some $n>0$.
	\begin{enumerate}
		\item 
		We construct a cross axiom model $M=(W,\stackrel{L}{\to},\stackrel{\Diamond}{\to},\sigma)$ with a point $p_{0,0}$ satisfying $M,p_{0,0} \models \mathrm{counter}_{\SSL,n}$ as follows; see Figure~\ref{figure:SSL}.
		\begin{figure}
			\begin{center}
				\includegraphics[width=0.9\linewidth]{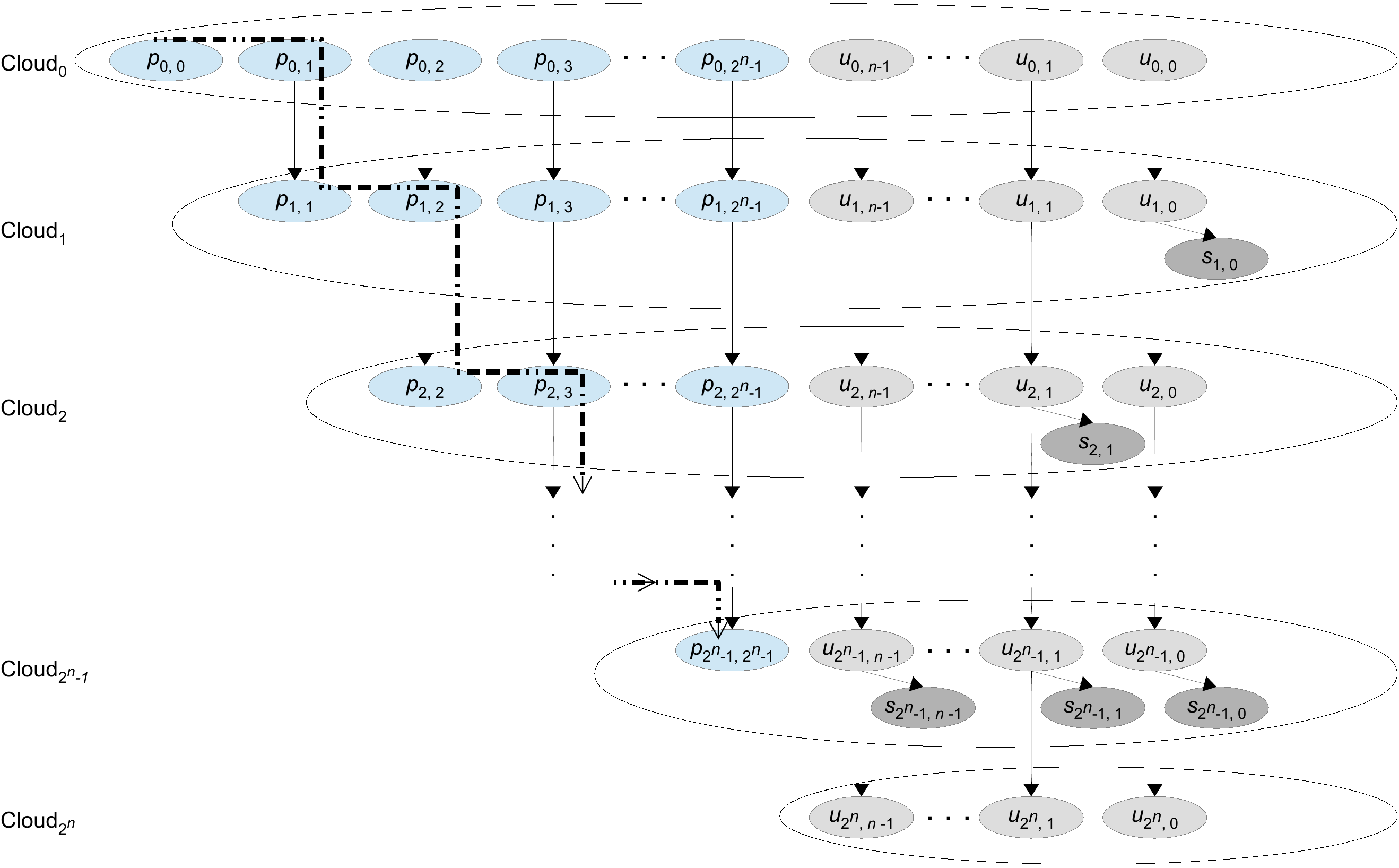}
				\caption{A cross axiom model of the formula $\mathrm{counter}_{\SSL,n}$.}
				\label{figure:SSL}
			\end{center}
		\end{figure}
		We define
		\[W := P \cup U \cup S \]
		where
		\begin{eqnarray*}
			P &:= & \{p_{i,j} ~:~ i,j \in \{0,\ldots,2^n-1\} \text{ and } 
			i \leq j\},  \\
			U &:= & \{u_{i,k} ~:~  i \in \{0,\ldots,2^n\} \text{ and } 
			k \in \{0,\ldots,n-1\} \}, \\ 
			S &:= & \{s_{i,k} ~:~ i \in \{0,\ldots,2^n-1\} \text{ and } 
			k \in \mathrm{Ones}(i)\} .
		\end{eqnarray*}	
		As the relation $\stackrel{L}{\to}$ is supposed to be an equivalence relation we can define it by defining 
		the $\stackrel{L}{\to}$-equivalence classes. These are the sets	
		\[	\mathit{C}_i := \{p_{i,j} ~:~ i \leq j < 2^n\} \cup \{u_{i,k} ~:~  k \in \{0,\ldots,n-1\} \} \cup \{s_{i,k} ~:~ k \in \mathrm{Ones}(i)\}, \]
		for all $i \in \{0,\ldots,2^n-1\}$, and 
		\[ \mathit{C}_{2^n} := \{u_{2^n,k} ~:~ k \in \{0,\ldots,n-1\} \} .  \]
		We define the relation $\stackrel{\Diamond}{\to}$ by:
		\begin{eqnarray*}
			\stackrel{\Diamond}{\to}  &:=& \{ (p_{i,j}, p_{i',j'}) \in P\times P ~:~ i \leq i' \text{ and }  j=j'\} \\
			&&\cup \{ (u_{i,k}, u_{i',k'}) \in U \times U ~:~ i \leq i' \text{ and } k=k'\} \\
			&&\cup \{ (u_{i,k}, s_{i',k'}) \in U \times S ~:~ i \leq i' \text{ and } k=k' \} \\
			&&\cup \{ (s_{i,k}, s_{i',k'}) \in S \times S ~:~ i = i' \text{ and } k=k' \} .  
		\end{eqnarray*}
		It is straightforward to check that $\stackrel{\Diamond}{\to}$ is reflexive and transitive.
		The cross property is satisfied as well.
		Thus, $(W,\stackrel{L}{\to},\stackrel{\Diamond}{\to})$ is a cross axiom frame.
		We define the valuation $\sigma$ by
		\begin{eqnarray*}
			\sigma(A_k) &:= & \{ u_{i,k} ~:~ i \in \{0,\ldots,2^n\}\} \cup\\
			&&\{ s_{i,k} ~:~ i \in \{0,\ldots,2^n-1\} \text{ and }
			k \in \mathrm{Ones}(i)\}, \\
			\sigma(B) &:= & P, \\
			\sigma(X_k) &:= & \{p_{i,j} ~:~ i,j \in \{0,\ldots,2^n-1\} \text{ and }
			k \in \mathrm{Ones}(j)\} ,
		\end{eqnarray*}
		for $k \in \{0,\ldots,n-1\}$.\\
		It is obvious that all of the propositional variables $A_0,\ldots,A_{n-1},B$ and $X_0,\ldots,X_{n-1}$ are persistent.
		Thus, $(W,\stackrel{L}{\to},\stackrel{\Diamond}{\to},\sigma)$ is a cross axiom model.
		We claim $p_{0,0} \models \mathrm{counter}_{\SSL,n}$.
		Before we show this we show the following claim, for all $i\in\{0,\ldots,2^n-1\}$ and for all $p\in\mathit{C}_i$,
		\begin{equation}
		\label{eq:SSL-counter-1}
		p \models (\underline{\alpha}=\mathrm{bin}_n(i)) . 
		\end{equation}
		In order to show this it is sufficient to show
		for all $i\in\{0,\ldots,2^n-1\}$, for all $p \in \mathit{C}_i$, and for all $k \in\{0,\ldots,n-1\}$
		\[ p \models \alpha_k \iff k \in \mathrm{Ones}(i) . \]
		Let us fix some $k \in\{0,\ldots,n-1\}$.
		Note that, for all $p' \in P$, we have $p' \models \neg A_k$,  hence
		\[(\forall p'\in P)\quad p' \models (\neg A_k \vee \Diamond K \neg B) .\]
		Furthermore, $u_{2^n,h} \models K \neg B$ for all $h\in\{0,\ldots,n-1\}$.
		Since for all $i\in\{0,\ldots,2^n\}$ and $h\in\{0,\ldots,n-1\}$
		we have $u_{i,h} \stackrel{\Diamond}{\to} u_{2^n,h}$ we obtain
		$u_{i,h} \models \Diamond K \neg B$ for all $i\in\{0,\ldots,2^n\}$ and $h\in\{0,\ldots,n-1\}$.	
		Hence,
		\[ (\forall u' \in U) \quad u' \models (\neg A_k \vee \Diamond K \neg B) . \]
		This shows that, for any $i\in\{0,\ldots,2^n-1\}$,
		the shared variable $\alpha_k=L(A_k \wedge \Box L B)$ is true in the cloud $\mathit{C}_i$ if, and only if, there exists some $s' \in \mathit{C}_i \cap V$ with $s' \models (A_k \wedge \Box L B)$.
		As $p' \models B$ for all $p' \in P$ and any $s' \in S$ is $\stackrel{L}{\to}$-equivalent to some $p' \in P$, we have
		$s' \models L B$, for all $s' \in S$.
		Actually, for $s' \in S$ we even have $s' \models \Box L B$ as $s'$ does not have any $\stackrel{\Diamond}{\to}$-successors besides itself.
		Thus, the shared variable $\alpha_k$ is true in the cloud $\mathit{C}_i$
		if, and only if, there exists some $s' \in \mathit{C}_i \cap S$ with $s' \models A_k$.
		As the only elements $s' \in \mathit{C}_i \cap S$ are the elements
		$s_{i,h}$ with $h \in \mathrm{Ones}(i)$
		and as $s_{i,h} \models A_k \iff h=k$, we obtain
		for $i\in\{0,\ldots,2^n-1\}$ and for $p \in \mathit{C}_i$,
		\[ p \models \alpha_k \iff k \in \mathrm{Ones}(i) . \]
		We have shown the claim~\eqref{eq:SSL-counter-1}.
		
		We claim that $p_{0,0} \models \mathrm{counter}_{\SSL,n}$.
		Indeed, it is obvious that 
		$$p_{0,0} \models B.$$
		Due to $p_{0,0} \in \mathit{C}_0$ and \eqref{eq:SSL-counter-1} we
		obtain 
		$$p_{0,0} \models (\underline{\alpha}=\mathrm{bin}_n(0)).$$
		Let us assume that for some $p \in W$ and some $k\in\{0,\ldots,n-1\}$ we have
		$p\models (B \wedge \mathrm{rightmost\_zero}(\underline{\alpha},k))$.
		It is sufficient to show that
		\[
		p \models L \bigl(B \wedge (\underline{X}=\underline{\alpha},>k) 
		\wedge \mathrm{rightmost\_one}(\underline{X},k) \wedge \Diamond(\underline{X}=\underline{\alpha})  \bigr) .
		\]
		From $p\models B$ we conclude 
		$p \in P$, hence, there exist
		$i,j\in\{0,\ldots,2^n-1\}$ with $i\leq j$ and with $p=p_{i,j}$.
		Thus, $p \in \mathit{C}_i$. Then we have
		$p \models (\underline{\alpha}=\mathrm{bin}_n(i))$.
		Now, $p \models \mathrm{rightmost\_zero}(\underline{\alpha},k)$ implies
		$\{0,\ldots,n-1\}\setminus \mathrm{Ones}(i) \neq \emptyset$ and
		$k = \min (\{0,\ldots,n-1\}\setminus \mathrm{Ones}(i))$.
		Note that this implies $i<2^n-1$.
		We have $p_{i,j} \stackrel{L}{\to} p_{i,i+1} \stackrel{\Diamond}{\to} p_{i+1,i+1}$,  and since $p_{i,i+1}\in P$ we have
		$$p_{i,i+1} \models B.$$ 
		It is sufficient to show
		$$p_{i,i+1}\models (\underline{X}=\underline{\alpha},>k) \wedge
		\mathrm{rightmost\_one}(\underline{X},k)$$
		and
		$$p_{i+1,i+1}\models (\underline{X}=\underline{\alpha}) .$$
		By definition of $\sigma$ we have for all $j\in\{0,\ldots,n-1\}$ $$p_{i,i+1}\models X_j \quad \Leftrightarrow\quad  j\in\mathrm{Ones}(i+1)$$
		and hence on the one hand 
		$p_{i,i+1}\models( \underline{X}=\mathrm{bin}_n(i+1))$. Due to \eqref{eq:SSL-counter-1}, we have on the other hand
		$p_{i,i+1} \models (\underline{\alpha} = \mathrm{bin}_n(i))$. This proves the first claim. For the second claim we 
		observe that by definition of $\sigma$ also 
		$p_{i+1,i+1}\models (\underline{X}=\mathrm{bin}_n(i+1))$. Since $p_{i+1,i+1}\in\mathit{C_{i+1}}$ we also have by \eqref{eq:SSL-counter-1} that $p_{i+1,i+1}\models (\underline{\alpha}=\mathrm{bin}_n(i+1))$ and hence $p_{i+1,i+1}\models (\underline{X} =\underline{\alpha})$.\\
		Thus, we have constructed a cross axiom model for $\mathrm{counter}_{\SSL,n}$.
		\item
		Suppose there is a cross axiom model $M$ of the formula $\mathrm{counter}_{\SSL,n}$ and some point $p_0\in M$ with 
		$M,p_0\models \mathrm{counter}_{\SSL,n}$.
		We show by induction that the claimed sequences of points 
		$p_1,\ldots p_{2^n-1}$ and $p'_0,\ldots p'_{2^n-2}$
		with the claimed properties and additionally with 
		$p_i \models B$, for $0\leq i \leq 2^n-1$, and with
		$p'_i \models B$, for $0\leq i \leq 2^n-2$, exist.
		In addition, we show that there exist points $t_i$ with 
		$p_0 \stackrel{L}{\to} t_i$ and $t_i \stackrel{\Diamond}{\to} p_i$, for $1\leq i \leq 2^n-1$.
		Note that the sequences $p_0,\ldots p_{2^n-1}$ and $p'_0,\ldots p'_{2^n-2}$ are supposed to form a ``staircase'' as in Figure~\ref{figure:staircase}.
		\begin{figure}
			\begin{center}
				\includegraphics[width=0.4\linewidth]{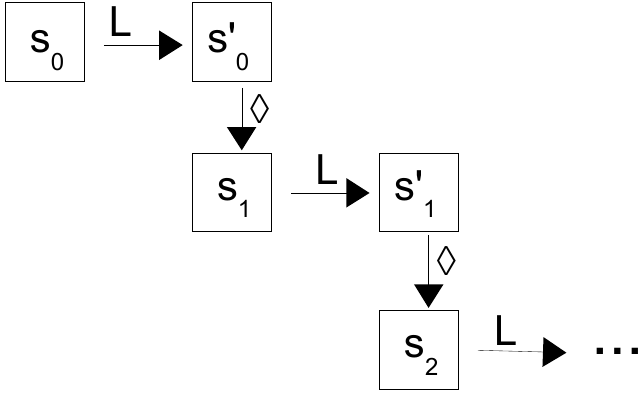}
			\end{center}
			\caption{Staircase of points.}
			\label{figure:staircase}	
		\end{figure}
		By definition,
		$p_0\models B \wedge (\underline{\alpha}=\mathrm{bin}_n(0))$.
		By induction hypothesis, let us assume that for some $m$ with $0 \leq m<2^n-1$
		there exist $p_1,\ldots,p_m$ and $p'_0,\ldots,p_{m-1}'$ with
		\[ p_0\, \stackrel{L}{\to}\, p'_0\, \stackrel{\Diamond}{\to}\, p_1\, \stackrel{L}{\to} \ldots \stackrel{L}{\to}\, p'_{m-1}\, \stackrel{\Diamond}{\to}\, p_m ,  \]
		with
		\[ p_i\models (B \wedge (\underline{\alpha}=\mathrm{bin}_n(i))) , \]
		for $0\leq i \leq m$, and with
		\[ p'_i\models (B \wedge (\underline{X}=\mathrm{bin}_n(i+1))) , \]
		for $0\leq i < m$, and 
		that there are $t_i$ with $p_0 \stackrel{L}{\to} t_i$ and $t_i \stackrel{\Diamond}{\to} p_i$, for $1\leq i \leq m$.
		Since $m<2^n-1$, the set $\{0,\ldots,n-1\}\setminus\mathrm{Ones}(m)$ is nonempty.
		With $k:= \min(\{0,\ldots,n-1\}\setminus\mathrm{Ones}(m))$ we have
		\[ p_m\models \mathrm{rightmost\_zero}(\underline{\alpha},k). \]
		Due to $p_0\models \mathrm{counter}_{\SSL,n}$ as well as $p_0 \stackrel{L}{\to} t_m \stackrel{\Diamond}{\to} p_m$
		this implies
		\[ p_m \models 
			L \bigl( B \wedge (\underline{X}=\underline{\alpha},>k) 
			\wedge \mathrm{rightmost\_one}(\underline{X},k) 
			\wedge \Diamond(\underline{X}=\underline{\alpha})  \bigr) .		
		\]
		Thus, there must exist points $p'_m$ and $p_{m+1}$ satisfying
		$p'_{m} \models B$ as well as
		$p_m\stackrel{L}{\to} p'_m \stackrel{\Diamond}{\to}p_{m+1}$,
		\[ p'_m \models(\underline{X}=\underline{\alpha},>k) 
		\wedge \mathrm{rightmost\_one}(\underline{X},k) \]		
		and
		\[ p_{m+1} \models  (\underline{X}=\underline{\alpha}) . \]
		We have to show
		\[ p'_m \models(\underline{X}=\mathrm{bin}_n(m+1)) \]
		and
		\[ p_{m+1} \models B\wedge(\underline{\alpha}=\mathrm{bin}_n(m+1)) . \]
		Due to the fact that $p'_m$ is an element of the same cloud as $p_m$, and $\underline{\alpha}$ has the same value in all points in a cloud we obtain
		\[ p'_m \models (\underline{\alpha}=\mathrm{bin}_n(m)) .\] Together with 
		\[ p'_m \models(\underline{X}=\underline{\alpha},>k) 
		\wedge \mathrm{rightmost\_one}(\underline{X},k) \]	
		this implies
		\[ p'_m \models(\underline{X}=\mathrm{bin}_n(m+1)) \]
		(the values of the leading bits $\alpha_{n-1},\ldots,\alpha_{k+1}$ of $\underline{\alpha}$ are copied to the leading bits $X_{n-1},\ldots,X_{k+1}$ and the other bits are defined explicitly by $p'_{m} \models  \mathrm{rightmost\_one}(\underline{X},k)$ so that the binary value of $\underline{X}$ is $m+1$).
		From $p'_m \stackrel{\Diamond}{\to} p_{m+1}$ and the fact that in $\SSL$ propositional variables are persistent we obtain
		$p_{m+1} \models (B \wedge (\underline{X}=\mathrm{bin}_n(m+1)))$.
		Using $p_{m+1} \models (\underline{X}=\underline{\alpha})$
		we obtain
		$p_{m+1} \models (\underline{\alpha}=\mathrm{bin}_n(m+1))$.
		Finally, the cross property applied to $t_m \stackrel{\Diamond}{\to} p_m$ and $p_m \stackrel{L}{\to} p'_m$
		implies that there exists a point $t_{m+1}$ with $t_m \stackrel{L}{\to} t_{m+1}$ and $t_{m+1} \stackrel{\Diamond}{\to} p'_m$.
		Using additionally $p_0 \stackrel{L}{\to} t_m$ and $p'_m \stackrel{\Diamond}{\to} p_{m+1}$ we obtain
		$p_0 \stackrel{L}{\to} t_{m+1}$ and $t_{m+1} \stackrel{\Diamond}{\to} p_{m+1}$.
		This ends the proof of the second assertion.		
	\qedhere
	\end{enumerate}
\end{proof}

\subsection{Alternating Turing Machines}
\label{subsection:ATM}

The concept of an \emph{Alternating Turing Machine} (ATM) was set forth by Chandra and Stockmeyer
\cite{chandra1976alternation} and independently by Kozen \cite{kozen1976parallelism} in 1976, with a joint journal publication in 1981 \cite{Chandra:1981:ALT:322234.322243}. 
We are going to use a variant of ATMs with a single tape as in \cite{lange20052}.
This is justified by the fact that one-tape ATMs can efficiently simulate multi-tape ATMs;
see~\cite[Proposition 3.4]{chandra1976alternation}. For an even more efficient
simulation of multi-tape ATMs by one-tape ATMs see \cite{paul1980alternation}.

An ATM is a nondeterministic Turing machine
where some configurations are ``or'' configurations that accept
if at least one of their successors does, while other configurations are
``and'' configurations that accept if all of their successors accept.
The mode of each configuration (``and'' vs. ``or'') is determined
by the state of the configuration.
There are two special states called $q_{accept}$ and $q_{reject}$.
All other states are either \emph{universal states} or \emph{existential states}.

For a relation $\delta\subseteq X \times Y$ and any $x\in X$ we write
$\delta(x):=\{y \in Y \mid (x,y) \in\delta\}$.
\begin{definition}
	An \emph{alternating Turing Machine $M$} is a quintuple 
	\[ M = (Q, \Sigma, \Gamma, q_0, \delta),\]
	where
	\begin{itemize}
		\item $Q$, the set of \emph{states} of $M$, is the disjoint union of the following
		four sets:
		\begin{itemize}
		\item $Q_\exists$, a finite set, its elements are called \emph{existential states},
		\item $Q_\forall$, a finite set, its elements are called \emph{universal states},
		\item $\{q_\mathrm{accept}\}$, a one-element set, its element is called \emph{accepting state},
		\item $\{q_\mathrm{reject}\}$, a one-element set, its element is called \emph{rejecting state},
		\end{itemize}
		\item $\Sigma$ is a nonempty finite set, called the \emph{input alphabet},
		\item $\Gamma \supseteq \Sigma$ is a finite set containing a \emph{blank} symbol 
				$\#\not\in\Sigma$, we call $\Gamma$ the {\em tape alphabet},
		\item $q_0 \in Q$ is the {\em initial state},  
		\item $\delta \subseteq 	
				(Q\times\Gamma)\times(Q\times\Gamma\times\{\mathit{left},\mathit{right}\})$ 		
		    	is the {\em transition relation},
	\end{itemize}
and where $\delta$ satisfies the condition
\[ \delta(q,a) \begin{cases}
    = \emptyset & \text{  for $q \in \{q_\mathrm{accept}, q_\mathrm{reject}\}$ and $a\in \Gamma$}, \\
    \neq \emptyset & \text{  for $q \in Q_\exists\cup Q_\forall$ and $a\in \Gamma$}.
    \end{cases} \]
\end{definition}

A {\em configuration} of an alternating Turing machine $M$ is an element 
$(q,z,\gamma)$ of
\[ C_M = Q \times \IZ \times \Gamma^\IZ \]
where $q\in Q$ is the current state of the finite control,
$z \in\IZ$ is the current position of the tape head (that is, the number of the cell on which the
tape head is positioned), and
the function $\gamma:\IZ\to\Gamma$ represents the current tape content
and satisfies the condition $\gamma(z)=\#$ for all but finitely many $z \in \IZ$.
A configuration represents an instantaneous description of $M$ at some point in a  computation. 
The \emph{initial configuration} of $M$ on input $w=w_1\ldots w_n\in\Sigma^*$ 
with $w_i\in\Gamma$ for $i=1\ldots, |w|$ is
\[ \sigma_M(w) := (q_0, 0, \gamma_w) \]
where $\gamma_w:\IZ\to\Gamma$ is defined by
\[ \gamma_w(z):=\begin{cases}
      \# & \text{ if } z\leq 0 \text{ or } z > |w|, \\
      w_z & \text{ if } z \in \{1,\ldots,|w|\}.
      \end{cases} \]
That means that at the start of the computation the tape head is positioned on cell no. $0$,
and the input string $w$ is contained in the cells with the numbers $1$ to $|w|$ while
all other cells contain the blank $\#$.

For two configurations $c$ and $c'$ we write $c\vdash c'$ and say \emph{$c'$ is a successor of $c$},
if, according to the transition relation $\delta$, the configuration $c'$
can be reached from the configuration $c$ in one step (this is defined in the usual way).
The reflexive-transitive closure of $\vdash$ is denoted $\vdash^*$.
A \emph{computation} or \emph{computation path} of $M$ on input $w$
is a sequence $c_0\vdash\ldots\vdash c_m$, where  $c_0 = \sigma_M(w)$. 	
In the following we will only consider ATMs $M$ such that there exists a function
$t:\IN\to\IN$ such that for any $n\in\IN$ and any possible input string $w \in\Sigma^n$,
any computation path of $M$ on input $w$ has length at most $t(n)$, that is,
if $c_0,\ldots,c_m$ is a computation path of $M$ on input $w$ then $m \leq t(n)$.
In this case we say that $M$ \emph{works  in time $t$}.
For such machines $M$ we can define the \emph{language $L(M)\subseteq \Sigma^*$
recognized by $M$} as follows: for $w \in\Sigma^*$
\[ w \in L(M) :\iff \text{there exists an ``accepting tree of $M$ on input $w$''}. \]
An \emph{accepting tree of $M$ on input $w$}
is a finite rooted and labeled tree each of whose nodes is labeled
with a configuration of $M$ such that the following five properties hold true:
\begin{enumerate}
\item[I.]
The root of the tree is labeled with the initial configuration $\sigma_M(w)$ of $M$ on input $w$.
\item[II.]
If $c$ is the label of an internal node of the tree then the labels of its successors 
are configurations $c'$ satisfying $c \vdash c'$
(note that this implies that the state $q$ of $c$ is an element of $Q_\exists\cup Q_\forall$).
\item[III]
If $c$ is the label of an internal node of the tree then the labels of its successors 
are pairwise different configurations.
\item[IV.]
If $c$ is the label of an internal node of the tree 
and the state $q$ of $c$ is an element of $Q_\forall$
then for every configuration $c'$ with $c\vdash c'$ there exists a successor node labeled
with $c'$.
\item[V.]
If $c$ is the label of a leaf of the tree then the state of $c$ is equal to $q_\mathrm{accept}$.
\end{enumerate}
Note that these conditions imply that, if $c$ is the label of a node of the tree 
and the state $q$ of $c$ is an element of $Q_\exists$,
then this node is an internal node, hence, it has a successor.
Let the height of a rooted tree be the length of the longest path in the tree.
It is clear that if there is an accepting tree of $M$ on input $w$ then its height
is at most $t(|w|)$.

The time complexity class $\mathrm{AEXPTIME}$ is the class of all languages $L$
such that there exist an alternating Turing machine $M$ with $L=L(M)$
and a polynomial $p$ such that $M$ works in time $2^{p(n)}$.
We will make use of the fundamental fact 
$\mathrm{AEXPTIME}=\mathrm{EXPSPACE}$~\cite[Corollary 3.6]{chandra1976alternation}.

For technical purposes we will also need the following notion.
A \emph{partial tree of $M$ on input $w$} is a finite rooted and labeled
tree each of whose nodes is labeled with a configuration of $M$
that satisfies the same four properties I, II, III, and IV as
an accepting tree of $M$ on input $w$ and instead of the property V the following 
weaker property:
\begin{enumerate}
\item[$\mathrm{V}^\prime$.]
If $c$ is the label of a leaf of the tree then the state of $c$
is different from $q_\mathrm{reject}$.
\end{enumerate}
It is clear that any accepting tree of $M$ on input $w$
is a partial tree of $M$ on input $w$.
Usually we will write a partial tree of $M$ on input $w$
as a triple $T=(V,E,c)$ where
$V$ is the set of nodes of the tree $T$, where $E\subseteq V\times V$ is the set
of edges of the tree $T$
(note that the root $\mathit{root}$ of the tree is uniquely determined by $V$ and $E$:
it is the only node that does not have an incoming edge),
and where $c:V\to Q \times \IZ \times \Gamma^\IZ$
is the labeling of the tree.
Let $E^*$ be the reflexive-transitive closure of $E$.
We will often need the following data associated with any computation node $v\in V$:
\begin{itemize}
\item
The time $\mathit{time}(v)$ of $v$.
This is the number of edges of the unique path in $T$ from $\mathit{root}$ to $v$.
Note that $\mathit{time}(\mathit{root})=0$.
\item
The configuration $c(v)$ of $v$. This is the configuration with which the node $v$ is labeled.
\item
The state $\mathit{state}(v)$ of the configuration $c(v)$.
\item
The position $\mathit{pos}(v)$ of the tape head in the configuration $c(v)$.
\item
The symbol $\mathit{read}(v)$ in the cell $\mathit{pos}(v)$ in the configuration $c(v)$.
\item
The predecessor $\mathit{pred}(v)$ of $v$ in the tree $T$, for $v\neq \mathit{root}$.
\item
The symbol $\mathit{written}(v)$ that has been written in the computation step that lead to this node $v$,
for $v\neq \mathit{root}$.
Note that this is the symbol now contained in the cell
$\mathit{pos}(\mathit{pred}(v))$
on which the tape head was positioned in the previous configuration.
\end{itemize}

\section[Reduction of ATMs Working in Exponential Time to $\ssl$]
{Reduction of Alternating Turing Machines Working in Exponential Time to $\ssl$}
\label{section:ATMs-SSL}
In this section we prove the following theorem.

\begin{theorem}
\label{theorem:SSL-EXPSPACE-hard}
The satisfiability problem of $\ssl$ is $\EXPSPACE$-hard under logarithmic space reduction.
\end{theorem}

As explained before, we are going to show this by showing that any language $L$ recognized by an Alternating Turing Machine working in exponential time can be reduced in logarithmic space to the satisfiability problem of $\ssl$.

In the following subsection we will first describe the idea of the reduction
and then define the reduction function $f_\ssl$ in detail.
In Subsection~\ref{subsection:reduction-ATMs-SSL-complexity}
we will show that it can be computed in logarithmic space.
In the final two subsections we are going to show that it is corrrect.
First we show that in case $w\in L$ the formula
$f_\ssl(w)$ is $\ssl$-satisfiable by explicitly 
constructing a cross axiom model for $f_\ssl(w)$.
In the last section we show that if $f_\ssl(w)$ is 
$\ssl$-satisfiable then $w$ is an element of $L$.

\subsection{Construction of the Formula}
\label{subsection:definition-ATMs-SSL}

Let $L \in \EXPSPACE$ be an arbitrary language
over some alphabet $\Sigma$, that is, $L \subseteq \Sigma^*$.
We are going to show that there is a logspace computable
function $f_\ssl$ mapping strings to strings such that,
for any $w\in\Sigma^*$, 
\begin{itemize}
\item
$f_\ssl(w)$ is a bimodal formula and 
\item
$f_\ssl(w)$ is $\ssl$-satisfiable if, and only if, $w \in L$.
\end{itemize}
Once we have shown this, we have shown the result.
In order to define this desired reduction function $f_\ssl$,
we are going to make use of an Alternating Turing Machine for $L$.
Since $\EXPSPACE = \mathrm{AEXPTIME}$, there
exist an Alternating Turing Machine
$M = (Q, \Sigma, \Gamma, q_0, \delta)$ 
and a univariate polynomial $p$ such that
$M$ accepts $L$, that is, $L(M)=L$, and such that the time used by $M$
on arbitrary input of length $n$ is bounded by $2^{p(n)}-1$.
We can assume without loss of generality
$Q=\{0,\ldots,|Q|-1\}$, $\Gamma=\{0,\ldots,|\Gamma|-1\}$, that 
the coefficients of the polynomial $p$ are natural numbers and that, for all $n\in\IN$,
we have $p(n) \geq n$ and $p(n)\geq 1$. In the following, whenever we have fixed some $n\in\IN$, we set 
\[ N:=p(n) . \] 
Let us consider an input word $w \in\Sigma^n$ of length $n$, for some $n \in\IN$,
and let us sketch the main idea of the construction of the formula
$f_\ssl(w)$.
The formula $f_\ssl(w)$ will describe the possible computations
of $M$ on input $w$ in the following sense: any cross axiom model
of $f_\ssl(w)$ will essentially contain an accepting tree
of $M$ on input $w$, and if $w \in L(M)$ then there exists an accepting tree of $M$ on input $w$
and one can turn this into a cross axiom model of $f_\ssl(w)$.
In such a model, any node in an accepting tree 
of $M$ on input $w$ will be modeled by a cloud
(that is, by an $\stackrel{L}{\to}$-equivalence class) in which certain shared variables
(we use the notion ``shared variables'' in the same sense as in 
Subsection~\ref{subsection:shared-variables})
will have values that describe the data of the computation node
that are important in this computation step. Which data are these?
First of all, we need the time of the computation node.
We assume that the computation starts with the initial configuration
of $M$ on input $w$ at time $0$.
Since the  ATM $M$ needs at most $2^N-1$ time steps,
we can store the time of each computation node in a binary counter
counting from $0$ to $2^N-1$. 
Since during each time step at most one additional cell
either to the right or to the left of the previous cell can be visited,
we can describe  any configuration reachable during a computation of $M$
on input $w$ by the following data:
{\setlength{\leftmargini}{2.15em}\begin{itemize}
\item
the current content of the tape, which is a string in
$\Gamma^{2 \cdot (2^N-1)+1} = \Gamma^{2^{N+1}-1}$,
\item
the current tape head position,
which is a number in $\{0,\ldots,2^{N+1}-2\}$.
\end{itemize}}
We assume that in the initial configuration on input $w$ the
tape content is $\#^{2^N}w\#^{2^N-1-n}$ 
(remember that we use $\#$ for the blank symbol) and that
the tape head scans the blank $\#$ to the left of the first symbol of $w$,
that is, the position of the tape head is $2^{N}-1$.
If a cloud in a cross axiom model of $f_\ssl(w)$ 
describes a computation node of $M$ on input $w$
then in this cloud the following shared variables
will have the following values:
\begin{itemize}
\item
a vector $\underline{\alpha}^\mathrm{time}=
(\alpha^\mathrm{time}_{N-1},\ldots,\alpha^\mathrm{time}_0)$ giving in binary
the current time of the computation,
\item
a vector $\underline{\alpha}^\mathrm{pos}=
(\alpha^\mathrm{pos}_{N},\ldots,\alpha^\mathrm{pos}_0)$ giving in binary
the current position of the tape head,
\item
a vector $\underline{\alpha}^\mathrm{state}=
(\alpha^\mathrm{state}_0,\ldots,\alpha^\mathrm{state}_{|Q|-1})$ giving in unary
the current state of the computation
(here ``unary'' means: exactly one of the shared variables
$\alpha^\mathrm{state}_i$ will be true,
namely the one with $i$ being the current state),
\item
a vector $\underline{\alpha}^\mathrm{written}=
(\alpha^\mathrm{written}_{0},\ldots,\alpha^\mathrm{written}_ {|\Gamma|-1})$ giving in unary
the symbol that has just been written into the cell that has just been left,
unless the cloud corresponds to the first node in the computation tree --- in that
case the value of this vector is irrelevant
(here ``unary'' means: exactly one of the variables $\alpha^\mathrm{written}_i$ will be true,
namely the one with $i$ being the symbol that has just been written), 
\item
a vector $\underline{\alpha}^\mathrm{read}=
(\alpha^\mathrm{read}_0,\ldots,\alpha^\mathrm{read}_{|\Gamma|-1})$
giving in unary the symbol in the current cell
(here ``unary'' means: exactly one of the shared variables
$\alpha^\mathrm{read}_i$ will be true,
namely the one with $i$ being the symbol in the current cell). 
\end{itemize}
The formula $f_\ssl(w)$ has to ensure that for any possible
computation step in an accepting tree starting from such a computation node
there exists a cloud describing the corresponding successor node in the accepting tree.
In this new cloud, the value of the counter for the time 
$\underline{\alpha}^\mathrm{time}$ has to be incremented.
This can be done by the technique described in Subsection~\ref{subsection:counter}
for implementing a binary counter.
In parallel, we have to make sure that in this new cloud also the vectors 
$\underline{\alpha}^\mathrm{pos}$,
$\underline{\alpha}^\mathrm{state}$,
$\underline{\alpha}^\mathrm{written}$, and
$\underline{\alpha}^\mathrm{read}$
are set to the right values.
For the vectors $\underline{\alpha}^\mathrm{pos}$, 
$\underline{\alpha}^\mathrm{state}$, and
$\underline{\alpha}^\mathrm{written}$
these values can be computed using the corresponding
element of the transition relation $\delta$ of the ATM.
For example, $\underline{\alpha}^\mathrm{pos}$
has to be decremented by one if the tape head moves to the left,
and it has to be incremented by one if the tape head moves to the right.
Also the new state 
(to be stored in $\underline{\alpha}^\mathrm{state}$) and the symbol written into the cell that has
just been left (to be stored in $\underline{\alpha}^\mathrm{written}$)
are determined by the data of the previous computation node
and by the corresponding element of the transition relation $\delta$.\\
But the vector $\underline{\alpha}^\mathrm{read}$ is supposed
to describe the symbol in the current cell.
This symbol is not determined by the current computation step
but has either been written the last time when this cell has been visited during this computation or,
when this cell has never been visited before,
the symbol in this cell is still the one that was contained in this cell 
before the computation started.
How can one ensure that $\underline{\alpha}^\mathrm{read}$ is set to the right value?
If the current cell has never been visited before,
we have to make sure
that the value is set to the correct value describing the inital content of this cell.
Otherwise, we make use of the cross property.
The point in the new cloud whose existence is enforced by the formula
must have a cross point in any cloud corresponding to any previous computation node.
The idea is that one of these cross points picks up the right value
in the right cloud.  We are going to make sure that the cloud is identified that corresponds
to the configuration after the previous visit of the same cell during the computation.
Then in the cloud corresponding to this configuration
the value of $\underline{\alpha}^\mathrm{written}$
will tell us the symbol that has been written into the current cell
during the previous visit.
In order to identify the correct cloud of the step after the previous visit
of the current cell and to copy the value of the symbol,
the formula $f_\ssl(w)$ will ensure that any cloud
describing a computation node will contain a point
in which the following (persistent!) propositional variables
have the following values:
\begin{itemize}
\item
a vector $\underline{X}^\mathrm{time}=
(X^\mathrm{time}_{N-1},\ldots,X^\mathrm{time}_0)$ giving in binary
the current time of the computation,
\item
a vector $\underline{X}^\mathrm{pos}=
(X^\mathrm{pos}_{N},\ldots,X^\mathrm{pos}_0)$ giving in binary
the current position of the tape head, that is the position of the current cell,
\item
a vector $\underline{X}^\mathrm{read}=
(X^\mathrm{read}_0,\ldots,X^\mathrm{read}_{|\Gamma|-1})$ giving in unary
the symbol in the current cell, 
(here ``unary'' means: exactly one of the variables $X^\mathrm{read}_i$ will be true,
namely the one with $i$ being the symbol in the current cell). 
\item
a vector $\underline{X}^{\text{time-apv}}=
(X^{\text{time-apv}}_{N-1},\ldots,X^{\text{time-apv}}_0)$ giving in binary
the time one step after the previous visit of the cell,
if it has been visited before
(``$\text{time-apv}$'' stands for ``time after previous visit'');
otherwise this vector will have the binary value $0$.
\end{itemize}
Now we come to the formal definition of the formula $f_\ssl(w)$.
The formula $f_\ssl(w)$ will have the following structure:
\begin{eqnarray*}
  f_\ssl(w) 
    &:= & K\Box \mathit{uniqueness} \\
    && \wedge \mathit{start} \\
    && \wedge K\Box \mathit{time\_after\_previous\_visit} \\
    && \wedge K\Box \mathit{get\_the\_right\_symbol} \\
    && \wedge K\Box \mathit{computation} \\
    && \wedge K\Box \mathit{no\_reject} .
\end{eqnarray*}
The formula $f_\ssl(w)$ will contain the following propositional variables:
\begin{eqnarray*}
&& B, \\
&& A^\mathrm{time}_{N-1}, \ldots, A^\mathrm{time}_{0}, 
 A^\mathrm{pos}_{N}, \ldots, A^\mathrm{pos}_{0}, 
 A^\mathrm{state}_{0}, \ldots, A^\mathrm{state}_{|Q|-1}, 
 A^\mathrm{written}_{0},\ldots, A^\mathrm{written}_{|\Gamma|-1}, 
 A^\mathrm{read}_{0},\ldots, A^\mathrm{read}_{|\Gamma|-1},\\
&& X^\mathrm{time}_{N-1}, \ldots, X^\mathrm{time}_{0}, 
 X^{\text{time-apv}}_{N-1},\ldots, X^{\text{time-apv}}_{0}, 
 X^\mathrm{pos}_{N}, \ldots, X^\mathrm{pos}_{0}, 
 X^\mathrm{read}_{0},\ldots, X^\mathrm{read}_{|\Gamma|-1}.
\end{eqnarray*}
For $\mathit{string} \in\{\mathrm{time}, \mathrm{pos}, \mathrm{state},
\mathrm{written}, \mathrm{read}\}$ and natural numbers $k$ we define
\[ \alpha^\mathit{string}_k := L(A^\mathit{string}_k \wedge \Box L B). \]
These formulas $\alpha^\mathit{string}_k$ are the shared variables we talked about above.
We are now going to define the subformulas of $f_\ssl(w)$.
We will use the abbreviations introduced above,
in Table~\ref{table:abbrev1}, and in Table~\ref{table:abbrev2}.

\begin{table}
\caption{Some (partially numerical) abbreviations for logical formulas,
where $\underline{F}=(F_{l-1},\ldots,F_0)$ and $\underline{G}=(G_{l-1},\ldots,G_0)$ are vectors of formulas.
An empty conjunction like $\bigwedge_{h=0}^{-1} F_h$ can be replaced by any propositional formula that is true always.
An empty disjunction like $\bigvee_{h=0}^{-1} F_h$ can be replaced by any propositional formula that is false always.
}
\label{table:abbrev2}
\renewcommand{\arraystretch}{1.5}
\begin{center}
\begin{tabular}{|l|l|l|}
\hline
For & the following & is an abbreviation \\[-2mm]
& expression & of the following formula \\ \hline\hline
$l\geq 1$ & $\mathrm{unique}(\underline{F})$
  & $\bigvee_{k=0}^{l-1} F_k \wedge \bigwedge_{k=0}^{l-1}\bigwedge_{m=k+1}^{l-1} \neg (F_k \wedge F_m)$ \\ \hline
$l\geq 1$ & $(\underline{F} \neq \underline{G})$ 
  & $\neg (\underline{F} = \underline{G})$ \\ \hline
$l\geq 1$ & $(\underline{F} < \underline{G})$ 
  & $\bigvee_{k=0}^{l-1} \left( (\underline{F}=\underline{G},>k) \wedge \neg F_k \wedge G_k \right)$ \\ \hline
$l\geq 1$ & $(\underline{F} \leq \underline{G})$ 
  & $(\underline{F} < \underline{G}) \vee (\underline{F} = \underline{G})$ \\ \hline
$l\geq 1$ & $(\underline{F} = \underline{G} + 1)$ 
  & $\bigvee_{k=0}^{l-1} \bigl( (\underline{F}=\underline{G},>k)$ \\
&& \quad $\wedge \mathrm{rightmost\_one}(\underline{F},k)$ \\
&& \quad $\wedge \mathrm{rightmost\_zero}(\underline{G},k) \bigr)$ \\ \hline
$l\geq 1$ & $(\underline{F} \neq \underline{G} + 1)$ 
  & $\neg (\underline{F} = \underline{G} + 1)$ \\ \hline
$l\geq 1$, $0\leq i < 2^l$ & $(\underline{F}<\mathrm{bin}_l(i))$ 
      & $\bigvee_{k \in \mathrm{Ones}(i)} \bigl( \neg F_k
       \wedge \bigwedge_{h \in  \{k+1,\ldots,l-1\}\setminus\mathrm{Ones}(i)} \neg F_h\bigr)$ \\ \hline
$l\geq 1$, $0\leq i < 2^l$ & $(\underline{F}\leq\mathrm{bin}_l(i))$ 
      & $(\underline{F}<\mathrm{bin}_l(i)) \vee (\underline{F}=\mathrm{bin}_l(i))$  \\ \hline
$l\geq 1$, $0\leq i < 2^l$ & $(\underline{F}>\mathrm{bin}_l(i))$ 
      & $\neg (\underline{F}\leq\mathrm{bin}_l(i))$  \\ \hline
\end{tabular}
\end{center}
\end{table}

The models of the formula $f_\ssl(w)$ will contain certain ``information'' points that will
realize an accepting tree of $M$ on input $w$ if, and only if, $w\in L$.
Besides these information points there will also be other, ``auxiliary'', points (and an $\stackrel{L}{\to}$-equivalence class
not containing any information points) whose sole purpose is to make the mechanism of 
shared variables work. In several formulas we need to distinguish between the information points
and the other, auxiliary, points.
It turns out that this can be done simply by the truth value of the propositional variable $B$.

The following formula makes sure that in each of the vectors of shared variables that
describe in a unary way the current state respectively the written symbol respectively the current symbol, exactly one shared variable is true:
\begin{eqnarray*}
\mathit{uniqueness}&:=& B \rightarrow
 \bigl( \mathrm{unique}({\underline{\alpha}^\mathrm{state}}) 
 \wedge  \mathrm{unique}({\underline{\alpha}^\mathrm{written}}) 
 \wedge  \mathrm{unique}({\underline{\alpha}^\mathrm{read}}) \bigr).
\end{eqnarray*}
The vector $\underline{X}^\mathrm{read}$ will satisfy the same
uniqueness condition automatically.

The following formula ensures that the variables in the cloud corresponding to the first node in a computation tree
get the correct values. The computation starts at time $0$ with the tape head at position $2^N-1$ and
in the state $q_0$ and with the blank symbol $\#$ in the current cell.
\[ \mathit{start}
   :=  B \wedge (\underline{\alpha}^\mathrm{time}=\bin_N(0))
           \wedge (\underline{\alpha}^\mathrm{pos}=\bin_{N+1}(2^N-1))
           \wedge \alpha^\mathrm{state}_{q_0}
           \wedge \alpha^\mathrm{read}_{\#} .
\]

The following formula ensures that the vector $\underline{X}^{\text{time-apv}}$ stores
the time after the previous visit of the same cell, if it has been visited before.
If it has never been visited before, this vector gets the binary value $0$.
\begin{eqnarray*}
   \lefteqn{\mathit{time\_after\_previous\_visit}} && \\
   &:=& B \rightarrow \Biggl(
     \Bigl( \underline{X}^{\text{time-apv}} \leq \underline{X}^\mathrm{time} \Bigr) \\
   &&      \wedge \Bigl( ( \underline{\alpha}^\mathrm{time} < \underline{X}^\mathrm{time}
              \wedge \underline{\alpha}^\mathrm{pos} \neq \underline{X}^\mathrm{pos} )
              \rightarrow (\underline{X}^{\text{time-apv}} \neq \underline{\alpha}^\mathrm{time} + 1) \Bigr) \\
   &&      \wedge \Bigl( (  \underline{\alpha}^\mathrm{time} < \underline{X}^\mathrm{time}
              \wedge \underline{\alpha}^\mathrm{pos} = \underline{X}^\mathrm{pos}) 
              \rightarrow (\underline{\alpha}^\mathrm{time} < \underline{X}^{\text{time-apv}}) \Bigr) \Biggr) .
\end{eqnarray*} 
We explain this formula. The time $\underline{X}^{\text{time-apv}}$ after the previous visit of the current
cell $\underline{X}^\mathrm{pos}$ is certainly at most as large as the current time
$\underline{X}^\mathrm{time}$. When during the computation at an earlier time 
a cell has been visited that is different from the current one
then one plus the time of that visit is certainly not the time after the previous visit of the current cell.
When during the computation at an earlier time the current cell has been visited then 
the time of that visit is a strict lower bound for the time after the previous visit of the current
cell. Together these conditions ensure that $\underline{X}^{\text{time-apv}}$
gets the correct value.

The following formula ensures that the vector $\underline{X}^\mathrm{read}$ stores (in unary form)
the symbol in the current cell.

\begin{eqnarray*}
	\lefteqn{\mathit{get\_the\_right\_symbol}} && \\
	&:= & 
	\Biggl(\bigl(B \wedge (\underline{X}^{\text{time-apv}}=\mathrm{bin}_N(0))\bigr)
	\rightarrow \\
	&&\quad\Bigl( 
	\bigwedge_{i=1}^n ( (\underline{X}^\mathrm{pos}=\mathrm{bin}_{N+1}(2^{N}-1+i))
	\rightarrow X^\mathrm{read}_{w_i} ) \\
	&&\quad\;
 	\wedge ( (\underline{X}^\mathrm{pos} \leq \mathrm{bin}_{N+1}(2^{N}-1))
	\vee  (\underline{X}^\mathrm{pos} > \mathrm{bin}_{N+1}(2^{N}-1+n)))
	\rightarrow X^\mathrm{read}_{\#} \Bigr) \Biggr) \\
	&&\wedge \Biggl( \bigl(B \wedge (\underline{X}^{\text{time-apv}} > \mathrm{bin}_N(0))
	\wedge (\underline{\alpha}^\mathrm{time} = \underline{X}^{\text{time-apv}})\bigr)
	\rightarrow  (\underline{X}^\mathrm{read}=\underline{\alpha}^\mathrm{written})   \Biggr) . 
\end{eqnarray*}
We explain this formula. If the current cell has never before been visited
(this is the case iff the vector $\underline{X}^{\text{time-apv}}$ has the binary value $0$)
then the vector $\underline{X}^\mathrm{read}$ is forced to store in unary format the
initial symbol in the current cell. This is either a symbol $w_i$ of the input string or the blank $\#$.
If the current cell has been visited before
(this is the case iff the vector $\underline{X}^{\text{time-apv}}$ has a binary value strictly greater than $0$)
then in the cloud corresponding to the time stored in $\underline{X}^{\text{time-apv}}$
the vector $\underline{\alpha}^\mathrm{written}$ describes the symbol that has been written into the current
cell during the previous visit. Therefore, this value is copied into the vector $\underline{X}^\mathrm{read}$.

Next, we wish to define the formula $\mathit{computation}$ that describes the computation steps.
We have to distinguish between the two cases whether the tape head is going to move to the left or to the right.
If in a computation step the symbol $\theta\in\Gamma$ is written into the current cell,
if the tape head moves to the right, and if the new state after this step is the state $r \in Q$,
then the following formula guarantees the existence of a point
and its cloud with suitable values in the shared variables
and in the persistent propositional variables.
\begin{eqnarray*}
\lefteqn{\mathit{compstep}_{\mathrm{right}}(r,\theta)} && \\ 
&:=& \bigwedge_{k=0}^{N-1} \bigwedge_{l=0}^{N}
   \Biggl( \bigl( B \wedge \mathrm{rightmost\_zero}(\underline{\alpha}^\mathrm{time},k)
           \wedge \mathrm{rightmost\_zero}(\underline{\alpha}^\mathrm{pos},l) \bigr) \\
&&      \rightarrow L \biggl( B \wedge (\underline{X}^\mathrm{time}=\underline{\alpha}^\mathrm{time},>k)
                                            \wedge \mathrm{rightmost\_one}(\underline{X}^\mathrm{time},k) \\
&& \phantom{\rightarrow L \biggl(}  \wedge (\underline{X}^\mathrm{pos}=\underline{\alpha}^\mathrm{pos},>l)
                                            \wedge \mathrm{rightmost\_one}(\underline{X}^\mathrm{pos},l) \\
&& \phantom{\rightarrow L \biggl(}\wedge \Diamond \bigl(  
                        (\underline{\alpha}^\mathrm{time}=\underline{X}^\mathrm{time})
                        \wedge (\underline{\alpha}^\mathrm{pos}=\underline{X}^\mathrm{pos}) 
                        \wedge \alpha^\mathrm{state}_{r} \wedge \alpha^\mathrm{written}_\theta
                        \wedge (\underline{\alpha}^\mathrm{read}=\underline{X}^\mathrm{read})         \bigr) \bigg) \Biggr).
\end{eqnarray*}
We explain this formula.
The procedure is quite similar to the one of the formula $\mathrm{counter}_{SSL,n}$ in Subsection~\ref{subsection:counter}
for a binary counter.
The first three lines of the formula make sure that there is a point in the same cloud as the current point
such that in this new point the binary value of the persistent variable vector $\underline{X}^\mathrm{time}$
is larger by one than the binary value of the shared variable vector $\underline{\alpha}^\mathrm{time}$
and that in this new point the binary value of the persistent variable vector $\underline{X}^\mathrm{pos}$
is larger by one than the binary value of the shared variable vector $\underline{\alpha}^\mathrm{pos}$.
The last two lines ensure the existence of a $\stackrel{\Diamond}{\to}$-successor of this new point 
in which the shared variable vectors
$\underline{\alpha}^\mathrm{time}$,
$\underline{\alpha}^\mathrm{pos}$,
$\underline{\alpha}^\mathrm{state}$, 
$\underline{\alpha}^\mathrm{written}$, 
and $\underline{\alpha}^\mathrm{read}$
get the correct new values.

If in a computation step the symbol $\theta\in\Gamma$ is written into the current cell,
if the tape head moves to the left, and if the new state after this step is the state $r \in Q$,
then the following formula guarantees the existence of a point
and its cloud with suitable values in the shared variables
and in the persistent propositional variables.
\begin{eqnarray*}
\lefteqn{\mathit{compstep}_{\mathrm{left}}(r,\theta)} && \\ 
&:=& \bigwedge_{k=0}^{N-1} \bigwedge_{l=0}^{N}
   \Biggl( \bigl( B \wedge \mathrm{rightmost\_zero}(\underline{\alpha}^\mathrm{time},k)
           \wedge \mathrm{rightmost\_one}(\underline{\alpha}^\mathrm{pos},l) \bigr) \\
&&      \rightarrow L \biggl( B \wedge (\underline{X}^\mathrm{time}=\underline{\alpha}^\mathrm{time},>k)
                                            \wedge \mathrm{rightmost\_one}(\underline{X}^\mathrm{time},k) \\
&& \phantom{\rightarrow L \biggl(}  \wedge (\underline{X}^\mathrm{pos}=\underline{\alpha}^\mathrm{pos},>l)
                                            \wedge \mathrm{rightmost\_zero}(\underline{X}^\mathrm{pos},l) \\
&& \phantom{\rightarrow L \biggl(} \wedge \Diamond \bigl(  
                        (\underline{\alpha}^\mathrm{time}=\underline{X}^\mathrm{time})
                        \wedge (\underline{\alpha}^\mathrm{pos}=\underline{X}^\mathrm{pos}) 
                        \wedge \alpha^\mathrm{state}_{r} \wedge \alpha^\mathrm{written}_\theta
                        \wedge (\underline{\alpha}^\mathrm{read}=\underline{X}^\mathrm{read})         \bigr) \bigg) \Biggr).
\end{eqnarray*}
This formula is very similar to the previous one with the exception that here
the binary counter for the position of the tape head is decremented.

The computation is modeled by the following subformula.
Remember that $Q$ is the disjoint union of the sets
$\{q_{\mathrm{accept}}\}$, $\{q_{\mathrm{reject}}\}$, $Q_{\forall}$, $Q_{\exists}$.
\begin{eqnarray*}
\mathit{computation}
&:=& \bigwedge_{q \in Q_\forall} \bigwedge_{\eta \in\Gamma}
   \Biggl(   (\alpha^\mathrm{state}_q \wedge \alpha^\mathrm{read}_\eta) \rightarrow \\
&& \biggl( \bigwedge_{(r,\theta,\mathit{left}) \in \delta(q,\eta)}  \mathit{compstep}_{\mathrm{left}}(r,\theta)
   \wedge   \bigwedge_{(r,\theta,\mathit{right}) \in \delta(q,\eta)}  \mathit{compstep}_{\mathrm{right}}(r,\theta) \biggr) \Biggr) \\
&& \wedge \bigwedge_{q \in Q_\exists} \bigwedge_{\eta \in\Gamma}
   \Biggl(   (\alpha^\mathrm{state}_q \wedge \alpha^\mathrm{read}_\eta) \rightarrow \\
&& \biggl( \bigvee_{(r,\theta,\mathit{left}) \in \delta(q,\eta)}  \mathit{compstep}_{\mathrm{left}}(r,\theta)
   \vee   \bigvee_{(r,\theta,\mathit{right}) \in \delta(q,\eta)}  \mathit{compstep}_{\mathrm{right}}(r,\theta) \biggr) \Biggr) \Biggr) .
\end{eqnarray*}

Finally, the subformula $\mathit{no\_reject}$ is defined as follows.
\[ \mathit{no\_reject} := 
    \neg \alpha^\mathrm{state}_{q_\mathrm{reject}} . \]

We have completed the description of the formula $f_\ssl(w)$ for $w\in \Sigma^*$.
It is clear that $f_\ssl(w)$ is a bimodal formula, for any $w \in\Sigma^*$.
We still have to show two claims:
\begin{enumerate}
\item
The function $f_\ssl$ can be computed in logarithmic space.
\item
For any $w \in\Sigma^*$,
\[ w \in L \iff \text{ the bimodal formula $f_\ssl(w)$ is $\ssl$-satisfiable.} \]
\end{enumerate}

The first claim will be shown in the following section.
The two directions of the second claim will be shown separately in Subsections~\ref{subsection:leftright-ATMs-SSL}
and~\ref{subsection:rightleft-ATMs-SSL}.

\subsection{LOGSPACE Computability of the Reduction}
\label{subsection:reduction-ATMs-SSL-complexity}

For the first claim, we observe that there are three kinds of subformulas of $f_\ssl(w)$:
\begin{enumerate}
\item
subformulas that do not depend on the input string $w$ at all,
\item
subformulas that depend only on the length $n$ of the input string $w$ but not on its symbols $w_1,\ldots,w_n$,
\item
subformulas that depend on the particular symbols $w_1,\ldots,w_n$ of the input string $w$.
\end{enumerate}

The subformula $K\Box \mathit{uniqueness}$ is of the first type.
Therefore, it can be written using only a constant amount of workspace.
And there is only one subformula of the third type, the subformula
$K\Box \mathit{get\_the\_right\_symbol}$. All other subformulas are of the second type.
All of them contain vectors of propositional variables of length at most $N+1$
or conjunctions or disjunctions of length at most $N+1$, where $N=p(n)$.
And all of these vectors and lists of conjunctions or disjunctions have a very regular structure.
This applies also to the only subformula of the third type.
This regular structure makes it possible to write down these formulas using a fixed
(that means: independent of the input string $w$) number of counters that can count up to $N$.
But such counters can be implemented in binary using not more than $O(\log N) = O(\log n)$ space.
Hence, given a string $w$,
the whole formula $f_\ssl(w)$ can be computed using not more than logarithmic space.

\subsection{Construction of a Model}
\label{subsection:leftright-ATMs-SSL}

We come to the second claim. First, we show the direction from left to right.
Let us assume $w \in L$. 
We will construct a cross axiom model 
$(W,\stackrel{L}{\to},\stackrel{\Diamond}{\to},\sigma)$ with a point
$p_{\mathit{root},\mathit{root}} \in W$ such that $p_{\mathit{root},\mathit{root}} \models f_\ssl(w)$.
There exists an accepting tree $T=(V,E,c)$ of $M$ on input $w$,
where $V$ is the set of nodes of $T$, where $E \subseteq V\times V$ is the set of edges,
and where the function $c:V\to Q \times \{0,\ldots,2^{N+1}-2\}\times \Gamma^{2^{N+1}-1}$
labels each node with a configuration
(remember the discussion about the description of configurations at the beginning of Subsection~\ref{subsection:definition-ATMs-SSL}).
Let $\mathit{root} \in V$ be the root of $T$.
The set $W$ is defined to be the (disjoint) union of the following three sets $P$, $U$, and $S$. We define
\[ P:=\{p_{v,x} ~:~ v,x \in V \text{ and } v E^* x \} .\]
For the definition of $U$ we use the following set as an index set:
\begin{eqnarray*}
  I &:=& (\{\mathrm{time}\} \times \{0,\ldots,N-1\}) \\
  && \cup (\{\mathrm{pos}\} \times \{0,\ldots,N\}) \\
  && \cup (\{\mathrm{state}\} \times Q) \\
  && \cup (\{\mathrm{written}\} \times \Gamma) \\
  && \cup (\{\mathrm{read}\} \times \Gamma) . 
\end{eqnarray*}
We define
\[ U := \{u_{v,\mathit{string},z} ~:~ v\in V \cup\{\top\}, (\mathit{string},z) \in I\}  \]
where $\top$ is a special element not contained in $V$.
We extend the binary relation $E^*$ on $V$ to a binary relation
$\widetilde{E}$ on $V\cup\{\top\}$ by 
\[ \widetilde{E} := \{(u,v) \in (V\cup\{\top\}) \times (V\cup\{\top\}) ~:~
          \text{ either } (u,v \in V \text{ and } uE^*v) \text{ or } v=\top \}. \]
We define the set $S$ by
\begin{eqnarray*}
 S &:=& \{ s_{v,\mathrm{time},k} ~:~ v \in V, k \in \mathrm{Ones}(\mathit{time}(v)) \} \\
   && \cup \{ s_{v,\mathrm{pos},k} ~:~ v \in V, k \in \mathrm{Ones}(\mathit{pos}(v)) \} \\
   && \cup \{ s_{v,\mathrm{state},q} ~:~ v \in V, q = \mathit{state}(v) \} \\
   && \cup \{ s_{v,\mathrm{written},\gamma} ~:~ v \in V \setminus\{\mathit{root}\}, \gamma = \mathit{written}(v) \}
        \cup \{ s_{\mathrm{root},\mathit{written},\#}\} \\
   && \cup \{ s_{v,\mathrm{read},\gamma} ~:~ v \in V, \gamma = \mathit{read}(v) \} .
\end{eqnarray*}
As the relation $\stackrel{L}{\to}$ is supposed to be an equivalence relation
we can define it by defining 
the $\stackrel{L}{\to}$-equivalence classes. These are the sets
\[
\mathit{Cloud}_v := \{p_{v,x} \in P ~:~ x \in V 
\}
\cup \{u_{v,i} \in U ~:~ i \in I \}  
\cup \{s_{v,i} \in S ~:~ i \in I \} 
\]
for all $v \in V$, and the set
\[ \mathit{Cloud}_{\top} := \{u_{\top,i} \in U ~:~ i \in I \} .  \]
We define the relation $\stackrel{\Diamond}{\to}$ by:
\begin{eqnarray*}
 \stackrel{\Diamond}{\to} &:=& \{ (p_{v,x}, p_{v',x'}) \in P\times P
     ~:~ v,v',x,x' \in V \text{ and } v  E^* v' \text{ and }  x=x'\} \\
  &&\cup \{ (u_{v,i}, u_{v',i'}) \in U\times U
   ~:~ v,v' \in V\cup\{\top\}, i,i' \in I \text{ and } v \widetilde{E} v' \text{ and } i=i' \} \\
  &&\cup \{ (u_{v,i}, s_{v',i'}) \in U \times S
     ~:~ v,v' \in V, i,i' \in I \text{ and } v E^* v' \text{ and } i=i'\} \\
  &&\cup \{ (s_{v,i}, s_{v',i'}) \in S \times S
  ~:~ v,v' \in V, i,i' \in I \text{ and } v=v' \text{ and } i=i'\}   .
\end{eqnarray*}
It is straightforward to check that $\stackrel{\Diamond}{\to}$ is reflexive and transitive.
The cross property is satisfied as well.
Thus, $(W,\stackrel{L}{\to},\stackrel{\Diamond}{\to})$ is an cross axiom frame.
Finally, we define the valuation $\sigma$ as follows.
\begin{eqnarray*}
\sigma(B) &:=& P, 
\end{eqnarray*}
and
\begin{eqnarray*}
\sigma(A^\mathrm{time}_k) &:=& \{u_{v,\mathrm{time},k} \in U ~:~ v\in V \cup\{\top\}\}
                                              \cup \{s_{v,\mathrm{time},k} \in S ~:~ v \in V \}, \\
\sigma(X^\mathrm{time}_k) &:=& \{p_{v,x} \in P ~:~ v,x \in V \text{ and } k \in \mathrm{Ones}(\mathit{time}(x))\}, \\
\sigma(X^{\text{time-apv}}_k) &:=& \{p_{v,x} \in P ~:~ v,x \in V \text{ and } k \in \mathrm{Ones}(j) \} \\
&& \!\!\!\!\!\!\!\!\!\!\text{where } j := \begin{cases}
  0 & \text{if on the path from $\mathit{root}$ to $x$ the cell $\mathit{pos}(x)$ has}\\
     & \text{not been visited before the cell $x$ is reached}, \\
  1 + \mathit{time}(v') & \text{otherwise, where } v' \text{ is the last node on the path} \\
    & \text{from }\mathit{root} \text{ to } \mathit{pred}(x) \text{ with } \mathit{pos}(v')=\mathit{pos}(x) , 
  \end{cases}
\end{eqnarray*}
for $k \in \{0,\ldots,N-1\}$,
\begin{eqnarray*}
\sigma(A^\mathrm{pos}_k) &:=& \{u_{v,\mathrm{pos},k} \in U~:~ v\in V \cup\{\top\}\}
                                              \cup \{s_{v,\mathrm{pos},k} \in S ~:~ v \in V \}, \\
\sigma(X^\mathrm{pos}_k) &:=& \{p_{v,x} \in P ~:~ v,x \in V \text{ and } k \in \mathrm{Ones}(\mathit{pos}(x))\}, 
\end{eqnarray*}
for $k \in \{0,\ldots,N\}$,
\begin{eqnarray*}
\sigma(A^\mathrm{state}_q) &:=& \{u_{v,\mathrm{state},q} \in U ~:~ v\in V \cup\{\top\}\}
                                              \cup \{s_{v,\mathrm{state},q} \in S ~:~ v \in V \}, 
\end{eqnarray*}
for $q \in Q$,
\begin{eqnarray*}
\sigma(A^\mathrm{written}_\gamma) &:=& \{u_{v,\mathrm{written},\gamma} \in U~:~ v\in V \cup\{\top\}\}
                                              \cup \{s_{v,\mathrm{written},\gamma} \in S ~:~ v \in V \}, \\
\sigma(A^\mathrm{read}_\gamma) &:=& \{u_{v,\mathrm{read},\gamma} \in U ~:~ v\in V \cup\{\top\}\}
                                              \cup \{s_{v,\mathrm{read},\gamma} \in S ~:~ v \in V \}, \\
\sigma(X^\mathrm{read}_\gamma) &:=& \{p_{v,x} \in P ~:~ v,x \in V \text{ and } \gamma = \mathit{read}(x))\}, 
\end{eqnarray*}
for $\gamma \in \Gamma$.
It is obvious that all propositional variables are persistent.
Thus, we have defined a cross axiom model
$(W,\stackrel{L}{\to},\stackrel{\Diamond}{\to},\sigma)$.
We claim that $p_{\mathit{root},\mathit{root}} \models f_\ssl(w)$.
For an illustration of an important detail of the structure see Figure~\ref{figure:model-ssl}.
\begin{figure}
	\begin{center}
	\includegraphics[width=0.8\linewidth]{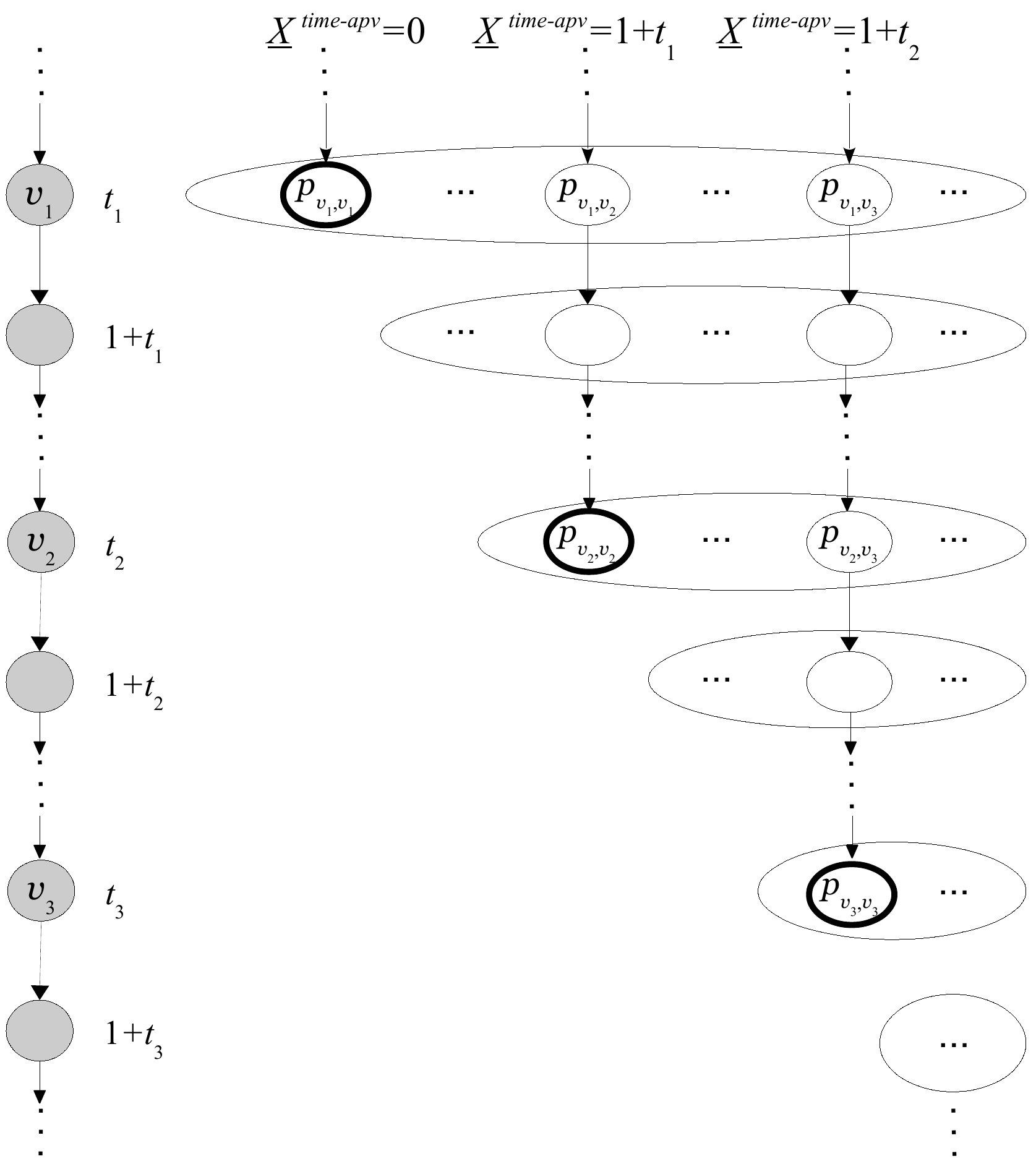} 
	\end{center}
	\caption{A possible detail of a cross axiom model of the formula $f_\ssl(w)$.
	Consider a certain cell and let us assume that $v_1, v_2, v_3$ are the first three computation nodes
	on some computation path in which this cell is visited. Let
	$t_i:=\mathit{time}(v_i)$. The diagram on the left shows a part of the computation path.
	The diagram on the right shows the corresponding part of the cross axiom model.}
	\label{figure:model-ssl}
\end{figure}

We start with some preliminary observations.
First, for any cloud, any shared variable has the same truth values in all points in the cloud.
Secondly, 
\[ y \models B \iff y \in P , \]
for all $y\in W$.
So, the points in $P$ are the ``information'' points.
On the other hand, as the cloud $\mathit{Cloud}_\top$ does not contain any elements from $P$,
for all points $y\in \mathit{Cloud}_\top$ we have $y \models K \neg B$,
hence,
\[(\forall y \in \mathit{Cloud}_\top) \  y \models \neg \alpha^\mathit{string}_k  ,\]
for all $\mathit{string}\in\{\mathrm{times},\mathrm{pos},\mathrm{state},\mathrm{written},\mathrm{read}\}$
and all $k$.
That means, the truth value of any shared variable in the cloud $\mathit{Cloud}_\top$ is false.
We claim that in the other clouds all shared variables have the values indicated by their names, namely,
\begin{eqnarray*}
  y &\models& (\underline{\alpha}^\mathrm{time} = \mathrm{bin}_N(\mathit{time}(v)), \\
  y &\models& (\underline{\alpha}^\mathrm{pos} = \mathrm{bin}_{N+1}(\mathit{pos}(v)), \\
  (y &\models& \alpha^\mathrm{state}_q) \iff q = \mathit{state}(v), \text{ for } q \in Q, \\
  (y &\models& \alpha^\mathrm{read}_\gamma) \iff \gamma = \mathit{read}(v), \text{ for } \gamma \in \Gamma, 
\end{eqnarray*}
for $v\in V$ and $y\in \mathit{Cloud}_v$,
\[
  (y \models \alpha^\mathrm{written}_\gamma) \iff \gamma = \mathit{written}(v), 
\]
for $\gamma \in \Gamma$, $v\in V\setminus\{\mathit{root}\}$ and $y\in \mathit{Cloud}_v$, and
\[
  (y \models \alpha^\mathrm{written}_\gamma) \iff \gamma = \#, 
\]
for $\gamma \in \Gamma$ and $y\in \mathit{Cloud}_\mathit{root}$.
This can be checked similarly as the corresponding claim~\eqref{eq:SSL-counter-1} in the proof of 
Proposition~\ref{prop: counter}.
We prove the assertions about $\alpha^\mathrm{written}_\gamma$ and leave the proofs
of the other assertions to the reader.
Let us fix some $\gamma \in\Gamma$.
Note that, for all $p' \in P$, we have $p' \models \neg A^\mathrm{written}_\gamma$,  hence
\[ (\forall p' \in P) \ p' \models (\neg A^\mathrm{written}_\gamma \vee \Diamond K \neg B) . \]
Furthermore, $u_{\top,i} \models K \neg B$ for all $i \in I$.
Since for all $v \in V\cup\{\top\}$ and $i \in I$
we have $u_{v,i} \stackrel{\Diamond}{\to} u_{\top,i}$ we obtain
$u_{v,i} \models \Diamond K \neg B$
for all $v \in V\cup\{\top\}$ and $i \in I$.
Hence,
\[ (\forall u' \in U) \ u' \models (\neg A^\mathrm{written}_\gamma \vee \Diamond K \neg B) . \]
This shows that, for any $v \in V$,
the shared variable
$\alpha^\mathrm{written}_\gamma=L(A^\mathrm{written}_\gamma \wedge \Box L B)$ 
is true in the cloud $\mathit{Cloud}_v$ 
if, and only if, there exists some $s' \in \mathit{Cloud}_v \cap S$
with $s' \models A^\mathrm{written}_\gamma \wedge \Box L B$.
As $p' \models B$ for all $p' \in P$ and any
$s' \in S$ is $\stackrel{L}{\to}$-equivalent to some $p' \in P$, we have
$s' \models L B$, for all $s' \in S$.
Actually, for $s' \in S$ we even have $s' \models \Box L B$ as $s'$ does not have
any $\stackrel{\Diamond}{\to}$-successor besides itself.
Thus, for any $v \in V$,
the shared variable
$\alpha^\mathrm{written}_\gamma=L(A^\mathrm{written}_\gamma \wedge \Box L B)$ 
is true in the cloud $\mathit{Cloud}_v$ if, and only if, there exists some $s' \in \mathit{Cloud}_v \cap S$
with $s' \models A^\mathrm{written}_\gamma$.
The elements $s' \in \mathit{Cloud}_v \cap S$ have the form $s' = s_{v,i}$ for some $i \in I$.
On the one hand, we observe that, for $v \in V$ and $i \in I$,
$s_{v,i} \models A^\mathrm{written}_\gamma \iff i = (\mathrm{written},\gamma)$.
On the other hand, for $v \in V$ we have
\[ s_{v,\mathrm{written},\gamma} \in \mathit{Cloud}_v \iff 
  \begin{cases} 
     \gamma = \mathit{written}(v) & \text{if } v \in V\setminus \{\mathit{root}\}, \\
     \gamma = \# & \text{if } v = \mathit{root} . 
   \end{cases}
\]
Thus, we have shown the desired claims:
\[
  (y \models \alpha^\mathrm{written}_\gamma) \iff \gamma = \mathit{written}(v), 
\]
for $\gamma \in \Gamma$, $v\in V\setminus\{\mathit{root}\}$ and $y\in \mathit{Cloud}_v$, and
\[
  (y \models \alpha^\mathrm{written}_\gamma) \iff \gamma = \#, 
\]
for $\gamma \in \Gamma$ and $y\in \mathit{Cloud}_\mathit{root}$.

It is clear from the definition of the valuation $\sigma$ that in the points in $P$ the variable vectors 
$\underline{X}^\mathrm{time}$,
$\underline{X}^\mathrm{pos}$,
$\underline{X}^\mathrm{read}$, and
$\underline{X}^{\text{time-apv}}$
have the values indicated by their names:
\begin{eqnarray*}
  p_{v,x} &\models& (\underline{X}^\mathrm{time} = \mathrm{bin}_N(\mathit{time}(x)), \\
  p_{v,x} &\models& (\underline{X}^\mathrm{pos} = \mathrm{bin}_{N+1}(\mathit{pos}(x)), \\
  (p_{v,x} &\models& X^\mathrm{read}_\gamma) \iff \gamma = \mathit{read}(x), \text{ for } \gamma \in \Gamma, \\
  p_{v,x} &\models& (\underline{X}^{\text{time-apv}} = \mathrm{bin}_N(j)), \\
&&\text{where } j := \begin{cases}
  0 & \text{if on the path from $\mathit{root}$ to $x$ the cell $\mathit{pos}(x)$ has} \\
     & \text{not been visited before the cell $x$ is reached},\\
  1 + \mathit{time}(v') & \text{otherwise, where } v' \text{ is the last node on the path} \\
    & \text{from }\mathit{root} \text{ to }\mathit{pred}(x) \text{ with} \mathit{pos}(v')=\mathit{pos}(x) , 
  \end{cases}
\end{eqnarray*}
for $v,x\in V$ satisfying $v E^* x$.

Now we are prepared to show $p_{\mathit{root},\mathit{root}} \models f_\ssl(w)$.
Our observations about the values of the shared variable vectors 
$\underline{\alpha}^\mathrm{state}$,
$\underline{\alpha}^\mathrm{read}$, and
$\underline{\alpha}^\mathrm{written}$ show
\[ p_{\mathit{root},\mathit{root}} \models K \Box \mathit{uniqueness}  \]
(remember that $B$ is false in all points of the cloud $\mathit{Cloud}_\top$).
Similarly, our observations about the values of the shared variable vectors
$\underline{\alpha}^\mathrm{time}$,
$\underline{\alpha}^\mathrm{pos}$,
$\underline{\alpha}^\mathrm{state}$, and
$\underline{\alpha}^\mathrm{read}$
imply
\[ p_{\mathit{root},\mathit{root}} \models \mathit{start} . \]
As the state of any node in the accepting tree $T$ of $M$ on input $w$ is different from $q_\mathrm{reject}$,
our observations about the value of the vector 
$\underline{\alpha}^\mathrm{state}$ 
(in any cloud $\mathit{Cloud}_v$ for $v\in V$ the shared variable $\alpha^\mathrm{state}_q$ is true
if, and only if, $q = \mathit{state}(v)$, and in the cloud $\mathit{Cloud}_\top$ all shared variables are false)
shows 
\[ p_{\mathit{root},\mathit{root}} \models K \Box \mathit{no\_reject} . \]

Next, we show 
$$p_{\mathit{root},\mathit{root}} \models K \Box \mathit{time\_after\_previous\_visit}.$$
As the variable $B$ is true only in the elements of $P$,
it is sufficient to show for all $v,x \in V$ with $v E^* x$
\begin{eqnarray*}
p_{v,x} \models
   && \Bigl( \underline{X}^{\text{time-apv}} \leq \underline{X}^\mathrm{time} \Bigr) \\
   &&      \wedge \Bigl( 
   			\bigl(( \underline{\alpha}^\mathrm{time} < \underline{X}^\mathrm{time})
              \wedge (\underline{\alpha}^\mathrm{pos} \neq \underline{X}^\mathrm{pos} )\bigr)
              \rightarrow (\underline{X}^{\text{time-apv}} \neq \underline{\alpha}^\mathrm{time} + 1 )\Bigr) \\
   &&      \wedge \Bigl( \bigl( 
   			( \underline{\alpha}^\mathrm{time} < \underline{X}^\mathrm{time})
              \wedge (\underline{\alpha}^\mathrm{pos} = \underline{X}^\mathrm{pos}) \bigr)
              \rightarrow (\underline{\alpha}^\mathrm{time} < \underline{X}^{\text{time-apv}}) \Bigr).
\end{eqnarray*}
We distinguish between the two cases whether the cell $\mathit{pos}(x)$ under the tape head in the configuration $c(x)$
has been visited on the path from $\mathit{root}$ to $x$ before $x$ is reached or not.

First let us assume that the cell $\mathit{pos}(x)$ has not been visited before.
Then $p_{v,x} \models (\underline{X}^{\text{time-apv}} = \mathrm{bin}_N(0))$, hence
$p_{v,x} \models (\underline{X}^{\text{time-apv}} \leq \underline{X}^\mathrm{time})$
and 
$p_{v,x} \models (\underline{X}^{\text{time-apv}} \neq \underline{\alpha}^\mathrm{time} + 1)$, that is,
the formulas in the first two lines are true in $p_{v,x}$.
And if we have $p_{v,x} \models (\underline{\alpha}^\mathrm{time} < \underline{X}^\mathrm{time})$
then, due to $v E^* x$, the point $v$ is a point on the path from $\mathit{root}$ to $\mathit{pred}(x)$.
Our assumption (that the cell $\mathit{pos}(x)$ has not been visited before $x$ is reached) 
implies $p_{v,x} \models \neg (\underline{\alpha}^\mathrm{pos} = \underline{X}^\mathrm{pos})$. Hence
the formula in the third line is true in $p_{v,x}$ as well. \\
Now, let us assume that the cell $\mathit{pos}(x)$ has been visited on the path from $\mathit{root}$
to $x$ before $x$ is reached. Let $x'$ be the last node on the path from $\mathit{root}$ 
to $\mathit{pred}(x)$ with $\mathit{pos}(x')=\mathit{pos}(x)$.
Let $i:= \mathit{time}(x')$ and $j:=1+i$. Then $i < \mathit{time}(x)$ and $j \leq \mathit{time}(x)$.
As $p_{v,x} \models (\underline{X}^{\text{time-apv}} = \mathrm{bin}_N(j))$, the formula in the first line
above is true, that is,
$p_{v,x} \models (\underline{X}^{\text{time-apv}} \leq \underline{X}^\mathrm{time})$.
For the formula in the second line, let us assume
$p_{v,x} \models  (\underline{\alpha}^\mathrm{time} < \underline{X}^\mathrm{time})
              \wedge (\underline{\alpha}^\mathrm{pos} \neq \underline{X}^\mathrm{pos}) $.
Then $v$ is a node on the path from $\mathit{root}$ to $\mathit{pred}(x)$ with $\mathit{pos}(v) \neq \mathit{pos}(x)$.
Hence, $v \neq x'$, hence, $\mathit{time}(v) \neq \mathit{time}(x')$, hence, $j \neq \mathit{time}(v) + 1$, hence,
$p_{v,x} \models (\underline{X}^{\text{time-apv}} \neq \underline{\alpha}^\mathrm{time} + 1)$.
For the formula in the third line, let us assume
$p_{v,x} \models ( \underline{\alpha}^\mathrm{time} < \underline{X}^\mathrm{time})
              \wedge (\underline{\alpha}^\mathrm{pos} = \underline{X}^\mathrm{pos} )$.
Then $v$ is a node on the path from $\mathit{root}$ to $\mathit{pred}(x)$ with $\mathit{pos}(v) = \mathit{pos}(x)$,
that is, with the property that in this node the same cell is visited as in the node $x$.
As $x'$ is the last node on the path from $\mathit{root}$ to $\mathit{pred}(x)$ with this property,
we have $\mathit{time}(v) \leq \mathit{time}(x')$, hence, $j=1+i=1+\mathit{time}(x') > \mathit{time}(v)$, hence
$p_{v,x} \models (\underline{\alpha}^\mathrm{time} < \underline{X}^{\text{time-apv}})$.

Next, we show 
$$p_{\mathit{root},\mathit{root}} \models K \Box \mathit{get\_the\_right\_symbol}.$$
It is sufficient to show for all $y\in W$, 
$y \models \mathit{get\_the\_right\_symbol}$. There are two cases to be considered.
In each of them, due to the presence of the variable $B$, we need to consider only
points $y\in P$. Let us consider elements $v,x\in V$ with $v E^* x$. It is sufficient to show
$p_{v,x} \models \mathit{get\_the\_right\_symbol}$.
We distinguish between the two cases considered in this formula.
First, let us assume
$p_{v,x} \models (\underline{X}^{\text{time-apv}} = \mathrm{bin}_N(0))$.
We have to show that in this case
\begin{eqnarray*}
 p_{v,x} &\models& 
	\bigwedge_{i=1}^n \Bigl( (\underline{X}^\mathrm{pos}=\mathrm{bin}_{N+1}(2^{N}-1+i))
            \rightarrow X^\mathrm{read}_{w_i} \Bigr) \\
   && 
      \wedge \Bigl(\bigl( (\underline{X}^\mathrm{pos} \leq \mathrm{bin}_{N+1}(2^{N}-1))
                \vee  (\underline{X}^\mathrm{pos} > \mathrm{bin}_{N+1}(2^{N}-1+n))\bigr)
                   \rightarrow X^\mathrm{read}_{\#} \Bigr) .
\end{eqnarray*}
According to our observations about the value of
$\underline{X}^{\text{time-apv}}$, the cell $\mathit{pos}(x)$
under the tape head in the configuration $c(x)$
has not been visited on the path from $\mathit{root}$ to $\mathit{pred}(x)$.
Thus, the symbol $\mathit{read}(x)$ is still the initial symbol in the cell $\mathit{pos}(x)$.
Let us call this symbol $\gamma$. 
Then $p_{v,x} \models X^\mathrm{read}_\gamma$,
and 
\[ \gamma= \begin{cases}
  w_i & \text{if }\mathit{pos}(x)=2^N-1+i, \text{ for some }i \in \{1,\ldots,n\}, \\
  \# & \text{if } \mathit{pos}(x)\leq 2^N-1 \text{ or } \mathit{pos}(x) > 2^N-1+n.
  \end{cases}
\]
On the other hand,
$p_{v,x} \models ( \underline{X}^\mathrm{pos} = \mathrm{bin}_{N+1}(\mathit{pos}(x)))$.
We have shown the assertion.\\
Now, let us assume 
\[ p_{v,x} \models 
((\underline{X}^{\text{time-apv}} > \mathrm{bin}_N(0))
        \wedge (\underline{\alpha}^\mathrm{time} = \underline{X}^{\text{time-apv}})) . \]
We have to show that in this case
\[ p_{v,x} \models (\underline{X}^\mathrm{read}=\underline{\alpha}^\mathrm{written}) . \]
The assumption
$p_{v,x} \models (\underline{X}^{\text{time-apv}} > \mathrm{bin}_N(0))$
implies that the cell $\mathit{pos}(x)$ has already been visited on the path from
$\mathit{root}$ to $\mathit{pred}(x)$ and that
$p_{v,x} \models (\underline{X}^{\text{time-apv}} = \mathrm{bin}_{N+1}(1+i))$ where
$i=\mathit{time}(x')$ and $x'$ is the last node on the path from $\mathit{root}$ to
$\mathit{pred}(x)$ with $\mathit{pos}(x') = \mathit{pos}(x)$.
The assumption 
$p_{v,x} \models (\underline{\alpha}^\mathrm{time} = \underline{X}^{\text{time-apv}})$
implies $x' = \mathit{pred}(v)$. But then in the point $v$ the vector
$\underline{\alpha}^\mathrm{written}$ encodes in unary the symbol that was written into
the cell $\mathit{pos}(x)$ in the computation step from $x'$ to $v$.
If we call this symbol $\gamma$, this means $p_{v,x} \models \alpha^\mathrm{written}_\gamma$.
This is still the symbol in the cell $\mathit{pos}(x)$ when $x$ is reached, hence,
$p_{v,x} \models X^\mathrm{read}_\gamma$.
So, we have indeed shown
\[ p_{v,x} \models (\underline{X}^\mathrm{read}=\underline{\alpha}^\mathrm{written}) . \]

Finally, we show
$$p_{\mathit{root},\mathit{root}} \models K \Box \mathit{computation}.$$
It is sufficient to show
\[ y \models \mathit{computation} , \]
for all $y \in W$.
We will separately treat the conjunctions over the set $(q,\eta) \in Q_\exists \times \Gamma$
and over the set $Q_\forall \times \Gamma$.
Let us fix a pair $(q,\eta) \in Q_\exists \times \Gamma$
and let us assume that $y\in W$ is a point with
$y \models (\alpha^\mathrm{state}_q \wedge \alpha^\mathrm{read}_\eta)$.
We have to show that there is an element
$(r,\theta,\mathit{left}) \in \delta(q,\eta)$ such that
$y \models \mathit{compstep}_{\mathrm{left}}(r,\theta)$
or that there is an element
$(r,\theta,\mathit{right}) \in \delta(q,\eta)$ such that
$y \models \mathit{compstep}_{\mathrm{right}}(r,\theta)$.
As in the cloud $\mathit{Cloud}_\top$ the truth value of any shared variable is false,
the assumption
$y \models (\alpha^\mathrm{state}_q \wedge \alpha^\mathrm{read}_\eta)$
implies that there exists some $v\in V$ with $y \in \mathit{Cloud}_v$.
Furthermore, $q=\mathit{state}(v)$ and $\eta = \mathit{read}(v)$.
As $T$ is an accepting tree and the state $q$ of $c(v)$ is an element of $Q_\exists$,
the node $v$ is an inner node of $T$, hence, it has a successor $v'$.
Let us assume that $((q,\eta),(r,\theta,\mathit{left})) \in \delta$
is the element of the transition relation $\delta$ that leads from $v$ to $v'$
(the case that this element is of the form
$((q,\eta),(r,\theta,\mathit{right}))$ is treated analogously).
We claim that then
\[ y \models \mathit{compstep}_{\mathrm{left}}(r,\theta) . \]
Let us check this.
Let us assume that, for some $k\in\{0,\ldots,N-1\}$
and for some $l \in \{0,\ldots,N\}$, 
\[ y \models \bigl( B \wedge \mathrm{rightmost\_zero}(\underline{\alpha}^\mathrm{time},k)
           \wedge \mathrm{rightmost\_one}(\underline{\alpha}^\mathrm{pos},l) \bigr).
\]
The number $i:=\mathit{time}(v)$ is an element of $\{0,\ldots,2^N-2\}$
because $v$ is an inner point of the tree $T$ and the length of any computation path is at most $2^N-1$,
and the number $j:= \mathit{pos}(v)$ is an element of $\{1,\ldots,2^{N+1}-3\}$
because the computation starts in cell $2^N-1$ and
because during each computation step the tape head can move at most one step to the left or to the right.
We obtain $k=\min(\{0,\ldots,N-1\} \setminus\mathrm{Ones}(i))$
and $l=\min\mathrm{Ones}(j)$.
We claim that the two points
$p_{v,v'}$ and $p_{v',v'}$ have the properties formulated in the formula
$\mathit{compstep}_{\mathrm{left}}(r,\theta)$.
Indeed, we observe $y \stackrel{L}{\to}p_{v,v'}$ and
$p_{v,v'} \stackrel{\Diamond}{\to} p_{v',v'}$
as well as $p_{v,v'} \models B$.
The facts 
\begin{align*}
& \mathit{time}(v)=i, && \mathit{pos}(v)=j, \\
& \mathit{time}(v')=i+1, && \mathit{pos}(v')=j-1,
\end{align*}
imply
\begin{eqnarray*}
p_{v,v'} &\models& (\underline{\alpha}^\mathrm{time}=\mathrm{bin}_N(i)) 
              \wedge (\underline{\alpha}^\mathrm{pos}=\mathrm{bin}_{N+1}(j)) \wedge \\
             &&  (\underline{X}^\mathrm{time}=\mathrm{bin}_N(i+1)) 
              \wedge (\underline{X}^\mathrm{pos}=\mathrm{bin}_{N+1}(j-1)) 
\quad \text{ and }  \\
p_{v',v'} &\models& (\underline{\alpha}^\mathrm{time}=\mathrm{bin}_N(i+1)) 
              \wedge (\underline{\alpha}^\mathrm{pos}=\mathrm{bin}_{N+1}(j-1)) \wedge \\
            &&  (\underline{X}^\mathrm{time}=\mathrm{bin}_N(i+1)) 
              \wedge (\underline{X}^\mathrm{pos}=\mathrm{bin}_{N+1}(j-1)) .
\end{eqnarray*}
We obtain
\[ p_{v',v'} \models (\underline{\alpha}^\mathrm{time}=\underline{X}^\mathrm{time})
                        \wedge (\underline{\alpha}^\mathrm{pos}=\underline{X}^\mathrm{pos}) , \]
and with
\[ k=\min(\{0,\ldots,N-1\} \setminus\mathrm{Ones}(i))
 \text{ and }  l=\min(\mathrm{Ones}(j))  \]
we obtain as well
\begin{eqnarray*} p_{v,v'} &\models&  (\underline{X}^\mathrm{time}=\underline{\alpha}^\mathrm{time},>k)
                       \wedge \mathrm{rightmost\_one}(\underline{X}^\mathrm{time},k) \\
                 && \wedge (\underline{X}^\mathrm{pos}=\underline{\alpha}^\mathrm{pos},>l)
                      \wedge \mathrm{rightmost\_zero}(\underline{X}^\mathrm{pos},l) .
\end{eqnarray*}
Finally, our observations about the values of the shared variable vectors
$\underline{\alpha}^\mathrm{state}$,
$\underline{\alpha}^\mathrm{written}$,
$\underline{\alpha}^\mathrm{read}$
and about the vector
$\underline{X}^\mathrm{read}$
imply that also
\[ p_{v',v'} \models \alpha^\mathrm{state}_{r} \wedge \alpha^\mathrm{written}_\theta
                        \wedge (\underline{\alpha}^\mathrm{read}=\underline{X}^\mathrm{read})    . \]
(remember that $((q,\eta),(r,\theta,\mathit{left})) \in \delta$
is the element of the transition relation $\delta$ that leads from the node $v$ to the node $v'$).
This ends the treatment of the conjunctions over the set $(q,\eta) \in Q_\exists \times \Gamma$
in the formula $\mathit{computation}$.
Let us now consider a pair $(q,\eta) \in Q_\forall \times \Gamma$.
Let us assume that $y \in W$ is a point such that
$y \models (\alpha^\mathrm{state}_q \wedge \alpha^\mathrm{read}_\eta)$.
We have to show that for all elements
$(r,\theta,\mathit{left}) \in \delta(q,\eta)$ we have
$y \models \mathit{compstep}_{\mathrm{left}}(r,\theta)$
and for all elements
$(r,\theta,\mathit{right}) \in \delta(q,\eta)$ we have
$y \models \mathit{compstep}_{\mathrm{right}}(r,\theta)$.
Let us consider an arbitrary element $(r,\theta,\mathit{left}) \in \delta(q,\eta)$
(the case of an element $(r,\theta,\mathit{right}) \in \delta(q,\eta)$
is treated analogously).
As in the cloud $\mathit{Cloud}_\top$ the truth value of any shared variable is false,
the assumption
$y \models (\alpha^\mathrm{state}_q \wedge \alpha^\mathrm{read}_\eta)$
implies that there exists some $v\in V$ with $y \in \mathit{Cloud}_v$.
Furthermore, $q=\mathit{state}(v)$ and $\eta = \mathit{read}(v)$.
As $q \in Q_\forall$ and $T$ is an accepting tree,
in $T$ there is a successor $v'$ of $v$ such that
the element $((q,\eta),(r,\theta,\mathit{left}))$ leads from $v$ to $v'$.
Above, we have already seen that this implies
\[ y \models \mathit{compstep}_{\mathrm{left}}(r,\theta) . \]
Thus, we have shown $y \models \mathit{computation}$
for all $y \in W$.
This ends the proof of the claim $p_{\mathit{root},\mathit{root}} \models f_\ssl(w)$.

\subsection{Existence of an Accepting Tree}
\label{subsection:rightleft-ATMs-SSL}

We come to the other direction. Let $w \in \Sigma^*$.
We wish to show that if $f_\ssl(w)$
is $\ssl$-satisfiable then $w\in L$.
We will show that any cross axiom model of $f_\ssl(w)$
essentially contains an accepting tree of the Alternating Turing Machine $M$
on input $w$. Of course, this implies $w\in L$.

Let us sketch the main idea.
We will consider a cross axiom model of $f_\ssl(w)$.
And we will consider partial trees of $M$ on input $w$
as considered in Subsection~\ref{subsection:ATM}, for any $w\in\Sigma^*$.
First, we will show that a certain very simple partial tree of $M$ on input $w$
``can be mapped to'' the model (later we will give a precise
meaning to ``can be mapped to'').
Then we will show that any partial tree of $M$ on input $w$
that can be mapped to the model 
and that is not an accepting tree of $M$ on input $w$
can be properly extended to a
strictly larger partial tree of $M$ on input $w$ that can be mapped to the model as well.
If there would not exist an accepting tree of $M$ on input $w$ then
we would obtain an infinite strictly increasing sequence
of partial trees of $M$ on input $w$.
But we show that this cannot happen by giving a 
finite upper bound on the size of partial trees of $M$ on input $w$.

Let $w \in \Sigma^*$ be a string such that the formula
$f_\ssl(w)$ is $\ssl$-satisfiable.
We set $n:=|w|$.
Let $(W, \stackrel{\Diamond}{\to},\stackrel{L}{\to},\sigma)$ be a cross axiom model and $r_0\in W$ a point such that
$r_0 \models f_\ssl(w)$. The quintuple
$$\mathit{Model}:=(W, \stackrel{\Diamond}{\to},\stackrel{L}{\to},\sigma,r_0)$$
will be important in the following.
Points in $W$ that cannot be reached from $r_0$ by finitely many 
$\stackrel{\Diamond}{\to}$ and $\stackrel{L}{\to}$-steps (in any order)
can be deleted from $W$ with no harm: the resulting smaller quintuple will still be a
model of $f_\ssl(w)$.
Hence, we will assume without loss of generality that
every point $x \in W$ can be reached from $r_0$ by finitely many
$\stackrel{\Diamond}{\to}$ and $\stackrel{L}{\to}$-steps
(in any order).
Note that now the cross property implies that for any $x\in W$ there exists
some $x' \in W$ with $r_0 \stackrel{L}{\to} x'$ and $x' \stackrel{\Diamond}{\to} x$.
Hence, if $\varphi$ is a formula with $r_0 \models K \Box \varphi$
then, for all $x \in W$, we have $x \models \varphi$.
For every $x\in W$, let
$\mathit{Cloud}_x$
be the $\stackrel{L}{\to}$-equivalence class of $x$.
Remember that for every $\stackrel{L}{\to}$-equivalence class
and every shared variable 
$\alpha^{\mathit{string}}_i$ for 
$\mathit{string} \in \{\mathrm{time}, \mathrm{pos},
       \mathrm{state}, \mathrm{written}, \mathrm{read}\}$
and a natural number $i$, the truth value of this shared variable 
is the same in all elements of the $\stackrel{L}{\to}$-equivalence class.

Partial trees of $M$ on input $w$ as introduced in 
Subsection~\ref{subsection:ATM} will play an important role
in the following.
We will write a partial tree of $M$ on input $w$
similarly as in Subsection~\ref{subsection:ATM}
as a triple $T=(V,E,c)$, but with the difference
that we will describe configurations as at the beginning of this section:
the labeling function $c$ will be a function of the form
$c:V\to Q \times \{0,\ldots,2^{N+1}-2\}\times \Gamma^{2^{N+1}-1}$.
If $T=(V,E,c)$ is a partial tree of $M$ on input $w$
with root $\mathit{root}$
then a function $\pi:V\to W$ is called a
\emph{morphism from $T$ to $\mathit{Model}$}
if it satisfies the following four conditions:
\begin{enumerate}
\item
$\pi(\mathit{root})=r_0$,
\item
$(\forall v,v' \in V) \, ( \text{ if } v E v' \text{ then }
   \mathit{Cloud}_{\pi(v)} \stackrel{\Diamond}{\to}^{\stackrel{L}{\to}}\mathit{Cloud}_{\pi(v')})$,
\item
$(\forall v \in V \setminus\{\mathit{root}\}) \,
\pi(v) \models \alpha^\mathrm{written}_{\mathit{written}(v)}$,
\item
$(\forall v \in V)\, 
\pi(v) \models 
\begin{array}[t]{l} \bigl(B
  \wedge (\underline{\alpha}^\mathrm{time} = \mathrm{bin}_N(\mathit{time}(v)))\\
  \wedge (\underline{\alpha}^\mathrm{pos} = \mathrm{bin}_{N+1}(\mathit{pos}(v))) 
  \wedge \alpha^\mathrm{state}_{\mathit{state}(v)}
  \wedge \alpha^\mathrm{read}_{\mathit{read}(v)} \bigr).
\end{array}$
\end{enumerate}

We say that $T$ \emph{can be mapped to} $\mathit{Model}$ if there exists a morphism
from $T$ to $\mathit{Model}$.
Below we shall prove the following lemma.

\begin{lemma}
\label{lemma:inductionSSL}
If a partial tree $T=(V,E,c)$ of $M$ on input $w$
can be mapped to $\mathit{Model}$
and is not an accepting tree of $M$ on input $w$
then there exists a partial tree $T=(\widetilde{V},\widetilde{E},\widetilde{c})$ of $M$ on input $w$
that can be mapped to $\mathit{Model}$ and that satisfies
$V \subsetneq \widetilde{V}$.
\end{lemma}

Before we prove this lemma, we deduce the desired assertion from it.
Let 
\[ D:= \max(\{ |\delta(q,\eta)| ~:~ q \in Q, \eta \in\Gamma\}) . \]
Then, due to Condition III in the definition of a
``partial tree of $M$ on input $w$'',
any node in any partial tree of $M$ on input $w$ has
at most $D$ successors.  
As any computation of $M$ on $w$ stops after
at most $2^N-1$ steps, any partial tree of $M$ on input $w$ has
at most 
\[ \widetilde{D}:=(D^{2^N}-1)/(D-1) \]
nodes. 
We claim that the rooted and labeled tree
\[ T_0:=(\{\mathit{root}\},\emptyset,c) \text{ where }
c(\mathit{root}):=(q_0,2^N-1,\#^{2^N}w\#^{2^N-1-n}) \]
is a partial tree of $M$ on input $w$ and can be mapped to $\mathit{Model}$.
Indeed, Condition I in the definition of a ``partial tree of $M$ on input $w$'' is
satisfied because the node $\mathit{root}$ is labeled with the initial configuration of $M$ on input $w$.
Conditions II, III, and IV are satisfied because $T_0$ does not have any inner nodes.
Condition $V^\prime$ is satisfied because $q_0 \neq q_\mathrm{reject}$, and this
follows from $r_0 \models \alpha^\mathrm{state}_{q_0}$
(this is a part of $r_0 \models \mathit{start}$)
and $r_0 \models \neg \alpha^\mathrm{state}_{q_\mathrm{reject}}$
(this follows from $r_0 \models K \Box \mathit{no\_reject}$).
Thus, $T_0$ is indeed a partial tree of $M$ on input $w$.
Now we show that $T_0$ can be mapped $\mathit{Model}$.
Of course, we define $\pi:\{\mathit{root}\}\to W$ by $\pi(\mathit{root}):=r_0$.
We claim that $\pi$ is a morphism from $T$ to $\mathit{Model}$.
We check the four conditions one by one.
\begin{enumerate}
\item
The condition $\pi(\mathit{root})=r_0$ is true by definition.
\item
The second condition is satisfied trivially because the tree $T_0$ does not have any edges.
\item
The third condition in the definition of a ``morphism from $T$ to $\mathit{Model}$''
is satisfied trivially because $T_0$ has only one node, its root.
\item
On the one hand, we have
$\mathit{time}(\mathit{root})=0$,
$\mathit{pos}(\mathit{root})=2^N-1$,\\
$\mathit{state}(\mathit{root})=q_0$, and
$\mathit{read}(\mathit{root})=\#$.\\
On the other hand, the condition $r_0 \models \mathit{start}$ implies 
\[ r_0 \models B \wedge (\underline{\alpha}^\mathrm{time}=\bin_N(0))
   \wedge (\underline{\alpha}^\mathrm{pos}=\bin_{N+1}(2^N-1))
   \wedge \alpha^\mathrm{state}_{q_0}
   \wedge \alpha^\mathrm{read}_{\#} . \]
\end{enumerate}
Thus, we have shown that $T_0$ is a partial tree of $M$ on input $w$
and that $T_0$ can be mapped to $\mathit{Model}$.

If there would not exist an accepting tree of $M$ on input $w$
then, starting with $T_0$ and using Lemma~\ref{lemma:inductionSSL} we could construct an infinite sequence of partial
trees $T_0,T_1,T_2,\ldots$ of $M$ on input $w$ that 
can be mapped to $\mathit{Model}$ 
such that the number of nodes in these trees is strictly increasing.
But we have seen that any partial tree of $M$ on input $w$ can have at most $\widetilde{D}$ nodes.
Thus, there exists an accepting tree of $M$ on input $w$.
We have shown $w\in L$.

In order to complete the proof of Theorem~\ref{theorem:SSL-EXPSPACE-hard}
it remains to prove Lemma~\ref{lemma:inductionSSL}.

\begin{proof}[Proof of Lemma~\ref{lemma:inductionSSL}]
Let $T=(V,E,c)$ be a partial tree of $M$ on input $w$
that is not an accepting tree of $M$ on input $w$
and that can be mapped to $\mathit{Model}$.
Let $\pi:V \to W$ be a morphism from $T$ to $\mathit{Model}$.
Then $T$ has a leaf $\widehat{v}$ such that the state
$q:=\mathit{state}(\widehat{v})$ is either an element of $Q_\exists$ or of $Q_\forall$.
First we treat the case that it is an element of $Q_\exists$,
then the case that it is an element of $Q_\forall$.

So, let us assume that $q \in Q_\exists$.
We define $\eta:=\mathit{read}(\widehat{v})$.
Then, because $\pi$ is a morphism from $T$ to $\mathit{Model}$, we have
\[ \pi(\widehat{v}) \models (\alpha^\mathrm{state}_q 
   \wedge \alpha^\mathrm{read}_\eta) , \]
hence, due to $\pi(\widehat{v}) \models \mathit{computation}$,
\[ \pi(\widehat{v}) \models 
 \bigvee_{(r,\theta,\mathit{left}) \in \delta(q,\eta)}  \mathit{compstep}_{\mathrm{left}}(r,\theta)
   \vee   \bigvee_{(r,\theta,\mathit{right}) \in \delta(q,\eta)}  \mathit{compstep}_{\mathrm{right}}(r,\theta) . \]
Let us assume that there is an element $(r,\theta,\mathit{left}) \in \delta(q,\eta)$
such that $\pi(\widehat{v}) \models  \mathit{compstep}_{\mathrm{left}}(r,\theta)$
(the other case, the case when there is an element $(r,\theta,\mathit{right}) \in \delta(q,\eta)$
such that $\pi(\widehat{v}) \models \mathit{compstep}_{\mathrm{right}}(r,\theta)$, is treated analogously).  
We claim that we can define the new tree $\widetilde{T}=(\widetilde{V},\widetilde{E},\widetilde{c})$
as follows:
\begin{itemize}
\item
$\widetilde{V} := V \cup\{\widetilde{v}\}$ where $\widetilde{v}$ is a new element (not in $V$),
\item
$\widetilde{E} := E \cup \{(\widehat{v},\widetilde{v})\}$,
\item
$\widetilde{c}(x) := \begin{cases}
          c(x) & \text{for all } x \in V, \\
          c' & \text{for } x=\widetilde{v}, \text{ where $c'$ is the configuration that is reached from $c(\widehat{v})$}  \\
             & \text{in the computation step given by } ((q,\eta),(r,\theta,\mathit{left})) \in\delta.
             \end{cases}$
\end{itemize}
Before we show that $\widetilde{T}$
is a partial tree of $M$ on input $w$,
we define a function $\widetilde{\pi}:\widetilde{V}\to W$
that we will show to be a morphism from $\widetilde{T}$ to $\mathit{Model}$.

As $\widehat{v}$ is an element of a partial tree of $M$ on input $w$
with $\mathit{state}(\widehat{v}) \in Q_\exists$, at least one more computation step
can be done. As any computation of $M$ on input $w$ stops after at most $2^N-1$ steps,
we observe that the number $i:=\mathit{time}(\widehat{v})$ 
satisfies $0 \leq i < 2^N-1$. 
Then $\{0,\ldots,N-1\} \setminus\mathrm{Ones}(i) \neq \emptyset$.
We set $k:= \min(\{0,\ldots,N-1\} \setminus\mathrm{Ones}(i))$.
The assumption that $\pi$ is a morphism from $T$ to $\mathit{Model}$
implies $\pi(\widehat{v}) \models (\underline{\alpha}^\mathrm{time} = \mathrm{bin}_N(i))$.
We conclude 
$\pi(\widehat{v}) \models \mathrm{rightmost\_zero}(\underline{\alpha}^\mathrm{time},k)$.
As during each computation step, the tape head can move at most one step to the left or to the right
and as the computation started in position $2^N-1$
the number $j:=\mathit{pos}(\widehat{v})$ satisfies
$0 < j \leq 2^{N+1}-3$.
Then $\mathrm{Ones}(j) \neq \emptyset$.
We set $l:= \min\mathrm{Ones}(j)$.
The assumption that $\pi$ is a morphism from $T$ to $\mathit{Model}$
implies $\pi(\widehat{v}) \models (\underline{\alpha}^\mathrm{pos} = \mathrm{bin}_{N+1}(j))$.
We conclude
$\pi(\widehat{v}) \models \mathrm{rightmost\_one}(\underline{\alpha}^\mathrm{pos},l)$.
Furthermore, as $\pi$ is a morphism from $T$ to $\mathit{Model}$ we have $\pi(\widehat{v}) \models B$.
Summarizing this, we have
\[
\pi(\widehat{v}) \models \bigl( B \wedge \mathrm{rightmost\_zero}(\underline{\alpha}^\mathrm{time},k)
    \wedge \mathrm{rightmost\_one}(\underline{\alpha}^\mathrm{pos},l) \bigr).
\]
Due to $\pi(\widehat{v}) \models  \mathit{compstep}_{\mathrm{left}}(r,\theta)$
there exist an element $x \in \mathit{Cloud}_{\pi(\widehat{v})}$ and an element $y \in W$ such that
$x \stackrel{\Diamond}{\to} y$
as well as
\begin{eqnarray*}
 x &\models&
     B \wedge (\underline{X}^\mathrm{time}=\underline{\alpha}^\mathrm{time},>k)
        \wedge \mathrm{rightmost\_one}(\underline{X}^\mathrm{time},k) \\
  && \wedge (\underline{X}^\mathrm{pos}=\underline{\alpha}^\mathrm{pos},>l)
                                            \wedge \mathrm{rightmost\_zero}(\underline{X}^\mathrm{pos},l) 
\end{eqnarray*}
and
\[ y \models
(\underline{\alpha}^\mathrm{time}=\underline{X}^\mathrm{time})
                        \wedge (\underline{\alpha}^\mathrm{pos}=\underline{X}^\mathrm{pos}) 
                        \wedge \alpha^\mathrm{state}_{r} \wedge \alpha^\mathrm{written}_\theta
                        \wedge (\underline{\alpha}^\mathrm{read}=\underline{X}^\mathrm{read}) .
\]
We claim that we can define the desired function $\widetilde{\pi}:\widetilde{V} \to W$ by
\[ \widetilde{\pi}(v) := \begin{cases}
       \pi(v) & \text{if } v \in V, \\
       y & \text{if } v=\widetilde{v}.
       \end{cases} 
\]

We have to show that $\widetilde{T}$ is a partial tree of $M$ on input $w$
and that $\widetilde{\pi}$ is a morphism from $\widetilde{T}$ to $\mathit{Model}$.
Condition I in the definition of a 
``partial tree of $M$ on input $w$''
is satisfied because $\widetilde{T}$ has the same root as $T$,
and the label of the root does not change.
A node in $\widetilde{T}$ is an internal node of $\widetilde{T}$ if, and only if, it is either an internal node of $T$
or equal to $\widehat{v}$. For internal nodes of $T$ Conditions II, III, and IV are satisfied by assumption
(and due to the fact that the labels of nodes in $V$ do not change).
The new internal node $\widehat{v}$ satisfies Condition II by our definition of $\widetilde{c}(\widetilde{v})$.
Condition III is satisfied for $\widehat{v}$ because $\widehat{v}$ has exactly one successor.
And Condition IV does not apply to $\widehat{v}$ because $\widehat{v}\in Q_\exists$.
A node in $\widetilde{T}$ is a leaf if, and only if, it is either equal to $\widetilde{v}$
or a leaf in $T$ different from $\widehat{v}$.
For the leaves in $T$ different from $\widehat{v}$ Condition $\mathrm{V}^\prime$ is satisfied by assumption
(and due to the fact that the labels of nodes in $V$ do not change).
For the new leaf $\widetilde{v}$ in $\widetilde{T}$ Condition $\mathrm{V}^\prime$
says $\mathit{state}(\widetilde{v}) \neq q_\mathrm{reject}$.
This is true because on the one hand $\mathit{state}(\widetilde{v}) = r$
and on the other hand
$y \models \alpha^\mathrm{state}_r$ and
$y \models \neg \alpha^\mathrm{state}_{q_\mathrm{reject}}$
(this follows from $r \models K \Box \mathit{no\_reject}$).
We have shown that $\widetilde{T}$ is a partial tree of $M$ on input $w$.

Now we show that $\widetilde{\pi}$ is a morphism from $\widetilde{T}$ to $\mathit{Model}$.
The first condition in the definition of a ``morphism from $\widetilde{T}$ to $\mathit{Model}$''
is satisfied because $\widetilde{\pi}(\mathit{root}) = \pi(\mathit{root})=r_0$.
Let us look at the second condition and let us assume that $v,v' \in \widetilde{V}$ satisfy
$v \widetilde{E} v'$.
We distinguish between two different cases for $v$ and $v'$.
\begin{enumerate}
\item
Case: $v' \in V$. Then our assumption $v \widetilde{E} v'$ implies $v \in V$ and $v E v'$. In this case
the facts $\widetilde{\pi}(v)=\pi(v)$ and $\widetilde{\pi}(v')=\pi(v')$
as well as the assumption that $\pi:V\to W$ is a morphism from $T$ to $\mathit{Model}$
imply the desired assertion:
\[ \mathit{Cloud}_{\widetilde{\pi}(v)}
  = \mathit{Cloud}_{\pi(v)} 
  \stackrel{\Diamond}{\to}^{\stackrel{L}{\to}}
  \mathit{Cloud}_{\pi(v')})
  =  \mathit{Cloud}_{\widetilde{\pi}(v')} . \]
\item
Case: $v'=\widetilde{v}$.
Then our assumption $v \widetilde{E} v'$ implies $v= \widehat{v}$.
On the one hand, we have
$x \in \mathit{Cloud}_{\pi(\widehat{v})} = \mathit{Cloud}_{\widetilde{\pi}(\widehat{v})}$,
on the other hand $y = \widetilde{\pi}(\widetilde{v})$, hence,
$y \in \mathit{Cloud}_{\widetilde{\pi}(\widetilde{v})}$.
With $x \stackrel{\Diamond}{\to} y$ we obtain the desired assertion
$\mathit{Cloud}_{\widetilde{\pi}(\widehat{v})} 
\stackrel{\Diamond}{\to}^{\stackrel{L}{\to}}
\mathit{Cloud}_{\widetilde{\pi}(\widetilde{v})}$.
\end{enumerate}
The third condition in the definition of a
``morphism from $\widetilde{T}$ to $\mathit{Model}$''
is satisfied for $v\in V\setminus\{\mathit{root}\}$ by assumption
(and by $\widetilde{\pi}(v)=\pi(v)$ and $\widetilde{c}(v) = c(v)$).
It is satisfied for $\widetilde{v}$ because
$\mathit{written}(\widetilde{v}) = \theta$,
because $\widetilde{\pi}(\widetilde{v})=y$,
and because $y \models \alpha^\mathrm{written}_\theta$.
We come to the fourth condition.
It is satisfied for $v\in V\setminus\{\mathit{root}\}$ by assumption
(and due to $\widetilde{\pi}(v)=\pi(v)$ and $\widetilde{c}(v) = c(v)$).
We still need to show that it is satisfied for $\widetilde{v}$.
Remember $\widetilde{\pi}(\widetilde{v})=y$.
We need to show
\begin{eqnarray*}
 y &\models& 
  \bigl(B
  \wedge (\underline{\alpha}^\mathrm{time} = \mathrm{bin}_N(\mathit{time}(\widetilde{v})))
  \wedge (\underline{\alpha}^\mathrm{pos} = \mathrm{bin}_{N+1}(\mathit{pos}(\widetilde{v}))) 
  \wedge \alpha^\mathrm{state}_{\mathit{state}(\widetilde{v})}
  \wedge \alpha^\mathrm{read}_{\mathit{read}(\widetilde{v})} \bigr).
\end{eqnarray*}
This assertion consists really of five assertions. We treat them one by one.
\begin{itemize}
\item
The condition $y \models B$ is satisfied because $x \models B$ and
$x \stackrel{\Diamond}{\to} y$ and because $B$ is persistent.
\item
In the trees $T$ and $\widetilde{T}$ we have $\mathit{time}(\widehat{v})=i$, 
and in the tree $\widetilde{T}$ we have $\mathit{time}(\widetilde{v})=i+1$.
We wish to show $y \models (\underline{\alpha}^\mathrm{time} = \mathrm{bin}_N(i+1))$.
We have already seen
$\pi(\widehat{v}) \models (\underline{\alpha}^\mathrm{time} = \mathrm{bin}_N(i))$
and $\pi(\widehat{v}) \models \mathrm{rightmost\_zero}(\underline{\alpha}^\mathrm{time},k)$.
As $x \in \mathit{Cloud}_{\pi(\widehat{v})}$, we obtain
\[ x \models ((\underline{\alpha}^\mathrm{time} = \mathrm{bin}_N(i))
   \wedge \mathrm{rightmost\_zero}(\underline{\alpha}^\mathrm{time},k)) . \]
The conditions $x \stackrel{\Diamond}{\to} y$ as well as
\begin{eqnarray*}
 x &\models&
 (\underline{X}^\mathrm{time}=\underline{\alpha}^\mathrm{time},>k)
        \wedge \mathrm{rightmost\_one}(\underline{X}^\mathrm{time},k) \text{ and } \\
y & \models& 
    (\underline{\alpha}^\mathrm{time}=\underline{X}^\mathrm{time})
\end{eqnarray*}
imply $y \models (\underline{\alpha}^\mathrm{time} = \mathrm{bin}_N(i+1))$.
\item
In the trees $T$ and $\widetilde{T}$ we have $\mathit{pos}(\widehat{v})=j$, 
and in the tree $\widetilde{T}$ we have $\mathit{pos}(\widetilde{v})=j-1$.
We wish to show $y \models (\underline{\alpha}^\mathrm{pos} = \mathrm{bin}_{N+1}(j-1))$.
We have already seen
$\pi(\widehat{v}) \models (\underline{\alpha}^\mathrm{pos} = \mathrm{bin}_{N+1}(j))$
and $\pi(\widehat{v}) \models \mathrm{rightmost\_one}(\underline{\alpha}^\mathrm{pos},l)$.
As $x \in \mathit{Cloud}_{\pi(\widehat{v})}$, we obtain
\[ x \models ((\underline{\alpha}^\mathrm{pos} = \mathrm{bin}_{N+1}(j))
   \wedge \mathrm{rightmost\_one}(\underline{\alpha}^\mathrm{pos},l)) . \]
The conditions $x \stackrel{\Diamond}{\to} y$ as well as
\begin{eqnarray*}
 x &\models&
  (\underline{X}^\mathrm{pos}=\underline{\alpha}^\mathrm{pos},>l)
                                            \wedge \mathrm{rightmost\_zero}(\underline{X}^\mathrm{pos},l) \text{ and } \\
y & \models& 
  (\underline{\alpha}^\mathrm{pos}=\underline{X}^\mathrm{pos})
\end{eqnarray*}
imply $y \models (\underline{\alpha}^\mathrm{pos} = \mathrm{bin}_{N+1}(j-1))$.
\item
We have $\mathit{state}(\widetilde{v}) = r$.
And we have $y \models \alpha^\mathit{state}_r$.
\item
Let $\gamma:=\mathit{read}(\widetilde{v})$ in $\widetilde{T}$.
We wish to show $y \models \alpha^\mathrm{read}_\gamma$.
As we know $y \models 
(\underline{\alpha}^\mathrm{read}=\underline{X}^\mathrm{read})$
it is sufficient to show
$y \models X^\mathrm{read}_\gamma$.
We have already seen
$y \models (\underline{\alpha}^\mathrm{time} = \mathrm{bin}_N(i+1))$
and 
$y \models (\underline{\alpha}^\mathrm{pos} = \mathrm{bin}_{N+1}(j-1))$.
As we know
$y \models 
(\underline{\alpha}^\mathrm{time}=\underline{X}^\mathrm{time})
                        \wedge (\underline{\alpha}^\mathrm{pos}=\underline{X}^\mathrm{pos})$
we conclude
$y \models ((\underline{X}^\mathrm{time}=\mathrm{bin}_N(i+1))
     \wedge (\underline{X}^\mathrm{pos}=\mathrm{bin}_{N+1}(j-1)))$.
We have seen $y \models B$ as well.
Furthermore, we have $y \models \mathit{time\_after\_previous\_visit}$.
The first line in the formula $\mathit{time\_after\_previous\_visit}$
implies 
$y \models (\underline{X}^{\text{time-apv}} \leq \underline{X}^\mathrm{time})$,
hence,
$y \models (\underline{X}^{\text{time-apv}} \leq \mathrm{bin}_N(i+1))$.
Let $v_m \in V$ for $m=0,\ldots,i+1$ be the uniquely determined node on the path from 
$\mathit{root}$ to $\widetilde{v}$ with $\mathit{time}(v_m)=m$ 
(hence $v_0=\mathit{root}$, $v_i = \widehat{v}$, and $v_{i+1} = \widetilde{v}$).
Then
\[ v_0 \widetilde{E} v_1 \widetilde{E} \ldots \widetilde{E} v_i \widetilde{E} v_{i+1} . \]
By the second condition in the definition of a
``morphism from $\widetilde{T}$ to $\mathit{Model}$''
\[ \mathit{Cloud}_{\widetilde{\pi}(v_0)} \stackrel{\Diamond}{\to}^{\stackrel{L}{\to}}
   \mathit{Cloud}_{\widetilde{\pi}(v_1)} \stackrel{\Diamond}{\to}^{\stackrel{L}{\to}}
   \ldots \stackrel{\Diamond}{\to}^{\stackrel{L}{\to}}
   \mathit{Cloud}_{\widetilde{\pi}(v_i)} \stackrel{\Diamond}{\to}^{\stackrel{L}{\to}}
   \mathit{Cloud}_{\widetilde{\pi}(v_{i+1})}.
\]
By repeated application of the cross property 
and by starting with $z_{i+1}:=y$ 
we conclude that for $m=i+1,i,\ldots,1,0$
there exists some $z_m \in \mathit{Cloud}_{\widetilde{\pi}(v_m)}$ with $z_m \stackrel{\Diamond}{\to} y$.
Let us consider $m \in \{0,1,\ldots,i+1\}$.
As $y \models B$ we also have $z_m \models B$
(remember that variables are persistent in $\ssl$).
Due to 
$\widetilde{\pi}(v_m) \models ((\underline{\alpha}^\mathrm{time} = \mathrm{bin}_N(m))
    \wedge (\underline{\alpha}^\mathrm{pos} = \mathrm{bin}_{N+1}( \mathit{pos}(v_m))))$
we obtain
$z_m \models    
      ((\underline{\alpha}^\mathrm{time} = \mathrm{bin}_N(m))
    \wedge (\underline{\alpha}^\mathrm{pos} = \mathrm{bin}_{N+1}( \mathit{pos}(v_m))))$
as well.
Due to
$y \models ((\underline{X}^\mathrm{time}=\mathrm{bin}_N(i+1))
     \wedge (\underline{X}^\mathrm{pos}=\mathrm{bin}_{N+1}(j-1)))$
and the persistence of variables we obtain
$z_m \models ((\underline{X}^\mathrm{time}=\mathrm{bin}_N(i+1))
     \wedge (\underline{X}^\mathrm{pos}=\mathrm{bin}_{N+1}(j-1)))$
as well.
Furthermore, we have
$z_m \models \mathit{time\_after\_previous\_visit}$.
We distinguish between the two cases whether the cell $j-1$ has been visited
on the path from $\mathit{root}$ to $\widehat{v}$ or not.

Let us first consider the case when the cell $j-1$ has not been visited
on the path from $\mathit{root}$ to $\widehat{v}$.
Then, on the one hand, the symbol $\gamma = \mathit{read}(\widetilde{v})$
is still the initial symbol in the cell $j-1$.
On the other hand,
for all  $m \in \{0,\ldots,i\}$ we have
$\mathit{pos}(v_m) \neq j-1$ and
\[ z_m \models ( \underline{\alpha}^\mathrm{time} < \underline{X}^\mathrm{time}
              \wedge \underline{\alpha}^\mathrm{pos} \neq \underline{X}^\mathrm{pos} )
              \rightarrow \underline{X}^{\text{time-apv}} \neq \underline{\alpha}^\mathrm{time} + 1) .
\]
Together with
$z_m \models ((\underline{X}^\mathrm{time}=\mathrm{bin}_N(i+1))
     \wedge (\underline{X}^\mathrm{pos}=\mathrm{bin}_{N+1}(j-1)))$
and
$z_m \models ((\underline{\alpha}^\mathrm{time} = \mathrm{bin}_N(m))
    \wedge (\underline{\alpha}^\mathrm{pos} = \mathrm{bin}_{N+1}( \mathit{pos}(v_m))))$
we conclude that 
$z_m \models (\underline{X}^{\text{time-apv}} \neq \mathrm{bin}_N(m + 1))$
for $m\in\{0,\ldots,i\}$.
The persistence of $\underline{X}^{\text{time-apv}}$ implies that
$y \models (\underline{X}^{\text{time-apv}} \neq \mathrm{bin}_N(m + 1))$
for $m\in\{0,\ldots,i\}$.
Together with
$y \models (\underline{X}^{\text{time-apv}} \leq \underline{X}^\mathrm{time})$
we conclude that the binary value of $\underline{X}^{\text{time-apv}}$ in $y$ must be $0$.
Now the fact
$y \models \mathit{get\_the\_right\_symbol}$
implies $y \models X^\mathrm{read}_\gamma$.

Let us consider the second case,
the case when the cell $j-1$ has been visited
on the path from $\mathit{root}$ to $\widehat{v}$.
Let $v_{m'}$ be the last node on this path with $\mathit{pos}(v_{m'})=j-1$.
Then $0 \leq m' \leq i$.
On the one  hand, then the symbol $\gamma = \mathit{read}(\widetilde{v})$
is the symbol that was written into the cell $j-1$ in the computation
step from node $v_{m'}$ to node $v_{m'+1}$,
and by the third condition in the definition of a
``morphism from $\widetilde{T}$ to $\mathit{Model}$'' we have
$\widetilde{\pi}(v_{m'+1}) \models \alpha^\mathrm{written}_\gamma$, hence,
$z_{m'+1} \models \alpha^\mathrm{written}_\gamma$.
On the other hand,
from $y \models \mathit{time\_after\_previous\_visit}$ we get
$y \models (\underline{X}^{\text{time-apv}} \leq \underline{X}^\mathrm{time})$, 
hence,
$y \models (\underline{X}^{\text{time-apv}} \leq \mathrm{bin}_N(i+1))$.
From $z_{m'} \models \mathit{time\_after\_previous\_visit}$ we get
\[
z_{m'} \models \bigl( (  \underline{\alpha}^\mathrm{time} < \underline{X}^\mathrm{time})
              \wedge (\underline{\alpha}^\mathrm{pos} = \underline{X}^\mathrm{pos})\bigr) 
              \rightarrow (\underline{\alpha}^\mathrm{time}<\underline{X}^{\text{time-apv}} ) .
\]
Together with $z_{m'} \models (\underline{\alpha}^\mathrm{time} = \mathrm{bin}_N(m'))$
we obtain
$z_{m'} \models ( \mathrm{bin}_N(m') < \underline{X}^{\text{time-apv}} )$.
And similarly as in the first case,
from 
\[ z_m \models \bigl( ( \underline{\alpha}^\mathrm{time} < \underline{X}^\mathrm{time})
              \wedge( \underline{\alpha}^\mathrm{pos} \neq \underline{X}^\mathrm{pos} )\bigr)
              \rightarrow( \underline{X}^{\text{time-apv}} \neq \underline{\alpha}^\mathrm{time} + 1) , \]
for $m=m'+1,\ldots,i$ we obtain
\[ z_m \models (\underline{X}^{\text{time-apv}} \neq \mathrm{bin}_N(m+1)) .\]
All this implies 
$y \models (\underline{X}^{\text{time-apv}} = \mathrm{bin}_N(1+m'))$
and
$z_{m'+1} \models (\underline{X}^{\text{time-apv}} = \mathrm{bin}_N(1+m'))$.
Then
$z_{m'+1} \models \mathit{get\_the\_right\_symbol}$
implies  $z_{m'+1} \models (\underline{X}^\mathrm{read} = \underline{\alpha}^\mathrm{written})$.
With
$z_{m'+1} \models \alpha^\mathrm{written}_\gamma$
we conclude 
$z_{m'+1} \models X^\mathrm{read}_\gamma$, hence,
$y \models X^\mathrm{read}_\gamma$.
That was to be shown.    
\end{itemize}
Thus, $\widetilde{T}$ is not only a partial tree of $M$ on input $w$ but can
also be mapped to $\mathit{Model}$.
This ends the treatment of the case $q \in Q_\exists$.

Now we consider the other case, the case $q \in Q_\forall$.
We define $\eta:=\mathit{read}(\widehat{v})$.
Let 
\[ (r_1,\theta_1,\mathit{dir}_1),\ldots,(r_d,\theta_d,\mathit{dir}_d) \]
be the elements of
$\delta(q,\eta)$ where $d\geq 1$ and $\mathit{dir}_m \in \{\mathit{left},\mathit{right}\}$, for $m=1,\ldots,d$.
We claim that we can define the new tree $\widetilde{T}=(\widetilde{V},\widetilde{E},\widetilde{c})$
as follows:
\begin{itemize}
\item
$\widetilde{V} := V \cup\{\widetilde{v}_1,\ldots,\widetilde{v}_d\}$ where $\widetilde{v}_1,\ldots,\widetilde{v}_d$ 
are new (not in $V$) pairwise different elements,
\item
$\widetilde{E} := E \cup \{(\widehat{v},\widetilde{v}_1),\ldots,(\widehat{v},\widetilde{v}_d)\}$,
\item
$\widetilde{c}(x) := \begin{cases}
          c(x) & \text{for all } x \in V, \\
          c'_m & \text{for } x=\widetilde{v}_m, \text{where $c'_m$ is the configuration that is reached from }\\
             & \text{$c(\widehat{v})$ in the computation step given by } ((q,\eta),(r_m,\theta_m,\mathit{dir}_m)) \in\delta.
             \end{cases}$
\end{itemize}
Before we show that $\widetilde{T}$
is a partial tree of $M$ on input $w$,
we define a function $\widetilde{\pi}:\widetilde{V}\to W$
that we will show to be a morphism from $\widetilde{T}$ to $\mathit{Model}$.
Since $\pi$ is a morphism from $T$ to $\mathit{Model}$), we have
\[ \pi(\widehat{v}) \models
B \wedge (\underline{\alpha}^\mathrm{time} = \mathrm{bin}_N(\mathit{time}(\widehat{v})))
  \wedge (\underline{\alpha}^\mathrm{pos} = \mathrm{bin}_{N+1}(\mathit{pos}(\widehat{v}))) 
  \wedge \alpha^\mathrm{state}_{\mathit{state}(\widehat{v})}
  \wedge \alpha^\mathrm{read}_{\mathit{read}(\widehat{v})} ,
\]
hence, due to $\pi(\widehat{v}) \models \mathit{computation}$,
\[ \pi(\widehat{v}) \models 
 \bigwedge_{(r,\theta,\mathit{left}) \in \delta(q,\eta)}  \mathit{compstep}_{\mathrm{left}}(r,\theta)
   \wedge  \bigwedge_{(r,\theta,\mathit{right}) \in \delta(q,\eta)}  \mathit{compstep}_{\mathrm{right}}(r,\theta) . \]
As in the case $q \in Q_\exists$ one shows that the numbers
$i:=\mathit{time}(\widehat{v})$
and
 $j:=\mathit{pos}(\widehat{v})$
satisfy $0 \leq i < 2^N-1$ 
and
$0 < j \leq 2^{N+1}-3$, and one defines
\begin{eqnarray*}
k &:=& \min(\{0,\ldots,N-1\} \setminus\mathrm{Ones}(i)), \\
l_\mathrm{left} &:=& \min(\mathrm{Ones}(j)), \text{ and} \\
l_\mathrm{right} &:=& \min(\{0,\ldots,N\} \setminus\mathrm{Ones}(j)).
\end{eqnarray*}
As in the case $q \in Q_\exists$ one obtains
\[
\pi(\widehat{v}) \models B \wedge \mathrm{rightmost\_zero}(\underline{\alpha}^\mathrm{time},k) 
     \wedge \mathrm{rightmost\_one}(\underline{\alpha}^\mathrm{pos},l_\mathrm{left})
    \wedge \mathrm{rightmost\_zero}(\underline{\alpha}^\mathrm{pos},l_\mathrm{right}).
\]
Let us now consider some $m\in\{1\ldots,d\}$.
If $\mathit{dir}_m=\mathit{left}$ then,
due to $\pi(\widehat{v}) \models  \mathit{compstep}_{\mathrm{left}}(r,\theta)$,
there exist an element $x_m \in \mathit{Cloud}_{\pi(\widehat{v})}$ and an element $y_m \in W$ such that
$x_m \stackrel{\Diamond}{\to} y_m$
as well as
\begin{eqnarray*}
 x_m &\models&
     B \wedge (\underline{X}^\mathrm{time}=\underline{\alpha}^\mathrm{time},>k)
        \wedge \mathrm{rightmost\_one}(\underline{X}^\mathrm{time},k) \\
  && \wedge (\underline{X}^\mathrm{pos}=\underline{\alpha}^\mathrm{pos},>l_\mathrm{left})
                                            \wedge \mathrm{rightmost\_zero}(\underline{X}^\mathrm{pos},l_\mathrm{left}) 
\end{eqnarray*}
and
\[ y_m \models
(\underline{\alpha}^\mathrm{time}=\underline{X}^\mathrm{time})
                        \wedge (\underline{\alpha}^\mathrm{pos}=\underline{X}^\mathrm{pos}) 
                        \wedge \alpha^\mathrm{state}_{r} \wedge \alpha^\mathrm{written}_\theta
                        \wedge (\underline{\alpha}^\mathrm{read}=\underline{X}^\mathrm{read}) .
\]
Similarly, if $\mathit{dir}_m=\mathit{right}$ then,
due to $\pi(\widehat{v}) \models  \mathit{compstep}_{\mathrm{right}}(r,\theta)$,
there exist an element $x_m \in \mathit{Cloud}_{\pi(\widehat{v})}$ and an element $y_m \in W$ such that
$x_m \stackrel{\Diamond}{\to} y_m$
as well as
\begin{eqnarray*}
 x_m &\models&
     B \wedge (\underline{X}^\mathrm{time}=\underline{\alpha}^\mathrm{time},>k)
        \wedge \mathrm{rightmost\_one}(\underline{X}^\mathrm{time},k) \\
  && \wedge (\underline{X}^\mathrm{pos}=\underline{\alpha}^\mathrm{pos},>l_\mathrm{right})
                                            \wedge \mathrm{rightmost\_one}(\underline{X}^\mathrm{pos},l_\mathrm{right}) 
\end{eqnarray*}
and
\[ y_m \models
(\underline{\alpha}^\mathrm{time}=\underline{X}^\mathrm{time})
                        \wedge (\underline{\alpha}^\mathrm{pos}=\underline{X}^\mathrm{pos}) 
                        \wedge \alpha^\mathrm{state}_{r} \wedge \alpha^\mathrm{written}_\theta
                        \wedge (\underline{\alpha}^\mathrm{read}=\underline{X}^\mathrm{read}) .
\]
We claim that we can define the desired function $\widetilde{\pi}:\widetilde{V} \to W$ by
\[ \widetilde{\pi}(v) := \begin{cases}
       \pi(v) & \text{if } v \in V, \\
       y_m & \text{if } v=\widetilde{v}_m, \text{ for some } m \in \{1,\ldots,d\}.
       \end{cases} 
\]
Similarly as in the case $q \in Q_\exists$ one shows
that $\widetilde{T}$ is a partial tree of $M$ on input $w$.
Note that also Condition IV is satisfied for $\widehat{v}$.
Finally, similarly as in the case $q \in Q_\exists$ one shows
that $\widetilde{\pi}$ is a morphism from $\widetilde{T}$ to $\mathit{Model}$.
This ends the treatment of the case $q \in Q_\forall$.
We have proved Lemma~\ref{lemma:inductionSSL}.
\end{proof}

\section{Reduction of $\ssl$ to $\sxs$ }
\label{section:SSL-S4xS5}

In this section we prove the $\EXPSPACE$-hardness of $\sxs$ by showing the following result. Remember that according to Theorem~\ref{theorem:SSL-EXPSPACE-hard} the satisfiability problem of $\ssl$ is $\EXPSPACE$-hard.

\begin{theorem}
	The satisfiability problem of the bimodal logic $\SSL$ can be reduced in logarithmic space to the satisfiability problem of the bimodal logic
	$\sxs$.
\end{theorem}

To prove this theorem we proceed in three steps:
\begin{enumerate}
	\item We start with the definition of the reduction function. 
	\item We prove its correctness.
	\item We show that the reduction function can be computed using not more than logarithmic space.
\end{enumerate}

\subsection{The Reduction Function}
\label{subsection:DefinitionReduction}

We show that the satisfiability problem of $\SSL$ can be reduced
to the satisfiability problem of $\SfourSfive$.
To this end we define a translation $\widehat{T}$
of bimodal formulas in the language $\mathcal{L}$
to bimodal formulas in $\mathcal{L}$
such that for all $\varphi \in \mathcal{L}$
\[\varphi \text{ is $\SSL$-satisfiable } \iff \widehat{T}(\varphi) \text{ is $\SfourSfive$-satisfiable.}\]

For the reduction we face two main problems. 
\begin{enumerate}
\item
The first problem is that in cross axiom models literals are persistent.
To handle this we make sure that the translation $\widehat{T}(\varphi)$
contains a subformula postulating the persistence of literals. 
\item
The second problem is that in general cross axiom models do not satisfy right commutativity.
To handle this we add to each cloud in a cross-axiom model a special ``new point''
serving as successor point for all points in a predecessor cloud that fail to have in the original model a successor point in this cloud. 
In order to distinguish the new points from the original ones we use a special propositional variable $\emph{main}$ which is false exactly at the new points.
\end{enumerate} 

We define the desired function $\widehat{T}:\mathcal{L}\to\mathcal{L}$ in four steps:

\begin{definition}[Translation $\widehat{T}$]
\label{def: reduction SSL-S4S5}
\begin{enumerate}
\item
For every $\varphi\in\mathcal{L}$ let $\main \in AT$ by the alphabetically first
propositional variable that is not a subformula of $\varphi$.
\item
Recursively, we define a function $T:\mathcal{L}\to\mathcal{L}$ as follows:
	\[\begin{array}{lll}
		T(A)					&:=	& A	\smallskip\\
		T(\neg\psi)				&:=	& \neg T(\psi)\smallskip\\
		T((\psi_1\wedge\psi_2))	&:=	& (T(\psi_1)\wedge T(\psi_2))\smallskip\\
		T(K\psi)				&:=	& K\neg (\main \wedge \neg T(\psi))\smallskip\\
		T(\Box\psi)				&:=	& \Box\neg (\main \wedge \neg T(\psi))
	\end{array}\]
for all $A\in AT$ and for all $\psi,\psi_1,\psi_2\in\mathcal{L}$.
(Note that $K\neg (\main \wedge \neg T(\psi))$ is equivalent to $K(\main\to T(\psi))$
and that $\Box\neg (\main \wedge \neg T(\psi))$ is equivalent to $\Box(\main\to T(\psi))$.)
\item
For $\varphi\in\mathcal{L}$ we define
\[ persistent_{main} := 
		\bigwedge_{A\in AT\cap\subf(\varphi)}K(\Box(\main\to A)\vee\Box(\main\to\neg A)) . \]
\item
For $\varphi\in\mathcal{L}$ we define
a function $\widehat{T}:\mathcal{L}\to\mathcal{L}$ by
\[\widehat{T}(\varphi):= \main
	\wedge K \Box(\neg\main\to\Box\neg\main)  
	\wedge persistent_{main} 
	\wedge T(\varphi). \]
\end{enumerate}
\end{definition}

We claim that the function $\widehat{T}$ is indeed
a reduction function from the satisfiability problem
of $\SSL$ to the satisfiability problem of $\SfourSfive$.

\begin{proposition}
\label{prop:red-SSL-S4xS5}
The function $\widehat{T}:\mathcal{L} \to \mathcal{L}$
satisfies, for all $\varphi \in \mathcal{L}$,
\[\varphi \text{ is $\SSL$-satisfiable }
\iff \widehat{T}(\varphi) \text{ is $\SfourSfive$-satisfiable}. \]
\end{proposition}

The next section is dedicated to the proof of Proposition \ref{prop:red-SSL-S4xS5}. We treat both directions of the claimed equivalence separately.

\subsection{Correctness}

\begin{lemma}
\label{lemma:SSL-sat}
Let $\varphi\in\mathcal{L}$.
Then
\[
\varphi \text{ is $\SSL$-satisfiable } \Rightarrow \widehat{T}(\varphi) \text{ is $\SfourSfive$-satisfiable.}
\]
\end{lemma}

\begin{proof}
Let $\varphi\in\mathcal{L}$ be $\SSL$--satisfiable.
Then there are a cross axiom model $M=(W,\stackrel{\Diamond}{\to}, \stackrel{L}{\to},\sigma)$
and a point $w\in W$ such that $M,w\models\varphi$.
Let $C_i$, for $i\in I$, where $I$ is a suitable index set, be the $\stackrel{L}{\to}$-equivalence classes in $W$.
We construct an $\SfourSfive$-commutator model 
$M'=(W',\stackrel{\Diamond'}{\to}, \stackrel{L'}{\to},\sigma')$ for $\widehat{T}(\varphi)$ as follows.
Let $newpoint_i$ for $i\in I$ be pairwise different ``new'' points that are not elements of $W$.
We define 
\[W' :=	W\cup\{newpoint_i \mid i\in I\}. \]
We define the equivalence relation $\stackrel{L'}{\to}$ on $W'$
by demanding that the following sets $C'_i$ for $i\in I$
are the $\stackrel{L'}{\to}$-equivalence classes of $W'$:
\[C'_i := C_i\cup\{newpoint_i\} , \]
for $i\in I$.
Let $\stackrel{\Diamond}{\to}^{\stackrel{L}{\to}}$ 
be the relation on the set of clouds in $M$ induced by $\stackrel{\Diamond}{\to}$;
compare \cite[Definition 4.1]{HK2019-1}.
We define
\[\stackrel{\Diamond'}{\to}\quad :=\quad	\stackrel{\Diamond}{\to} 
										\cup\{(p,newpoint_j)\mid j \in I \text{ and } \exists i \text{ with }
										p\in C'_i \text{ and } C_i \stackrel{\Diamond}{\to}^{\stackrel{L}{\to}} C_j\} .
\]
\begin{figure}[h]
   \begin{center}		
   \includegraphics[width=0.4\linewidth]{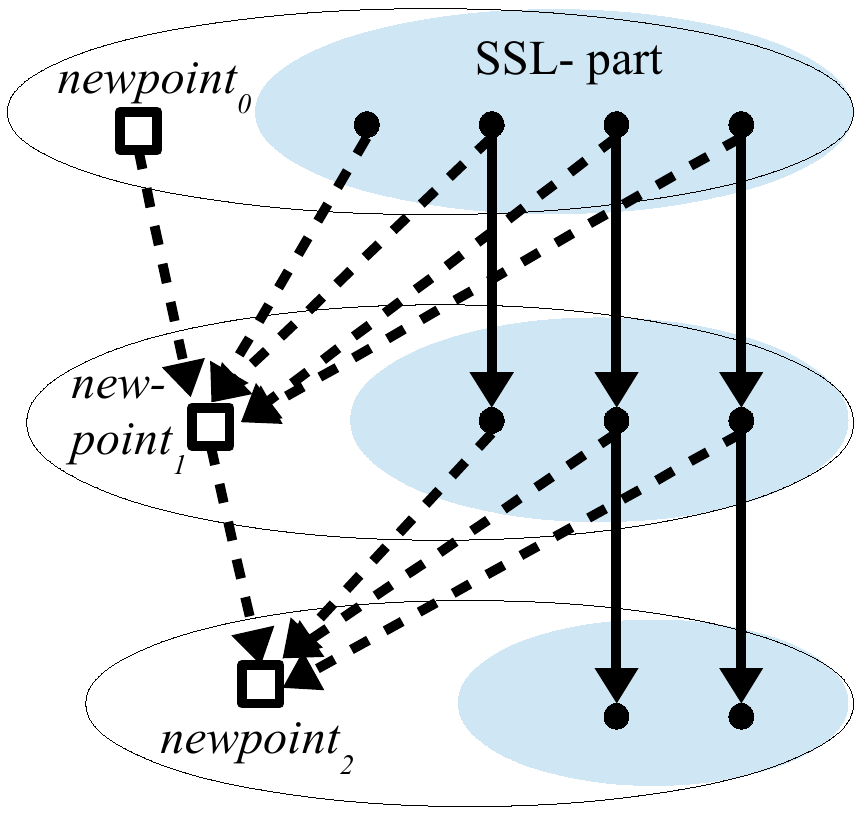}
	\end{center}
	\caption{Illustration of the $\sxs$-commutator model associated in the proof of Lemma~\ref{lemma:SSL-sat} with a cross axiom model.}
	\label{figure:newpoints}
\end{figure}

Finally we define
\begin{eqnarray*}
	\sigma'(\main)	&:=	&W, \\
	\sigma'(A)		&:=	& \sigma(A), \qquad \text{ for } A\in AT  \setminus\{\main\} . 
\end{eqnarray*}
We wish to show the following:
\begin{enumerate}
\item 
$M'$ is an $\SfourSfive$-commutator model,
\item
$M',w\models\widehat{T}(\varphi)$.
\end{enumerate}

We prove the first claim. 

Clearly, $\stackrel{L'}{\to}$ is an equivalence relation.

The relations $\stackrel{\Diamond}{\to}$ and $\stackrel{\Diamond}{\to}^{\stackrel{L}{\to}}$ are preorders 
(by assumption respectively by \cite[Corollary 4.4]{HK2019-1}),
hence, both of them are reflexive and transitive.
It is clear that this implies that the relation $\stackrel{\Diamond'}{\to}$ is reflexive as well.
For transitivity of $\stackrel{\Diamond'}{\to}$, assume that
$p \stackrel{\Diamond'}{\to} q$ and $q \stackrel{\Diamond'}{\to} r$.
If $r\in W$ then the definition of $\stackrel{\Diamond'}{\to}$ implies that also $q \in W$ and $p\in W$ and
$p \stackrel{\Diamond}{\to} q$ as well as $q \stackrel{\Diamond}{\to} r$, hence,
$p \stackrel{\Diamond}{\to} r$ by transitivity of $\stackrel{\Diamond}{\to}$.
This implies $p \stackrel{\Diamond'}{\to} r$.
If $r\not\in W$ then we let $i,j,k\in I$ be the indices with 
$p \in C'_i$, $q \in C'_j$, and $r \in C'_k$. Note that in this case $r=newpoint_k$.
We observe that
$p \stackrel{\Diamond'}{\to} q$ and $q \stackrel{\Diamond'}{\to} r$
imply $C_i \stackrel{\Diamond}{\to}^{\stackrel{L}{\to}} C_j$ and
$C_j \stackrel{\Diamond}{\to}^{\stackrel{L}{\to}} C_k$.
Transitivity of $\stackrel{\Diamond}{\to}^{\stackrel{L}{\to}}$ 
implies $C_i \stackrel{\Diamond}{\to}^{\stackrel{L}{\to}} C_k$, hence,
$p \stackrel{\Diamond'}{\to} newpoint_k =r$.
Thus, $\stackrel{\Diamond'}{\to}$ is transitive as well. 
We have shown that $\stackrel{\Diamond'}{\to}$ is a preorder.

Next, we wish to show left commutativity.
Let us consider some $p,q,r \in W'$ with $p \stackrel{\Diamond'}{\to} q$ and $q \stackrel{L'}{\to} r$. 
Let $i,j$ be the indices with $p \in C'_i$ and $q,r \in C'_j$.
It is sufficient to show that there exists some $s \in C'_i$ with $s \stackrel{\Diamond'}{\to} r$.
The condition $p \stackrel{\Diamond'}{\to} q$ implies
$C_i \stackrel{\Diamond}{\to}^{\stackrel{L}{\to}} C_j$.
If $r \in W$ then the left commutativity of $M$ gives us an $s\in C_i$
with $s \stackrel{\Diamond}{\to} r$, hence, with $s \stackrel{\Diamond'}{\to} r$.
If $r\not\in W$ then $r=newpoint_j$, and $s:=p$ does the job.

In order to prove right commutativity 
let us consider some $p,q,r \in W'$ with $p \stackrel{L'}{\to} q$ and $q \stackrel{\Diamond'}{\to} r$. 
Let $i,j$ be the indices with $p, q \in C'_i$ and $r \in C'_j$.
It is sufficient to show that there exists some $s \in C'_j$ with $p \stackrel{\Diamond'}{\to} s$.
The condition $q \stackrel{\Diamond'}{\to} r$ implies
$C_i \stackrel{\Diamond}{\to}^{\stackrel{L}{\to}} C_j$.
Hence, we obtain $p \stackrel{\Diamond'}{\to} newpoint_j$.
This proves the first claim, that $M'$ is an $\SfourSfive$-commutator model.\\
	
We come to the second claim, $M',w\models\widehat{T}(\varphi)$.\\
From the definition of $\sigma'(\main)$ we obtain
	\[M',w \models\main .\]
Exactly at the points in $\{newpoint_i \mid i \in I\}$ the propositional variable $\main$ is false.
From the fact that all $\stackrel{\Diamond'}{\to}$-successors of points in this set 
are elements of this set as well, we conclude
	\[M',w\models K\Box(\neg\main\to\Box\neg\main) .\]
Next, we observe that for all $v\in W$ and for all propositional variables $A\in\subf(\varphi)$
we have by definition of $\sigma'$
	\[M,v\models A \iff M',v \models A .\]
Since all literals are persistent in $M$ and $\main$ is true exactly at the points in $W \subseteq W'$
we obtain
	\[ M',w \models persistent_{main} . \]
Finally, we have to show $M',w\models T(\varphi)$.
By induction on the structure of $\psi$ we show the stronger assertion:	
	\[M,v\models\psi\iff M',v\models T(\psi) , \]
for all $v\in W$ and for all $\psi \in \subf(\varphi)$.
We distinguish the following cases:
\begin{itemize}
\item
Case $\psi \in AT$. We already mentioned that for all $v\in W$ and for all $A\in AT\cap\subf(\varphi)$ we have
$M,v\models A \iff M',v \models A$.
\item 
Case $\psi = \neg\chi$.
Then the following four assertions are equivalent, the second and the third by induction hypothesis:
(1) $M,v \models \psi$,
(2) $M,v \not\models \chi$,
(3) $M',v \not\models T(\chi)$,
(4) $M',v \models T(\psi)$.
\item
Case $\psi = (\chi_1\wedge\chi_2)$.
This case is treated in the same way.
\item 
Case $\psi = K\chi$.
We wish to show
\[ M,v\models K\chi\iff M',v\models K(\main\to T(\chi)) . \]
First, let us assume $M,v \models K\chi$.
Let us consider an arbitrary $v'\in W'$ such that $v\stackrel{L'}{\to}v'$.
It is sufficient to show that 
$$M',v' \models (\main\to T(\chi)).$$
If $v'\in W$ then we obtain $v \stackrel{L}{\to} v'$ and $M,v' \models \chi$.
By induction hypothesis we obtain $M',v'\models T(\chi)$, hence
$M',v' \models (\main\to T(\chi))$.
If $v'\not\in W$ then $M',v' \not\models \main$, hence in this case 
$M',v'\models (\main\to T(\chi))$ as well.

Now, for the other direction, let us assume $M',v\models K(\main\to T(\chi))$.
Consider some $u\in W$ with $v\stackrel{L}{\to} u$
(remember that this implies $v \stackrel{L'}{\to} u$).
It is sufficient to show that 
$$M,u \models \chi.$$
But for $u\in W$ we have $M',u \models \main$. The condition
$M',v\models K(\main \to T(\chi))$ implies $M',u\models (\main \to T(\chi))$.
We obtain $M',u \models T(\chi)$.
Finally, by induction hypothesis we obtain $M,u \models\chi$.
\item 
Case $\psi = \Box\chi$.
This case can be treated in the same way as the previous case.
In fact, it is sufficient to copy the argument for the case $\psi=K \chi$ and to replace
$K$ by $\Box$, $\stackrel{L}{\to}$ by $\stackrel{\Diamond}{\to}$, and
$\stackrel{L'}{\to}$ by $\stackrel{\Diamond'}{\to}$.
\qedhere
\end{itemize}	
\end{proof}

Now we turn to the other direction of Proposition~\ref{prop:red-SSL-S4xS5}.

\begin{lemma}
\label{lemma:S4xS5-sat}
Let $\varphi\in\mathcal{L}$. Then
\[\widehat{T}(\varphi) \text{ is $\SfourSfive$-satisfiable } \Rightarrow \varphi \text{ is $\SSL$-satisfiable.}\]
\end{lemma}

\begin{proof}
Let $\widehat{T}(\varphi)$ be $\SfourSfive$--satisfiable.
This implies that there exist an $\SfourSfive$--commutator model
$M'=(W',\stackrel{\Diamond'}{\to}, \stackrel{L'}{\to},\sigma')$
and a point $w\in W'$ such that $M',w\models\widehat{T}(\varphi)$, that is
\[ M',w\models \main
		\wedge K\Box(\neg\main\to\Box\neg\main)
		\wedge persistent_{main}
		\wedge T(\varphi). \]
We construct a cross axiom model 
$M:=(W,\stackrel{\Diamond}{\to}, \stackrel{L}{\to},\sigma)$ for $\varphi$ as follows. We construct $M$ as a rooted model, where all points are reachable from the point $w$:
\[W := \{v\in W' \mid M',v\models\main \text{ and }
           (\exists w')\,(w \stackrel{L'}{\to} w' \text{ and } w' \stackrel{\Diamond'}{\to} v)\} .\] 
We define the relations $\stackrel{L}{\to}$ and $\stackrel{\Diamond}{\to}$ on $W$ simply by
\begin{eqnarray*}
 \stackrel{L}{\to} &:= & \stackrel{L'}{\to} \cap \, (W\times W) , \\
 \stackrel{\Diamond}{\to} &:= & \stackrel{\Diamond'}{\to} \cap \, (W\times W) .
\end{eqnarray*} 
Finally,
\[\sigma(A) := \begin{cases}
		\sigma'(A) \cap W & \text{ if } A\in \subf(\varphi) \cup\{\main\}, \\
		\emptyset & \text{ if } A \not\in \subf(\varphi) \cup\{\main\},
		\end{cases}		\]
for all $A \in AT$.\\
First, we observe that $M',w\models \main$ 
and the reflexivity of $\stackrel{L'}{\to}$ and $\stackrel{\Diamond'}{\to}$ imply $w\in W$.
We wish to show the following:
\begin{enumerate}
\item
$M$ is a cross axiom model,
\item
$M,w\models\varphi$.
\end{enumerate}

We prove the first claim.

It is obvious that the relation $\stackrel{L}{\to}$ inherits 
reflexivity, symmetry and transitivity from $\stackrel{L'}{\to}$
and that the relation $\stackrel{\Diamond}{\to}$ inherits 
reflexivity and transitivity from $\stackrel{L'}{\to}$.
Thus, the relation $\stackrel{L}{\to}$
is an equivalence relation, and the relation
$\stackrel{\Diamond}{\to}$ is a preorder.

Next we wish to show that $M$ has the left commutativity property.
So let us consider points $p,q,r\in W$ with $p\stackrel{\Diamond}{\to}q$ and $q\stackrel{L}{\to}r$.
By definition of the relations $\stackrel{\Diamond}{\to}$ and $\stackrel{L}{\to}$
we obtain that also $p\stackrel{\Diamond'}{\to}q$ and $q\stackrel{L'}{\to}r$.
Since $M'$ has the left commutativity property there is some point $p'\in W'$
with $p\stackrel{L'}{\to}p'$ and $p'\stackrel{\Diamond'}{\to}r$.
The definition of $W$ and $p\in W$ imply that there exists some
$w' \in W'$ with $w \stackrel{L'}{\to}w'$ and $w' \stackrel{\Diamond'}{\to} p$.
The left commutativity of $M'$ and $w' \stackrel{\Diamond'}{\to} p$ as well as $p \stackrel{L'}{\to} p'$
imply that there exists some $w'' \in W'$ with
$w' \stackrel{L'}{\to} w''$ and $w'' \stackrel{\Diamond'}{\to} p'$.
So, we have $w \stackrel{L'}{\to} w''$ and $w'' \stackrel{\Diamond'}{\to} p'$.
In addition to that,
$M',w\models K \Box(\neg\main\to\Box\neg\main)$ implies
$M',w'' \models \Box(\neg\main\to\Box\neg\main)$, hence,
$M',p' \models \neg\main\to\Box\neg\main$. This, together with
$p' \stackrel{\Diamond'}{\to} r$ and $M',r \models \main$, implies
$M',p'\models \main$.
Hence $p'\in W$ and $p\stackrel{L}{\to}p'$ as well as $p'\stackrel{\Diamond}{\to}r$.
This shows that $M$ has the left commutativity property. 

Finally, we claim that propositional variables are persistent in $M$.
For all $A \in AT\setminus (\{\main\} \cup \subf(\varphi))$ we have
$\sigma(A)=\emptyset$. Hence, all
$A \in AT\setminus (\{\main\} \cup \subf(\varphi))$ are persistent in $M$.
Furthermore, we have $\sigma(\main)=W$. Hence, $\main$ is persistent in $M$ as well.
Let us consider $u,v \in W$ with $u \stackrel{\Diamond}{\to} v$.
Due to the definition of $W$ there exists some $w'\in W'$ with
$w \stackrel{L'}{\to} w'$ and $w' \stackrel{\Diamond'}{\to} u$.
Then, $M',w \models persistent_{main}$ implies,
for all $A \in AT\cap \subf(\varphi)$,
$M', w' \models \Box(\main\to A)\vee\Box(\main\to\neg A)$.
Note that $M',u \models \main$ and $M',v\models\main$.
So, $M',u \models A$ if, and only if, $M',v\models A$.
Due to $\sigma(A)=\sigma'(A)\cap W$ the same holds true with $M'$ replaced by $M$.
This shows that all $A \in AT\cap \subf(\varphi)$ are persistent in $M$.
We have shown that $M$ is a cross axiom model.

We come to the second claim, $M,w\models\varphi$.
Due to $M',w\models T(\varphi)$ it is sufficient to prove 
for all $v\in W$ and for all $\psi\in\subf(\varphi)$:
\[M',v\models T(\psi)\quad \iff \quad M,v\models\psi .\]
We show this by induction on the structure of $\psi$.
We distinguish the following cases:
\begin{itemize}
\item 
Case $\psi=A\in AT$.
In this case the claim is true due to the definition of $\sigma$.
\item 
Case $\psi=\neg\chi$ or $\psi = (\chi_1\wedge\chi_2)$. In both cases the claim follows directly from the induction hypothesis.
\item 
Case $\psi=K\chi$.
Let us first assume $M',v\models T(K\chi)$, that is $M',v\models K(\main\to T(\chi))$.
We wish to show $M,v \models K\chi$.
Consider an arbitrary $v'\in W$ with $v\stackrel{L}{\to}v'$.
It is sufficient to show $M,v' \models\chi$.
Note that the definition of $\stackrel{L}{\to}$ implies $v \stackrel{L'}{\to} v'$.
Hence, we have $M',v' \models (\main\to T(\chi))$.
Furthermore, due to $v' \in W$ we have $M',v'\models \main$.
We obtain $M',v' \models T(\chi)$.
By induction hypothesis we obtain $M,v' \models\chi$.
		
For the other direction let us assume that $M,v\models K\chi$.
We wish to show $M',v \models K(\main\to T(\chi))$.
Consider an arbitrary $v'\in W'$ with $v\stackrel{L'}{\to}v'$.
It is sufficient to show $M',v' \models (\main \to T(\chi))$.
If $v'\in W$ then $v\stackrel{L'}{\to}v'$ implies $v\stackrel{L}{\to}v'$,
and $M,v\models K\chi$ implies $M,v' \models \chi$.
By induction hypothesis we obtain $M',v' \models T(\chi)$, hence,
$M',v' \models (\main \to T(\chi))$.
If $v'\not\in W$ then we claim that $M',v' \not\models \main$, hence,
$M',v' \models (\main \to T(\chi))$.
Indeed, the assumption $v\in W$ implies that there exists some $w'\in W'$ with
$w \stackrel{L'}{\to} w'$ and $w' \stackrel{\Diamond'}{\to} v$. Together with $v\stackrel{L'}{\to}v'$
and left commutativity of $M'$ and transitivity of $\stackrel{L'}{\to}$ we can conclude that there exists some $w''\in W'$ with
$w \stackrel{L'}{\to} w''$ and $w'' \stackrel{\Diamond'}{\to} v'$.
Hence, $v'$ is reachable from $w$. By definition of $W$ the assumption $v' \not\in W$ indeed implies
$M',v' \not\models \main$.
\item 
Case $\psi=\Box\chi$.
This case can be treated similarly as the previous case.
The details are left to the reader.
%
\qedhere
\end{itemize}
\end{proof}

\subsection{LOGSPACE Computability of the Reduction Function}

In this subsection we show that the satisfiability problem of the logic $\SSL$ can be reduced in logarithmic space to the satisfiability problem of $\SfourSfive$.
We start with the following assertion.

\begin{lemma}
\label{lemma:L-ALOGTIME}
The language $\mathcal{L}$ of bimodal formulas is an element of $\ALOGTIME$.
\end{lemma}

Here, $\ALOGTIME$ is the set of all languages that can be 
decided in logarithmic time by an alternating Turing machine with ``random access'' to the input;
compare Buss~\cite[Page 124]{Buss87}, Ibarra, Jiang, and Rivakumar~\cite{IJR1988},
Clote~\cite[Def. 2.3]{Clote1999}.
Note that Buss~\cite[Pages 124, 125]{Buss87} has shown that a certain language of Boolean formulas is in $\ALOGTIME$. As our syntax of formulas is slightly different from the one used by Buss we cannot directly use his result. But Lemma~\ref{lemma:L-ALOGTIME} can be proved by arguments similar to the arguments used by Buss. Therefore, we omit the proof of Lemma~\ref{lemma:L-ALOGTIME}. Actually, we need only the following corollary.

\begin{corollary}
\label{cor:bimodal-formulas-in-logspace}
The language $\mathcal{L}$ of bimodal formulas can be decided in logarithmic space.
\end{corollary}

\begin{proof}
It is well-known that $\ALOGTIME$ is a subset of $\LOGSPACE$; see \cite[Page 601]{Clote1999}.
\end{proof}

Let $\Sigma:=\{(, ), \neg, \Box, K, \wedge, x , 0, 1\}$ be the alphabet over which
bimodal formulas are defined.
In order to prove the assertion we have to show that there is a 
logspace computable
function $\widetilde{T}:\Sigma^*\to \Sigma^*$
such that, for all $\varphi \in \Sigma^*$,
\[ (\varphi \in \mathcal{L} \text{ and }
   \varphi \text{ is $\SSL$-satisfiable})
   \iff 
   (\widetilde{T}(\varphi) \in \mathcal{L} \text{ and } 
   \widetilde{T}(\varphi) \text{ is $\SfourSfive$-satisfiable}) . \]
We define such a function $\widetilde{T}$ by formulating
an algorithm for computing it that works in logarithmic space.

So, let $\varphi \in \Sigma^*$ be the input string.
According to Corollary~\ref{cor:bimodal-formulas-in-logspace}
we can first check in logarithmic space whether $\varphi$
is a bimodal formula or not, that is, whether
$\varphi$ is an element of $\mathcal{L}$ or not.
If not then the algorithm outputs $\widetilde{T}(\varphi):=\wedge$
(which is certainly not a bimodal formula).
If, on the other hand, $\varphi$ is a bimodal formula
then we wish to compute and print $\widehat{T}(\varphi)$
as defined in Subsection~\ref{subsection:DefinitionReduction}.

First, we compute the smallest natural number $i$ such that
$x\bin(i)$ is not a subformula of $\varphi$.
We do this by starting with $j=0$, increasing $j$ step by step, and checking in 
each step whether $x\bin(j)$ is a subformula of $\varphi$. 
Note that a string $x\bin(j)$ for some $j\in\IN$
is a subformula of $\varphi$ if, and only if, 
there exists an occurrence of the string $x\bin(j)$ as a substring of $\varphi$
that is not followed by a $0$ or a $1$.
This algorithm works in logarithmic space because
$\varphi$ can contain only less than $|\varphi|$ many
subformulas of the form $x\bin(j)$.
Thus, we need to check whether $x\bin(j)$ is a subformula of $\varphi$
only for $j< |\varphi|$. For all these $j$ the binary
representation $\bin(j)$ can be stored in logarithmic space.
Finally, we can also store the string $\main:=x\bin(i)$
in logarithmic space.

Now we wish to compute and output $\widehat{T}(\varphi)$.
It is straightforward to print
$\main \wedge K \Box(\neg\main\to\Box\neg\main)$.
Next we wish to print $persistent_{main}$.
In order to do this we must identify all $j\in\IN$ such that
$x\bin(j)$ is a subformula of $\varphi$, and for each such $j$
we must print 
\[ K(\Box(\main\to x\bin(j))\vee\Box(\main\to\neg x\bin(j))) . \]
This can be done as follows. We read the string $\varphi$
from left to right. Whenever we read an $x$ we use two binary
counters in order to mark the beginning and the end of the subformula
$x\bin(j)$ that begins with this occurrence of $x$.
Then we check, again using binary counters, whether the same
subformula $x\bin(j)$ has appeared already further to the left in
$\varphi$. If it has then we just move on.
If it has not appeared before, then we print out
$K(\Box(\main\to x\bin(j))\vee\Box(\main\to\neg x\bin(j)))$,
and then we move on.
It is clear that all this can be done in logarithmic space.

Finally, we wish to output $T(\varphi)$.
In order to compute $T(\varphi)$ one has to replace
every occurrence of ``$K$'' by ``$K\neg (\main \wedge \neg$'',
and every occurrence of ``$\Box$'' by ``$\Box\neg (\main \wedge \neg$'',
and one has to add an additional closing bracket after each
subformula $K\psi$ and each subformula $\Box \psi$ of $\varphi$.
Besides that, all other symbols from $\varphi$ can simply be
copied. The only nontrivial part here is the addition of a closing
bracket after each occurrence of a subformula of the form
$K\psi$ or $\Box\psi$ of $\varphi$.
In order to do this, for each position in the string $\varphi$ one has
to count how many subformulas of the form $K\psi$ or
$\Box\psi$ of $\varphi$ end in this position
and then one has to print so many additional closing brackets.
So, how can one count how many subformulas of the form $K\psi$ or
$\Box\psi$ of $\varphi$ end in the current position
in the string $\varphi$?
If the symbol in this position is an element of
$\{(, \neg, \Box, K, \wedge, x \}$
then no subformula ends in this position.
The same is true if the symbol in this position is either $0$ or $1$
and this is followed by a bit $0$ or $1$ as well.
There are only the following possible cases for the last
symbol of a subformula.
\begin{itemize}
\item
If the symbol in the current position is a bit, so $0$ or $1$, and this is
not followed by a bit, then a subformula of the form 
$x\bin(j)$ for some $j$ ends in this position.
Then we move to the left of the corresponding occurrence of $x$
and count the number of occurrences of $K$ and $\Box$ until
we read a symbol not in $\{K, \Box,\neg\}$.
\item
If the symbol in the current position is a closing bracket $)$
then we go to the left step by step
until we have found the corresponding opening bracket $($.
This can be done by using a binary counter that is increased
by $1$ for each closing bracket and decreased by $1$ for each
opening bracket. One stops when this counter is back to its
initial value $0$.
Once we have found the corresponding opening bracket
we move to the left of it
and count the number of occurrences of $K$ and $\Box$ until
we read a symbol not in $\{K, \Box,\neg\}$.
\end{itemize}
By using binary counters all this can be done in logarithmic space.
This ends the description of the computation in logarithmic
space of the described reduction function $\widetilde{T}$.

\section{Reduction of $\sxs$ to $\kxs$ }
\label{section:S4xS5-K4xS5}

The $\EXPSPACE$-hardness of $\kxs$ follows from the $\EXPSPACE$-hardness of $\sxs$ and from the following result.

\begin{theorem}
\label{theorem:kxs-hard}
	The satisfiability problem of the bimodal logic $\sxs$ can be reduced in logarithmic space to the satisfiability problem of the bimodal logic $\kxs$.
\end{theorem}

We start with the definition of the reduction function $\widehat{T}$ translating 
bimodal formulas in the language $\mathcal{L}$ to bimodal formulas in $\mathcal{L}$
such that for all $\varphi \in \mathcal{L}$:
\[\varphi \text{ is $\sxs$-satisfiable } \iff \widehat{T}(\varphi) \text{ is $\kxs$-satisfiable.}\]

The problem that we face is that in general $\kxs$-models are not reflexive.
To handle this we add to the original formula $\varphi$ a formula that implies that all those instances of the reflexivity axiom scheme $\Box\psi\to\psi$ where $\Box\psi$ is a subformula of $\varphi$ must hold true in all reachable points.

\begin{definition}[Translation $\widehat{T}$]
	\label{def: reduction S4S5-K4S5}
	For $\varphi\in\mathcal{L}$ we define
	a function $\widehat{T}:\mathcal{L}\to\mathcal{L}$ by
	\[\widehat{T}(\varphi):= \varphi\wedge\bigwedge_{\Box\psi\;\in\,\subf(\varphi)}
	K((\Box\psi \to \psi) \wedge \Box(\Box\psi\to\psi))
	.\]
\end{definition}

We claim that the function $\widehat{T}$ is indeed
a reduction function from the satisfiability problem
of $\sxs$ to the satisfiability problem of $\kxs$.

\begin{proposition}
	\label{prop:red-S4xS5-K4xS5}
	The function $\widehat{T}:\mathcal{L} \to \mathcal{L}$
	satisfies, for all $\varphi \in \mathcal{L}$,
	\[\varphi \text{ is $\sxs$-satisfiable }
	\iff \widehat{T}(\varphi) \text{ is $\kxs$-satisfiable}. \]
\end{proposition}

\begin{proof}
	Let $\varphi$ be $\sxs$--satisfiable.
	There are an $\sxs$-commutator model $M=(W,\stackrel{\Diamond}{\to}, \stackrel{L}{\to},\sigma)$
	and a point $w\in W$ such that $M,w\models\varphi$. Since the relation $\stackrel{\Diamond}{\to}$ in $M$ is reflexive we have for all $w'\in W$ that $M,w'\models\bigwedge_{\Box\psi\;\in\,\subf(\varphi)} (\Box\psi\to\psi)$.
	Hence, 
	$M,w\models\bigwedge_{\Box\psi\;\in\,\subf(\varphi)}K((\Box\psi \to \psi) \wedge \Box(\Box\psi\to\psi))$,
	in other words, $M,w\models\widehat{T}(\varphi)$. Furthermore, $M$ satisfies the conditions for $\kxs$-commutator models because the relation $\stackrel{\Diamond}{\to}$ in $M$ is transitive as well. Hence, $\widehat{T}(\varphi)$ is $\kxs$-satisfiable.
	
	For the other direction of the equivalence let us assume that $\widehat{T}(\varphi)$ is $\kxs$-satisfiable.
This implies that there exist a $\kxs$--commutator model
$M'=(W',\stackrel{\Diamond'}{\to}, \stackrel{L'}{\to},\sigma')$
and a point $w\in W'$ such that $M',w\models\widehat{T}(\varphi)$, that is
\[ M',w\models \varphi\wedge\bigwedge_{\Box\psi\;\in\,\subf(\varphi)}K((\Box\psi \to \psi) \wedge \Box(\Box\psi\to\psi)). \]
We construct an $\sxs$--commutator model
$M:=(W,\stackrel{\Diamond}{\to}, \stackrel{L}{\to},\sigma)$ for $\varphi$ as follows.
We construct $M$ as a rooted model, where all points are reachable from the point $w$:
\[W := \{v\in W' \mid 
           w \stackrel{L'}{\to} v \text{ or } (\exists w')\,(w \stackrel{L'}{\to} w' \text{ and } w' \stackrel{\Diamond'}{\to} v)\} .\] 
We define the relations $\stackrel{L}{\to}$ and $\stackrel{\Diamond}{\to}$ on $W$ simply by
\begin{eqnarray*}
 \stackrel{L}{\to} &:= & \stackrel{L'}{\to} \cap \, (W\times W) , \\
 \stackrel{\Diamond}{\to} &:= & (\stackrel{\Diamond'}{\to} \cap \, (W\times W)) \cup\{(v,v) \mid v\in W\}.
\end{eqnarray*} 
Finally,
\[\sigma(A) := \sigma'(A) \cap W , \]
for all $A \in AT$.\\
It is clear that $w\in W$, and it is straightforward to see that $M$ is an $\sxs$-commutator model.
We claim that $M,w\models\varphi$. By induction on the structure of $\psi$ we show the stronger assertion:	
	\[M,v\models\psi\iff M',v\models \psi, \]
	for all $v\in W$ and for all $\psi \in \subf(\varphi)$.
	We distinguish the following cases:
	\begin{itemize}
	\item 
	Case $\psi=A\in AT$.
	In this case the claim follows directly from the definition of $\sigma$.
	\item 
	Case $\psi=\neg\chi$ or $\psi = (\chi_1\wedge\chi_2)$. In both cases the claim follows directly from the induction hypothesis.
	\item 
	Case $\psi=\Box\chi$. 
	Let us first assume $M',v\models \Box\chi$. We wish to show $M,v \models \Box\chi$.
	Consider an arbitrary $v'\in W$ with $v\stackrel{\Diamond}{\to}v'$. It is sufficient to show $M,v' \models\chi$.
	Note that the definition of $\stackrel{\Diamond}{\to}$ implies that $v \stackrel{\Diamond'}{\to} v'$ or $v=v'$.
	In the first case, $v \stackrel{\Diamond'}{\to} v'$, the assumption $M',v\models \Box\chi$ directly implies
	$M',v' \models \chi$. By induction hypothesis we obtain $M,v' \models\chi$.
	In the second case, $v=v'$, we use the fact that
   $M',w \models K ( (\Box \chi \to \chi) \wedge \Box ( \Box \chi \to \chi))$ and $v \in W$ imply
   $M',v \models (\Box \chi \to \chi)$. Together with $M',v\models \Box\chi$ this implies $M',v \models \chi$,
   hence, $M',v' \models\chi$.
	By induction hypothesis we obtain $M,v' \models\chi$ as well.
	
	For the other direction let us assume that $M,v\models \Box\chi$.
We wish to show $M',v \models \Box\chi$.
Consider an arbitrary $v'\in W'$ with $v\stackrel{\Diamond'}{\to}v'$.
It is sufficient to show $M',v' \models \chi$.
But from $v\in W$ and $v\stackrel{\Diamond'}{\to} v'$
we conclude $v' \in W$ and $v\stackrel{\Diamond}{\to}v'$.
Hence, $M,v\models \Box\chi$ implies $M,v' \models \chi$.
By induction hypothesis we obtain $M',v' \models \chi$.
	\item
	Case $\psi=K\chi$. 
	This case can be treated similarly. 
	For the direction ``$\Rightarrow$'' one needs to use left commutativity. The details are left to the reader.
%
\qedhere
	\end{itemize}
	\end{proof}

\begin{proof}[{Proof of Theorem~\ref{theorem:kxs-hard}}]
Let $\Sigma=:\{(, ), \neg, \Box, K, \wedge, x , 0, 1\}$ be the alphabet over which
bimodal formulas are defined.
In order to prove the assertion we have to show that there is a 
logspace computable
function $\widetilde{T}:\Sigma^*\to \Sigma^*$
such that, for all $\varphi \in \Sigma^*$,
\begin{eqnarray*}
 \lefteqn{(\varphi \in \mathcal{L} \text{ and }
   \varphi \text{ is $\sxs$-satisfiable})} && \\
   &\iff& 
   (\widetilde{T}(\varphi) \in \mathcal{L} \text{ and } 
   \widetilde{T}(\varphi) \text{ is $\kxs$-satisfiable}) . 
\end{eqnarray*}
We define such a function $\widetilde{T}$ by formulating
an algorithm for computing it that works in logarithmic space.

So, let $\varphi \in \Sigma^*$ be the input string.
According to Corollary~\ref{cor:bimodal-formulas-in-logspace}
we can first check in logarithmic space whether $\varphi$
is a bimodal formula or not, that is, whether
$\varphi$ is an element of $\mathcal{L}$ or not.
If not then the algorithm outputs $\widetilde{T}(\varphi):=\wedge$
(which is certainly not a bimodal formula).
If, on the other hand, $\varphi$ is a bimodal formula
then we wish to compute and print $\widetilde{T}(\varphi):=\widehat{T}(\varphi)$
as defined in Definition~\ref{def: reduction S4S5-K4S5}.
The algorithm that prints $\widehat{T}(\varphi)$ works as follows.
\begin{enumerate}
	\item It prints $\varphi$.
	\item For each $\Box\psi\in\subf(\varphi)$ it prints the string 
	$$\wedge K((\Box\psi\to\psi)\wedge\Box(\Box\psi\to\psi)).$$
\end{enumerate}
Of course, for the second part, for every occurrence of the symbol $\Box$ in $\varphi$ one has to determine
the uniquely determined formula $\psi$ beginning in $\varphi$ immediately to the right of this occurrence of $\Box$.
Then one has to print the string above.
All this can be done using several binary counters.
First one sets a binary counter called $\emph{StartOfPsi}$
to the value of the position to the right of the current occurrence of $\Box$.
In order to compute the correct value of a binary counter $\emph{EndOfPsi}$ that is supposed to be the position of the
rightmost symbol in $\psi$ one proceeds as follows. 
Starting from the position $\emph{StartOfPsi}$ one reads the given string from left to right.
As long as the read symbol is $\neg$ or $\Box$ or $K$ one continues reading.
At some stage one will either read an $x$ or an opening bracket $($.
If one reads an $x$ then $\emph{EndOfPsi}$ is set to the position of the rightmost bit, that is, the rightmost $0$ or $1$,
such that between this symbol and the just read occurrence of $x$ there are only bits.
If one reads an opening bracket then $\emph{EndOfPsi}$ is set
to the position of the first closing bracket to the right of this opening bracket
such that the string between these two brackets contains as many closing as opening brackets.
Note that all this can be done in logarithmic space.

Finally, using the two binary counters $\emph{StartOfPsi}$ and $\emph{EndOfPsi}$
containing the positions of the first and the last symbol of $\psi$ it is clear that one can print the string
``$\wedge K((\Box\psi\to\psi)\wedge\Box(\Box\psi\to\psi))$'', again using additional binary counters that use only 
logarithmic space.
This ends the description of the computation in logarithmic space of the described reduction function $\widetilde{T}$.
\end{proof}

\bibliographystyle{abbrv}
\bibliography{dis}

\newpage
\appendix

\section[Reduction of ATMs Working in Exponential Time to $\sxs$]
{Reduction of Alternating Turing Machines Working in Exponential Time to $\sxs$}
\label{section:ATMs-S4S5}
In Section~\ref{section:ATMs-SSL} we have shown that any language recognized
by an Alternating Turing Machine working in exponential time can be reduced in logarithmic space to the satisfiability problem of $\ssl$. This shows that the satisfiability problem of $\ssl$ is $\EXPSPACE$-hard under logspace reduction. In Section~\ref{section:SSL-S4xS5} we have shown that the satisfiability problem of $\ssl$ can be reduced in logarithmic space to the satisfiability problem of $\sxs$. This proves the following theorem.

\begin{theorem}
\label{theorem:S4S5-EXPSPACE-hard}
The satisfiability problem of $\sxs$ is $\EXPSPACE$-hard under logarithmic space reduction.
\end{theorem}

In this appendix we wish to given an alternative, ``direct'' proof of this theorem by showing directly that any language $L$ recognized by an Alternating Turing Machine working in exponential time can be reduced in logarithmic space to the satisfiability problem of $\sxs$.
The proof is quite similar to the reduction of Alternating Turing Machines working in exponential time to the satisfiability problem of $\ssl$ presented in Section~\ref{section:ATMs-SSL}.
But, there are also some important differences between the two reductions. On the one hand, we are going to use certain ``shared variables'' as well, but the mechanism of shared variables in $\sxs$ is much easier than in $\ssl$. On the other hand, the fact that in $\sxs$ the left commutativity property and the right commutativity property hold true (while in $\ssl$ only the left commutativity property has to hold true) causes problems that were not present in our treatment of $\ssl$ in Section~\ref{section:ATMs-SSL}. Similarly as in that section, for $\sxs$ we will model nodes in an accepting tree of an Alternating Turing Machine by clouds in an $\sxs$-product model. But the commutativity properties have the consequence that for every point in every cloud there exists a copy in every other cloud. The result is that in a cloud modeling a certain node on some path in an accepting tree there are also points that come from other nodes on other computation paths, perhaps even with the same time stamp. This makes the isolation of the correct points carrying the required information (in particular the information about the symbol to be read at a certain time on some computation path) more difficult. For details see Subsection~\ref{subsection:definition-ATMs-S4XS5} where the reduction function is defined and explained.

In the first of the following five subsections we present an implementation of a binary counter in $\sxs$, similar to the implementation of a binary counter in $\ssl$ in Subsection~\ref{subsection:counter}.
In the second subsection we give an outline of the definition of the reduction function and then define the reduction function $f_\sxs$ formally.
In the third subsection we show that the reduction function $f_\sxs$ can be computed in logarithmic space.
The final two subsections are devoted to the correctness proof of the reduction.
First we show that in the case $w\in L$ the formula
$f_\sxs(w)$ is $\sxs$-satisfiable by explicitly 
constructing an $\sxs$-product model for $f_\sxs(w)$.
In the last section we show that if $f_\sxs(w)$ is 
$\sxs$-satisfiable then $w$ is an element of $L$.

\subsection{Binary Counters in $\sxs$}
\label{subsection:counter-S4XS5}

Fix some natural number $n\geq 1$. We wish to implement in $\sxs$
a binary $n$-bit counter that counts from $0$ to $2^n-1$.
That means, we wish to construct an $\sxs$-satisfiable formula with the property
that any model of it contains a sequence of pairwise distinct points
$p_0,\ldots,p_{2^n-1}$ such that, for each $i\in\{0,\ldots,2^n-1\}$,
at the point $p_i$ the number $i$ is stored in binary form in a certain way.

Additionally, we need to copy and transmit various bits of information, 
both in the $\stackrel{\Diamond}{\to}$-direction 
as well as in the $\stackrel{L}{\to}$-direction 
This will be done by two kinds of formulas.
\begin{itemize}
\item
On the one hand, we need formulas that have the same truth value in the vertical ($\stackrel{\Diamond}{\to}$) direction but can change their truth values in the horizontal ($\stackrel{L}{\to}$) direction. In the logic $\sxs$ we can force certain propositional variables to be persistent by a suitable formula.
\item
On the other hand,  we need formulas that have the same truth value in the horizontal ($\stackrel{L}{\to}$) direction but can change their truth values in the vertical ($\stackrel{\Diamond}{\to}$) direction. Such formulas will be called {\em shared variables}.
In the logic $\sxs$ they can be defined as follows. For $i\in\IN$ let $A_i$ be special propositional variables.
The {\em shared variables $\alpha_i$} are defined as
\[\alpha_i:=LA_i . \]
\end{itemize}
Note that $\neg\alpha_i \equiv K\neg A_i$.

In the following we use the abbreviations introduced above, in Table~\ref{table:abbrev1}, and in Table~\ref{table:abbrev3}.
\begin{table}
		\caption{Some abbreviations for logical formulas,
			where $\underline{F}=(F_{l-1},\ldots,F_0)$ and $\underline{G}=(G_{l-1},\ldots,G_0)$
			are vectors of formulas.
			As usual, an empty conjunction like $\bigwedge_{h=0}^{-1} F_h$
			can be replaced by any propositional formula that is true always.}
		\label{table:abbrev3}
		\renewcommand{\arraystretch}{1.5}
\begin{center}
		\begin{tabular}{|l|l|l|}
			\hline
			For & the following &is an abbreviation\\[-3mm]
			&expression & of the following formula  \\[-1mm]
			\hline\hline
			$l\geq 1, k\geq -1$ & $\mathrm{persistent}(\underline{F},>k)$ & $\bigwedge_{h=k+1}^{l-1} K(\Box F_h \vee \Box \neg F_h)$ \\ \hline			
			$l\geq 1$ & $\mathrm{persistent}(\underline{F})$ & $\mathrm{persistent}(\underline{F},>-1)$ \\ \hline
		\end{tabular}
\end{center}
\end{table}

The idea of the construction is the same as in Subsection~\ref{subsection:counter}.
\begin{itemize}
	\item
	We store the counter values in a vector $\underline{\alpha}:=\alpha_{n-1},\ldots,\alpha_0$ of shared variables. To this end we embed the sequence $p_0,\ldots,p_{2^n-1}$ of points in a sequence of clouds $\mathit{C}_0,\ldots,\mathit{C}_{2^{n-1}}$ such that the cloud $\mathit{C}_i$ contains the point $p_i$ and such that the vector $\underline{\alpha}$ of shared variables satisfied at $p_i$ (and hence at all points in $\mathit{C}_i$) encodes the number $i$.
	\item
	Let $i\leq 2^n-1$ be the number encoded by $\underline{\alpha}$. If $\underline{\alpha}$ contains no $0$ then $i$ has reached its highest posible value, the number $2^n-1$. Otherwise let $k$ be the position of the rightmost $0$. We determine the position $k$ with the aid of the formula $\mathrm{rightmost\_zero}(\alpha,k)$. In order to increment the counter we have to keep all $\alpha_j$ at positions $j>k$ unchanged and to switch all $\alpha_j$ at positions $j\leq k$. We do this in two steps:
	\begin{enumerate}
		\item 
		First me make an $\stackrel{L}{\to}$-step from the point $p_i$ to a point $p'_i$ where we store the number $i+1$ in a vector $\underline{X}:=X_{n-1},\ldots,X_0$ of usual propositional variables by demanding that
		$$p'_i\models (\underline{X}=\underline{\alpha},>k) \wedge 
		\mathrm{rightmost\_one}(\underline{X},k).$$
		\item 
		Then we make a $\stackrel{\Diamond}{\to}$-step from the point $p'_i$ to a 
		point $p_{i+1}$ in the cloud $\mathit{C}_{i+1}$ and demand that 
		$$p_{i+1}\models (\underline{X}=\underline{\alpha}).$$
		Note that the value of $\underline{X}$ is copied from $p'_i$ to its $\stackrel{\Diamond}{\to}$-successor $p_{i+1}$ because by the formula
		$$\mathrm{persistent}(\underline{X})$$
		we force the vector $\underline{X}$ of propositional variables to be persistent.
	\end{enumerate}
 	Altogether we demand that for the number $k$
		$$p_i\models L\bigl((\underline{X}=\underline{\alpha},>k)
			\wedge \mathrm{rightmost\_one}(\underline{X},k) \wedge
			\Diamond(\underline{X}=\underline{\alpha}) \bigr).$$
	\item
	Additionally we need a formula to ensure that the starting value is $0$, that is we demand 
	$$p_0\models(\underline{\alpha}=\mathrm{bin}_n(0)).$$
\end{itemize}

We now define the complete counter formula, for $n>0$.
\begin{eqnarray*}
	\mathrm{\mathrm{counter}}_{\sxs,n}  &:= &
	\mathrm{persistent}(\underline{X}) \wedge 
	(\underline{\alpha}=\mathrm{bin}_n(0)) 
	\wedge K \,\Box  \Biggl(
	\bigwedge_{k=0}^{n-1}\Biggl( \mathrm{rightmost\_zero}(\underline{\alpha},k) \rightarrow	\\
	&& L \biggl( (\underline{X}=\underline{\alpha},>k)\wedge
	\mathrm{rightmost\_one}(\underline{X},k) \wedge 	
	\Diamond(\underline{X}=\underline{\alpha}) \biggr) \Biggr)\Biggr).
\end{eqnarray*}

\begin{proposition}
	\label{prop: counter-S4XS5}
	\begin{enumerate}
		\item 
		\sloppy{
		For all $n\in\IN\setminus \{0\}$, the formula $\mathrm{counter}_{\sxs,n}$ is $\sxs$-satisfiable.
		}
		\item 
		For all $n\in\IN\setminus \{0\}$, for every $\sxs$-commutator model of $\mathrm{counter}_{\sxs,n}$ and for every point $p_0$ in this model with $p_0 \models \mathrm{counter}_{\ssl,n}$
		there exist a sequence of $2^{n}-1$ points $p_1,p_2,\ldots,p_{2^n-1}$
		and a sequence of $2^{n}-1$ points $p_0',p_1',\ldots,p'_{2^n-2}$
		such that 
		\begin{itemize}			
			\item
			for $0\leq i\leq 2^n-1$, \quad 
			$p_i\models (\underline{\alpha}=\mathrm{bin}_n(i))$,
			\item
			for $0\leq i\leq 2^n-2$, \quad $p_i\stackrel{L}{\to}p'_{i}$ and $p'_i\stackrel{\Diamond}{\to}p_{i+1}$ and $p'_{i}\models (\underline{X}=\mathrm{bin}_n(i+1))$.
		\end{itemize}
	\end{enumerate}
\end{proposition}

\begin{proof}
	For the following let us fix some $n>0$.
	\begin{enumerate}
		\item
		\sloppy{
		We construct an $\sxs$-product model $M$ with a point $(0,0)$ in $M$ such that 
		$M, (0,0) \models \mathrm{counter}_{\sxs,n}$ as follows; see Figure \ref{figure:S4S5}.
		We define an S4-frame $(W_1,R_\Diamond)$ by
		}
		\[ W_1 := \{0,\ldots,2^n-1\} \quad
		\text{ and, for $i,i' \in W_1$, } i R_\Diamond i' :\iff i \leq i' . \]
		We define an S5-frame $(W_2,R_L)$ by
		\[ W_2 := \{0,\ldots,2^n-1\} \quad
		\text{ and } \quad R_L := W_2 \times W_2 . \]
		Then the product frame $(W, \stackrel{\Diamond}{\to}, \stackrel{L}{\to})$
		with $W:=W_1\times W_2$ and 
		with $\stackrel{\Diamond}{\to}$ and $\stackrel{L}{\to}$ defined as in
		\cite[Definition 2.2]{HK2019-1} is an $\sxs$-frame.
		We define the valuation $\sigma$ by
		\[\begin{array}{lll}
		\sigma(A_k) &:= & \{(i,j) ~:~ i,j \in \{0,\ldots,2^n-1\} \text{ and } k 
		\in \mathrm{Ones}(i) \} , \\
		\sigma(X_k) &:= & \{(i,j) ~:~ i,j \in \{0,\ldots,2^n-1\} \text{ and } k 
		\in \mathrm{Ones}(j)\} ,
		\end{array}\]
		for $k\in\{0,\ldots,n-1\}$.
		This implies for all $i,j\in \{0,\ldots,2^n-1\}$
		$$ (i,j) \models (\underline{\alpha}=\mathrm{bin}_n(i))
		\quad \text{and} \quad (i,j) \models (\underline{X}=\mathrm{bin}_n(j)). $$
		\begin{figure}
			\begin{center}
				\includegraphics[width=0.9\linewidth]{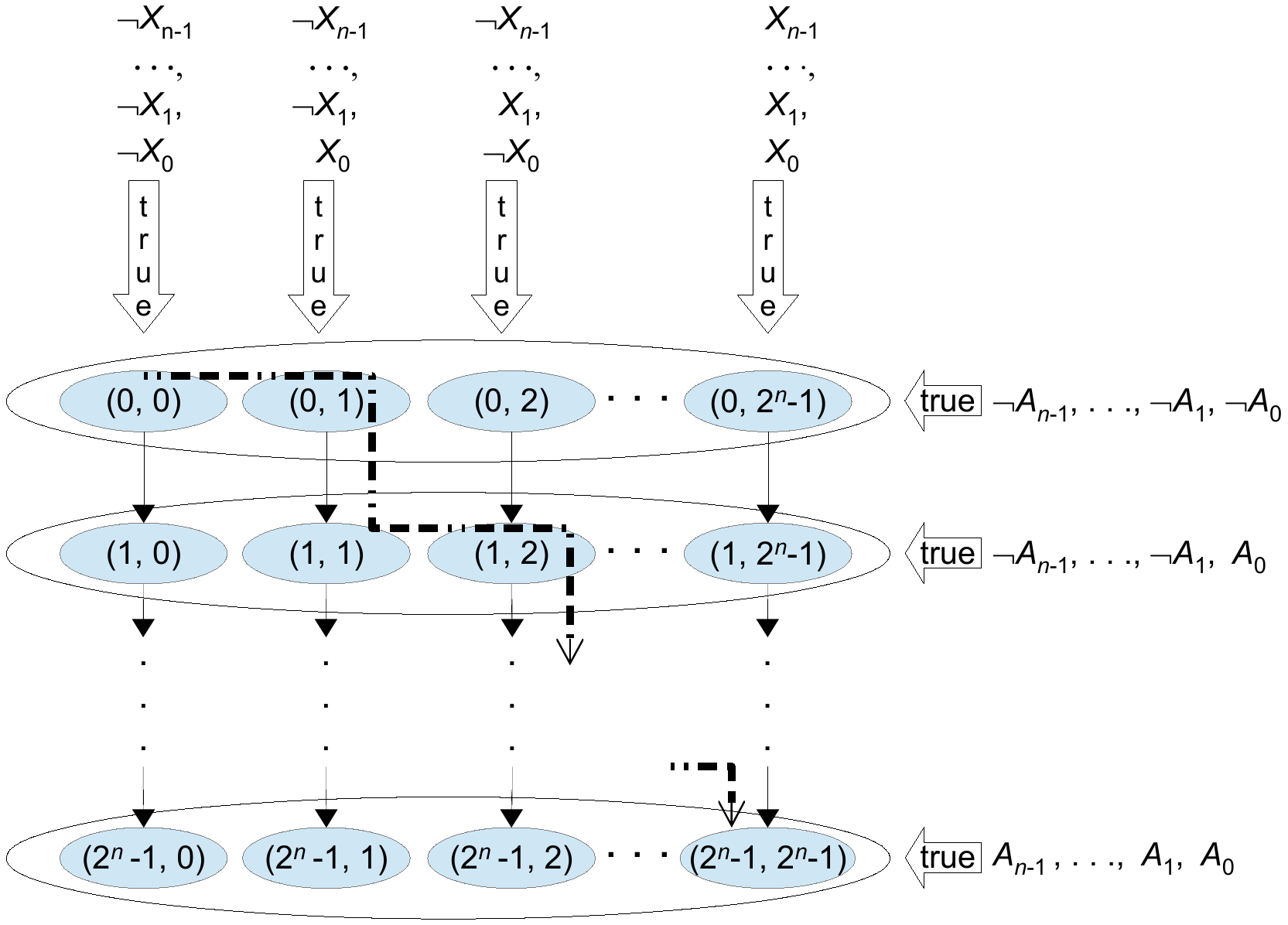}
			\end{center}	
			\caption{An $\sxs$-product model of the formula $\mathrm{counter}_{\sxs,n}$.}
		\label{figure:S4S5}
		\end{figure}	
		We claim
		$$(0,0) \models \mathrm{counter}_{\sxs,n}.$$ 
		Indeed, it is clear that the propositional variables $X_k$ are persistent, hence, we have
		$$(0,0) \models \mathrm{persistent}(\underline{X}).$$
		It is also clear that 
		$$(0,0) \models (\underline{\alpha}=\mathrm{bin}_n(0)).$$
		Let us assume that for some $(i,j) \in W$ and some $k\in\{0,\ldots,n-1\}$ we have
		\[ (i,j) \models \mathrm{rightmost\_zero}(\underline{\alpha},k) . \]
		It is sufficient to show that
		\[ (i,j) \models L \biggl( (\underline{X}=\underline{\alpha},>k)\wedge
		\mathrm{rightmost\_one}(\underline{X},k) \wedge 	
		\Diamond(\underline{X}=\underline{\alpha}) \biggr) . \]
		Indeed, $(i,j) \models \mathrm{rightmost\_zero}(\underline{\alpha},k)$ 
		implies $\{0,\ldots,n-1\}\setminus \mathrm{Ones}(i) \neq \emptyset$ and
		$k = \min (\{0,\ldots,n-1\}\setminus \mathrm{Ones}(i))$.
		Note that this implies $i<2^n-1$, hence,
		$(i,i+1)\in W$ and $(i+1,i+1)\in W$.
		In view of $(i,j)\stackrel{L}{\to}(i,i+1)\stackrel{\Diamond}{\to}(i+1,i+1)$ it is sufficient to show
	 	$$(i,i+1)\models(\underline{X}=\underline{\alpha},>k)\wedge
	 	\mathrm{rightmost\_one}(\underline{X},k)$$
	 	and 
	 	$$(i+1,i+1)\models (\underline{X}=\underline{\alpha}) .$$
	 	Both claims follow from the facts that $k=\min(\{0,\ldots,n-1\}\setminus\mathrm{Ones}(i))$ and that every point $(x,y)\in W$ satisfies the formulas $(\underline{\alpha}=\mathrm{bin}_n(x))$ and
	 	$(\underline{X}=\mathrm{bin}_n(y))$.
	\item
	The proof for $\sxs$-commutator models of the formula $\mathrm{counter}_{\sxs,n}$ is very similar to the 
	proof for cross axiom models of the formula $\mathrm{counter}_{\SSL,n}$ in Proposition~\ref{prop: counter}.
	For completeness sake we explicate it in detail.
	Suppose there are an $\sxs$-commutator model $M$ of the formula $\mathrm{counter}_{\sxs,n}$ and a point $p_0\in M$ with 
	$M,p_0\models \mathrm{counter}_{\sxs,n}$.
	We show by induction that the claimed sequences of points 
	$p_1,\ldots p_{2^n-1}$ and $p'_0,\ldots p'_{2^n-2}$
	with the claimed properties exist.
	In addition, we show that there exist points $t_i$ with 
	$p_0 \stackrel{L}{\to} t_i$ and $t_i \stackrel{\Diamond}{\to} p_i$, for $1\leq i \leq 2^n-1$.
	Note that the sequences $p_0,\ldots p_{2^n-1}$ and $p'_0,\ldots p'_{2^n-2}$ are supposed to form a ``staircase'' as in Figure~\ref{figure:staircase}.
	By definition,
	$p_0\models \mathrm{persistent}(\underline{X}) \wedge (\underline{\alpha}=\mathrm{bin}_n(0))$.
	By induction hypothesis, let us assume that for some $m$ with $0 \leq m<2^n-1$
	there exist $p_1,\ldots,p_m$ and $p'_0,\ldots,p_{m-1}'$ with
	\[ p_0\, \stackrel{L}{\to}\, p'_0\, \stackrel{\Diamond}{\to}\, p_1\, \stackrel{L}{\to} \ldots \stackrel{L}{\to}\, p'_{m-1}\, \stackrel{\Diamond}{\to}\, p_m ,  \]
	with
		\[ p_i\models (\underline{\alpha}=\mathrm{bin}_n(i)) , \]
		for $0\leq i \leq m$, and with
		\[ p'_i\models (\underline{X}=\mathrm{bin}_n(i+1)) , \]
		for $0\leq i < m$, and 
		that there are $t_i$ with $p_0 \stackrel{L}{\to} t_i$ and $t_i \stackrel{\Diamond}{\to} p_i$, for $1\leq i \leq m$.
		Since $m<2^n-1$, the set $\{0,\ldots,n-1\}\setminus\mathrm{Ones}(m)$ is nonempty.
		With $k:= \min(\{0,\ldots,n-1\}\setminus\mathrm{Ones}(m))$ we have
		\[ p_m\models \mathrm{rightmost\_zero}(\underline{\alpha},k). \]
		Due to $p_0\models \mathrm{counter}_{\sxs,n}$ as well as $p_0 \stackrel{L}{\to} t_m \stackrel{\Diamond}{\to} p_m$
		this implies
		\[	p_m \models 
			L \bigl( (\underline{X}=\underline{\alpha},>k) 
			\wedge \mathrm{rightmost\_one}(\underline{X},k) \\
			\wedge \Diamond(\underline{X}=\underline{\alpha})  \bigr) .		
		\]
		Thus, there must exist points $p'_m$ and $p_{m+1}$ satisfying
		$p_m\stackrel{L}{\to} p'_m \stackrel{\Diamond}{\to}p_{m+1}$ as well as 
		\[ p'_m \models(\underline{X}=\underline{\alpha},>k) 
		\wedge \mathrm{rightmost\_one}(\underline{X},k) \]		
		and
		\[ p_{m+1} \models  (\underline{X}=\underline{\alpha}) . \]
		We have to show
		\[ p'_m \models(\underline{X}=\mathrm{bin}_n(m+1)) \]
		and
		\[ p_{m+1} \models B\wedge(\underline{\alpha}=\mathrm{bin}_n(m+1)) . \]
		Due to the fact that $p'_m$ is an element of the same cloud as $p_m$, and $\underline{\alpha}$ has the same value in all points in a cloud we obtain
		\[ p'_m \models (\underline{\alpha}=\mathrm{bin}_n(m)) .\] Together with 
		\[ p'_m \models(\underline{X}=\underline{\alpha},>k) 
		\wedge \mathrm{rightmost\_one}(\underline{X},k) \]	
		this implies
		\[ p'_m \models(\underline{X}=\mathrm{bin}_n(m+1)) \]
		(the values of the leading bits $\alpha_{n-1},\ldots,\alpha_{k+1}$ of $\underline{\alpha}$ are copied to the leading bits $X_{n-1},\ldots,X_{k+1}$ and the other bits are defined explicitly by $p'_{m} \models  \mathrm{rightmost\_one}(\underline{X},k)$ so that the binary value of $\underline{X}$ is $m+1$).
		From $p_0 \models \mathrm{persistent}(\underline{X})$ as well as
		$p_0 \stackrel{L}{\to} t_m \stackrel{\Diamond}{\to} p'_m \stackrel{\Diamond}{\to} p_{m+1}$ we obtain
		$p_{m+1} \models (\underline{X}=\mathrm{bin}_n(m+1))$.
		Using $p_{m+1} \models (\underline{X}=\underline{\alpha})$
		we obtain
		$p_{m+1} \models (\underline{\alpha}=\mathrm{bin}_n(m+1))$.
		Finally, the left commutativity property applied to $t_m \stackrel{\Diamond}{\to} p_m$ and $p_m \stackrel{L}{\to} p'_m$
		implies that there exists a point $t_{m+1}$ with $t_m \stackrel{L}{\to} t_{m+1}$ and $t_{m+1} \stackrel{\Diamond}{\to} p'_m$.
		Using additionally $p_0 \stackrel{L}{\to} t_m$ and $p'_m \stackrel{\Diamond}{\to} p_{m+1}$ we obtain
		$p_0 \stackrel{L}{\to} t_{m+1}$ and $t_{m+1} \stackrel{\Diamond}{\to} p_{m+1}$.
		This ends the proof of the second assertion.	
	\qedhere
	\end{enumerate}
\end{proof}

\subsection{Construction of the Formula}
\label{subsection:definition-ATMs-S4XS5}

Let $L \in \EXPSPACE$ be an arbitrary language
over some alphabet $\Sigma$, that is, $L \subseteq \Sigma^*$.
We are going to show that there is a logspace computable
function $f_\sxs$ mapping strings to strings
such that, for any $w\in\Sigma^*$, 
\begin{itemize}
\item
$f_\sxs(w)$ is a bimodal formula and 
\item
$f_\sxs(w)$ is $\sxs$-satisfiable if, and only if, $w \in L$.
\end{itemize}
Once we have shown this, we have shown the result.

In order to define this desired reduction function $f_\sxs$,
we are going to make use of an Alternating Turing Machine for $L$.
Since $\EXPSPACE = \mathrm{AEXPTIME}$, there
exist an Alternating Turing Machine
$M = (Q, \Sigma, \Gamma, q_0, \delta)$ 
and a univariate polynomial $p$ such that
$M$ accepts $L$, that is, $L(M)=L$, and such that the time used by $M$
on arbitrary input of length $n$ is bounded by $2^{p(n)}-1$.
We can assume without loss of generality
$Q=\{0,\ldots,|Q|-1\}$, $\Gamma=\{0,\ldots,|\Gamma|-1\}$, that 
the coefficients of the polynomial $p$ are natural numbers and that, for all $n\in\IN$,
we have $p(n) \geq n$ and $p(n)\geq 1$.
In the following, whenever we have fixed some $n\in\IN$, we set
\[ N:=p(n) . \]
Let us consider an input string $w \in\Sigma^n$ of length $n$, for some $n \in\IN$,
and let us sketch the main idea of the construction of the formula
$f_\sxs(w)$.
The formula $f_\sxs(w)$ will describe the possible computations
of $M$ on input $w$ in the following sense: any $\sxs$-product model
of $f_\sxs(w)$ will essentially contain an accepting tree of $M$ on input $w$,
and if there exists an accepting tree of $M$ on input $w$
then one can turn this into an $\sxs$-product model of $f_\sxs(w)$.
In such a model, any node in an accepting tree 
of $M$ on input $w$ will be modeled by a cloud
(that is, by an $\stackrel{L}{\to}$-equivalence class) in which certain shared variables
(we use the notion ``shared variables'' in the same sense as in 
Subsection~\ref{subsection:counter-S4XS5})
will have values that describe the data of the computation node
that are important in this computation step. Which data are these?
First of all, we need the time of the computation node.
We assume that the computation starts with the initial configuration
of $M$ on input $w$ at time $0$.
Since the  ATM $M$ needs at most $2^N-1$ time steps,
we can store the time of each computation node in a binary counter
counting from $0$ to $2^N-1$. 
Since during each time step at most one additional cell
either to the right or to the left of the previous cell can be visited,
we can describe  any configuration reachable during a computation of $M$
on input $w$ by the following data:
{\setlength{\leftmargini}{2.15em}\begin{itemize}
\item
the state $q \in Q$ of the configuration,
\item
the current content of the tape, given by a string in
$\Gamma^{2 \cdot (2^N-1)+1} = \Gamma^{2^{N+1}-1}$,
\item
the current position of the tape head,
given by a number in $\{0,\ldots,2^{N+1}-2\}$.
\end{itemize}}
We assume that in the initial configuration on input $w$ the
tape content is $\#^{2^N}w\#^{2^N-1-n}$ 
(remember that we use $\#$ for the blank symbol) and that
the tape head scans the blank $\#$ to the left of the first symbol of $w$,
that is, the position of the tape head is $2^{N}-1$.
If a cloud in an $\sxs$-product model of $f_\sxs(w)$ 
describes a computation node of $M$ on input $w$
then in this cloud the following shared variables
will have the following values:
\begin{itemize}
\item
a vector $\underline{\alpha}^\mathrm{time}=
(\alpha^\mathrm{time}_{N-1},\ldots,\alpha^\mathrm{time}_0)$ giving in binary
the current time of the computation,
\item
a vector $\underline{\alpha}^\mathrm{pos}=
(\alpha^\mathrm{pos}_{N},\ldots,\alpha^\mathrm{pos}_0)$ giving in binary
the current position of the tape head,
\item
a vector $\underline{\alpha}^\mathrm{state}=
(\alpha^\mathrm{state}_0,\ldots,\alpha^\mathrm{state}_{|Q|-1})$ giving in unary
the current state of the computation
(here ``unary'' means: exactly one of the shared variables
$\alpha^\mathrm{state}_i$ will be true,
namely the one with $i$ being the current state),
\item
a vector $\underline{\alpha}^\mathrm{read}=
(\alpha^\mathrm{read}_0,\ldots,\alpha^\mathrm{read}_{|\Gamma|-1})$
giving in unary the symbol in the current cell
(here ``unary'' means: exactly one of the shared variables
$\alpha^\mathrm{read}_i$ will be true,
namely the one with $i$ being the symbol in the current cell),
\item
a vector $\underline{\alpha}^\mathrm{written}=
(\alpha^\mathrm{written}_{0},\ldots,\alpha^\mathrm{written}_{|\Gamma|-1})$ giving in unary
the symbol that has just been written into the cell that has just been left,
unless the cloud corresponds to the first node in the computation tree --- in that
case the value of this vector is irrelevant
(here ``unary'' means: exactly one of the variables $\alpha^\mathrm{written}_i$ will be true,
namely the one with $i$ being the symbol that has just been written),
\item
a vector $\underline{\alpha}^\mathrm{prevpos}=
(\alpha^\mathrm{prevpos}_{N},\ldots,\alpha^\mathrm{prevpos}_0)$ giving in binary
the previous position of the tape head, that is, the position of the cell that has just been left,
unless the cloud corresponds to the first node in the computation tree --- in that
case the value of this vector is irrelevant.
\end{itemize}
The formula $f_\sxs(w)$ has to ensure that for any possible
computation step starting from such a computation node in an accepting tree
there exists a cloud describing the corresponding successor node in the accepting tree.
In this new cloud, the value of the counter for the time 
$\underline{\alpha}^\mathrm{time}$ has to be incremented.
This can be done by the technique for implementing a binary counter in $\sxs$ that was
described in Subsection~\ref{subsection:counter-S4XS5}.
In parallel, we have to make sure that in this new cloud also the vectors 
$\underline{\alpha}^\mathrm{pos}$,
$\underline{\alpha}^\mathrm{state}$,
$\underline{\alpha}^\mathrm{read}$,
$\underline{\alpha}^\mathrm{written}$, and 
$\underline{\alpha}^\mathrm{prevpos}$
are set to the right values.
The vector $\underline{\alpha}^\mathrm{prevpos}$ is a copy of the vector
$\underline{\alpha}^\mathrm{pos}$ in the previous cloud (we will see below how one can copy it).
For the vectors $\underline{\alpha}^\mathrm{pos}$,
$\underline{\alpha}^\mathrm{state}$, 
and $\underline{\alpha}^\mathrm{written}$
these values can be computed using the corresponding
element of the transition relation $\delta$ of the ATM.
For example, $\underline{\alpha}^\mathrm{pos}$
has to be decremented by one if the tape head moves to the left,
and it has to be incremented by one if the tape head moves to the right.
Also the new state 
(to be stored in $\underline{\alpha}^\mathrm{state}$)
and the symbol written into the cell that has
just been left (to be stored in $\underline{\alpha}^\mathrm{written}$) are
determined by the data of the previous computation node
and by the corresponding element of the transition relation $\delta$.\\
But the vector $\underline{\alpha}^\mathrm{read}$ is supposed
to describe the symbol in the current cell.
This symbol is not determined by the current computation step
but has either been written the last time when this cell has been visited during this computation or,
when this cell has never been visited before,
the symbol in this cell is still the one that was contained in this cell 
before the computation started.
How can one ensure that $\underline{\alpha}^\mathrm{read}$ is set to the right value?
We will do this by using persistent variables and a variable $B^\mathrm{active}$
that is neither shared nor persistent.
We will ensure that any cloud corresponding to a node in the computation
tree in which the cell $i$ is being visited contains a point 
with a persistent position vector $\underline{X}^\mathrm{pos}=\mathrm{bin}_{N+1}(i)$
and with a persistent time vector $\underline{X}^\text{time-apv}$ that
will be forced to contain in binary form the time one step after the previous visit of the cell $i$ or,
if the cell $i$ has not been visited before, to contain the value $0$.
Furthermore, there will be a third persistent vector
$\underline{X}^\mathrm{read}$.
If the cell $i$ has not been visited before then we will make sure 
that the vector $\underline{X}^\mathrm{read}$ encodes the initial symbol contained in cell $i$.
If the cell $i$ has been visited before then we will make sure 
that the vector $\underline{X}^\mathrm{read}$ encodes the symbol that has been
written into the cell during the previous visit of the cell.
In order to do that we need to determine the time of the previous visit.
For that the Boolean variable $B^\mathrm{active}$ is used and the fact that due to the left commutativity
property the point has $\stackrel{\Diamond}{\to}$-predecessors in all clouds corresponding
to any node on the path in the accepting tree from the root to the current cloud.
Then, when one knows this time, one can 
look at the shared variable vector $\underline{\alpha}^\mathrm{written}$
in the cloud corresponding to the step after the previous visit of this cell
(this shared variable vector encodes the symbol that has just been written into the cell)
and can copy its value to the 
persistent variable vector $\underline{X}^\mathrm{read}$.
In order to implement all this we
use the following persistent variables:
\begin{itemize}
\item
a vector $\underline{X}^\text{time-apv}=
(X^\text{time-apv}_{N-1},\ldots,X^\text{time-apv}_0)$ containing in binary
the time one step after the previous visit of the current cell 
(the exponent of $X^\text{time-apv}_i$ stands for ``time after previous visit'') or, if the current
cell has not been visited before, containing in binary the number $0$,
\item
a vector $\underline{X}^\mathrm{pos}=
(X^\mathrm{pos}_{N},\ldots,X^\mathrm{pos}_0)$ containing in binary
the current position of the tape head, that is, the position of the current cell,
\item
a vector $\underline{X}^\mathrm{read}=
(X^\mathrm{read}_{0},\ldots,X^\mathrm{read}_{|\Gamma|-1})$ 
giving in unary the symbol in the cell described by $\underline{X}^\mathrm{pos}$
(here ``unary'' means: exactly one of the shared variables
$\alpha^\mathrm{read}_i$ will be true,
namely the one with $i$ being the symbol in the cell described by $\underline{X}^\mathrm{pos}$),
\end{itemize}
In fact, the current cell may have been visited several times already. Of course
only the $\stackrel{\Diamond}{\to}$-successor of the point corresponding
to the previous visit (that is, corresponding to the very last visit
of the cell before the current visit, not earlier visits) should be allowed to copy the value of its vector
$\underline{\alpha}^\mathrm{written}$ to $\underline{X}^\mathrm{read}$.
In order to determine the right point, we are going to use an additional Boolean
\begin{itemize}
\item
variable $B^\mathrm{active}$ (which is neither persistent nor shared)
saying whether the current point is active or not.
\end{itemize}
We shall explain later how all this works.
Finally, there are two more vectors of persistent variables that are
needed for changing the values of the shared variable vectors
$\underline{\alpha}^\mathrm{time}$ and
$\underline{\alpha}^\mathrm{pos}$:
\begin{itemize}
\item
a vector $\underline{X}^\mathrm{prevtime}=
(X^\mathrm{prevtime}_{N-1},\ldots,X^\mathrm{prevtime}_0)$ containing in binary
the current time of the computation minus one,
unless the cloud corresponds to the first node in the computation tree --- in that
case the value of this vector is irrelevant,
\item
a vector $\underline{X}^\mathrm{prevpos}=
(X^\mathrm{prevpos}_{N},\ldots,X^\mathrm{prevpos}_0)$ containing in binary
the previous position of the tape head, that is, the position of the cell that has just been left,
unless the cloud corresponds to the first node in the computation tree --- in that
case the value of this vector is irrelevant.
\end{itemize}

Now we come to the formal definition of the formula $f_\sxs(w)$.
The formula $f_\sxs(w)$ will have the following structure:
\begin{eqnarray*}
  f_\sxs(w) 
    &:= & \mathit{persistence} \\
    && \wedge K\Box \mathit{uniqueness} \\
    && \wedge \mathit{start} \\
    && \wedge K\Box \mathit{initial\_symbols} \\
    && \wedge K\Box \mathit{written\_symbols} \\
    && \wedge K\Box \mathit{read\_a\_symbol} \\
    && \wedge K\Box \mathit{computation} \\
    && \wedge K\Box \mathit{no\_reject} .
\end{eqnarray*}
The formula $f_\sxs(w)$ will contain the following propositional variables:
\begin{eqnarray*}
&& A^\mathrm{time}_{N-1}, \ldots, A^\mathrm{time}_{0}, 
 A^\mathrm{pos}_{N}, \ldots, A^\mathrm{pos}_{0}, 
 A^\mathrm{state}_{0}, \ldots, A^\mathrm{state}_{|Q|-1}, 
 A^\mathrm{written}_{0}, \ldots, A^\mathrm{written}_{|\Gamma|-1}, 
 A^\mathrm{read}_{0},\ldots, A^\mathrm{read}_{|\Gamma|-1}, \\
&& A^\mathrm{prevpos}_{N}, \ldots, A^\mathrm{prevpos}_{0}, \\
&& X^\mathrm{prevtime}_{N-1},\ldots, X^\mathrm{prevtime}_{0}, 
 X^\mathrm{prevpos}_{N},\ldots, X^\mathrm{prevpos}_{0}, 
 X^\mathrm{pos}_{N}, \ldots, X^\mathrm{pos}_{0}, 
 X^\text{time-apv}_{N-1}, \ldots, X^\text{time-apv}_{0}, \\
&& X^\mathrm{read}_{0},\ldots, X^\mathrm{read}_{|\Gamma|-1}, \\
&& B^\mathrm{active} .         
\end{eqnarray*}
For $\mathit{string} \in\{\mathrm{time}, \mathrm{pos}, \mathrm{state}, \mathrm{read},
       \mathrm{prevpos}, \mathrm{written}\}$
and natural numbers $i$ we use
$\alpha^\mathit{string}_i$ as an abbreviation for $L A^\mathit{string}_i$.
These formulas $\alpha^\mathit{string}_i$ are the shared variables we talked about above.

We are now going to define the subformulas of $f_\sxs(w)$.
We will use the abbreviations introduced above,
in Table~\ref{table:abbrev1}, in Table~\ref{table:abbrev2}, and in Table~\ref{table:abbrev3}.

The following formula makes sure that certain propositional variables are persistent:
\begin{eqnarray*}
\mathit{persistence}
&:=&   \mathrm{persistent}({\underline{X}^\mathrm{prevtime}})
 \wedge \mathrm{persistent}({\underline{X}^\mathrm{prevpos}})  \\
&& \wedge \mathrm{persistent}({\underline{X}^\mathrm{pos}}) 
 \wedge \mathrm{persistent}({\underline{X}^\text{time-apv}}) 
 \wedge \mathrm{persistent}({\underline{X}^\mathrm{read}}) .
\end{eqnarray*}

The following formula makes sure that in each of the vectors of shared or persistent variables that
describe in a unary way the current state respectively the written symbol 
exactly one variable is true:
\begin{eqnarray*}
\mathit{uniqueness}&:=&
 \mathrm{unique}({\underline{\alpha}^\mathrm{state}}) 
 \wedge  \mathrm{unique}({\underline{\alpha}^\mathrm{written}})
 \wedge  \mathrm{unique}({\underline{X}^\mathrm{read}}). 
\end{eqnarray*}
The vector $\underline{\alpha}^\mathrm{read}$ will satisfy the same
uniqueness condition automatically due to another formula
(due to the formula $\mathit{read\_a\_symbol}$).

The following formula ensures that the shared variables in the cloud
corresponding to the first node in a computation tree have the correct values.
The computation starts at time $0$ with the tape head at position $2^N-1$ and in the state $q_0$.
The vector $\underline{\alpha}^\mathrm{read}$ will automatically get the correct value $\#$
due to the formulas $\mathit{read\_a\_symbol}$ and $\mathit{initial\_symbols}$,
that will be introduced next.
For the vectors 
$\underline{\alpha}^\mathrm{prevpos}$ and
$\underline{\alpha}^\mathrm{written}$
we do not need to fix any values.
\begin{eqnarray*}
   \mathit{start}
   &:= & \phantom{\wedge} (\underline{\alpha}^\mathrm{time}=\bin_N(0))
            \wedge (\underline{\alpha}^\mathrm{pos}=\bin_{N+1}(2^N-1)) 
     \wedge \alpha^\mathrm{state}_{q_0} .
\end{eqnarray*}

The following formula ensures that whenever an initial symbol on the tape
is requested (by a point with persistent time $\underline{X}^\text{time-apv}$ equal to $0$)
it is stored in a persistent vector $\underline{X}^\mathrm{read}$
of this point.
\begin{eqnarray*}
   \lefteqn{\mathit{initial\_symbols}} && \\
   &:= & (\underline{X}^\text{time-apv}=\mathrm{bin}_N(0)) \\
   && \rightarrow \Biggl( 
            \bigwedge_{i=1}^n \bigl( (\underline{X}^\mathrm{pos}=\mathrm{bin}_{N+1}(2^{N}-1+i))
            \rightarrow X^\mathrm{read}_{w_i} \bigr) \\
   && \qquad
          \wedge \bigl(\bigl( (\underline{X}^\mathrm{pos} \leq \mathrm{bin}_{N+1}(2^{N}-1))
                \vee  (\underline{X}^\mathrm{pos} > \mathrm{bin}_{N+1}(2^{N}-1+n)) \bigr)
         \rightarrow X^\mathrm{read}_{\#} \bigr) \Biggr).
\end{eqnarray*}
We explain this formula. 
By another formula we will ensure that whenever a cell is visited for the first time
there will be an ``active'' point storing the position of this cell in a persistent vector
$\underline{X}^\mathrm{pos}$ and such that its persistent time vector
$\underline{X}^\text{time-apv}$ has the binary value $0$.
The formula above ensures that the persistent vector $\underline{X}^\mathrm{read}$
in this point stores the correct initial symbol in this cell.
This is either a symbol $w_i$ of the input string $w=w_1\ldots w_n$ or the blank $\#$.

The following formula ensures that whenever a symbol that has just been written
on the tape is requested (by an active point with the correct persistent time
$\underline{X}^\text{time-apv}$)
then it is copied into a persistent vector $\underline{X}^\mathrm{read}$
of this point.
\begin{eqnarray*} \lefteqn{\mathit{written\_symbols}} && \\
 &:=& \bigl((\underline{X}^\text{time-apv} > \mathrm{bin}_N(0)) 
                               \wedge (\underline{X}^\text{time-apv} = \underline{\alpha}^\mathrm{time})
                               \wedge B^\mathrm{active} \bigr) 
   \rightarrow (\underline{X}^\mathrm{read} = \underline{\alpha}^\mathrm{written}) . 
\end{eqnarray*}

The following formula ensures that the shared variable vector $\underline{\alpha}^\mathrm{read}$
describes the symbol in the current cell. 
\begin{eqnarray*}
   \mathit{read\_a\_symbol}
 &:=&  \mathit{existence\_of\_a\_reading\_point} \\
 && \wedge \mathit{time\_of\_previous\_visit} \\
 && \wedge \mathit{becoming\_inactive} \\
 && \wedge \mathit{staying\_inactive} \\
 && \wedge \mathit{storing\_the\_read\_symbol},
\end{eqnarray*}
where
\begin{eqnarray*}
\lefteqn{\mathit{existence\_of\_a\_reading\_point}} && \\
     &:=& L\bigl( ( \underline{X}^\mathrm{pos} = \underline{\alpha}^\mathrm{pos} )
                        \wedge (\underline{X}^\text{time-apv} \leq \underline{\alpha}^\mathrm{time} )
                        \wedge B^\mathrm{active}\bigr) , \\
\lefteqn{\mathit{time\_of\_previous\_visit}} && \\
   &:=& \Bigl( \bigl(( \underline{X}^\text{time-apv} > \mathrm{bin}_N(0) )
                               \wedge( \underline{X}^\text{time-apv} = \underline{\alpha}^\mathrm{time})
                               \wedge B^\mathrm{active} \bigr) 
   \rightarrow
       ( \underline{X}^\mathrm{pos} = \underline{\alpha}^\mathrm{prevpos}) \Bigr), \\ 
\lefteqn{\mathit{becoming\_inactive}} && \\
   &:=& \Bigl( \bigl(( \underline{X}^\mathrm{pos}=\underline{\alpha}^\mathrm{prevpos} )
                     \wedge (\underline{X}^\text{time-apv} < \underline{\alpha}^\mathrm{time}) \bigr) \rightarrow
                     \neg B^\mathrm{active} \Bigr), \\
\lefteqn{\mathit{staying\_inactive}} && \\                     
   &:=& \Bigl( \neg B^\mathrm{active} \rightarrow \Box \neg B^\mathrm{active} \Bigr), \\
\lefteqn{\mathit{storing\_the\_read\_symbol}} && \\                        
   &:=& \Bigl( \bigl( (\underline{X}^\mathrm{pos}=\underline{\alpha}^\mathrm{pos} )
                     \wedge (\underline{X}^\text{time-apv} \leq \underline{\alpha}^\mathrm{time} )
                     \wedge B^\mathrm{active} \bigr) \rightarrow
                     (\underline{\alpha}^\mathrm{read} = \underline{X}^\mathrm{read}) \Bigr) .
\end{eqnarray*}
We explain these formulas.
The formula $\mathit{existence\_of\_a\_reading\_point}$
enforces the existence of an ``active'' point in the current cloud that
contains in the persistent vector $\underline{X}^\mathrm{pos}$ the
number of the current cell and that we wish to force to contain in the persistent
vector $\underline{X}^\text{time-apv}$ the time one step after the previous
visit of this cell, if this cell has been visited before, or the value $0$, otherwise.
How can we enforce that? 
We make essential use of the left commutativity property.
The point whose existence is ensured by the formula
$\mathit{existence\_of\_a\_reading\_point}$
has $\stackrel{\Diamond}{\to}$-predecessors in all clouds corresponding
to the nodes on the path in the accepting tree from the root to the node corresponding
to the current cloud. Since we demand
$\underline{X}^\text{time-apv} \leq \underline{\alpha}^\mathrm{time}$
the time stored in binary in the persistent variable $\underline{X}^\text{time-apv}$
must be identical with the time of one of the clouds corresponding to a node on this path.
The formula $\mathit{time\_of\_previous\_visit}$ ensures that if it is positive
then it can be equal to the time $\underline{\alpha}^\mathrm{time}$
of a cloud corresponding to a node on this path only when 
$\underline{X}^\mathrm{pos} = \underline{\alpha}^\mathrm{prevpos}$,
that is, only when in the step leading to this node something has been written into the current cell.
So, if the current cell has not been visited before, the binary value
of $\underline{X}^\text{time-apv}$ must be $0$. Then 
the formula $\mathit{initial\_symbols}$ will ensure that the persistent variable
$\underline{X}^\mathrm{read}$ gets the correct value.
Otherwise, when the cell has been visited before, in fact, it may have been visited several times
already. In this case, we have to make sure that the number stored
in $\underline{X}^\text{time-apv}$ is the time after the very last visit to this cell
immediately before the current visit.
This is ensured by the formulas $\mathit{becoming\_inactive}$
and $\mathit{staying\_inactive}$.
If the number stored in $\underline{X}^\text{time-apv}$ were strictly
smaller than the the time after the very last visit to this cell
immediately before the current visit then, due to $\mathit{becoming\_inactive}$
the corresponding point would be set to ``inactive'', and, due to
$\mathit{staying\_inactive}$, also its $\stackrel{\Diamond}{\to}$-successor
would stay inactive. But this would apply also to the point in the current cloud
which is supposed to be active and to contain the correct value
in $\underline{X}^\text{time-apv}$ according to formula
$\mathit{existence\_of\_a\_reading\_point}$.
Thus, the first four subformulas of $\mathit{read\_a\_symbol}$
ensure that the current cloud contains an active point that
contains in the persistent vector $\underline{X}^\mathrm{pos}$ the
number of the current cell and that contains in the persistent
vector $\underline{X}^\text{time-apv}$ the time one step after the previous
visit of this cell, if this cell has been visited before, or the value $0$, otherwise.
Finally, the formula $\mathit{storing\_the\_read\_symbol}$
makes sure that the value of 
$\underline{X}^\mathrm{read}$ (which contains the symbol written into
the current cell during the previous visit of this cell, or, if the current cell has not
been visited before, the initial symbol in this cell) 
is copied into the shared variable vector
$\underline{\alpha}^\mathrm{read}$.

Next, we wish to define the formula $\mathit{computation}$ that describes the computation steps.
We have to distinguish between the two cases whether the tape head is going to move to the left or to the right.
If in a computation step the symbol $\theta\in\Gamma$ is written into the current cell,
if the tape head moves to the right, and if the new state after this step is the state $r \in Q$,
then the following formula $\mathit{compstep}_{\mathrm{right}}(r,\theta)$ guarantees the existence of a point
and its cloud with suitable values in the shared variable vectors
$\underline{\alpha}^\mathrm{time}, \underline{\alpha}^\mathrm{pos},
\underline{\alpha}^\mathrm{prevpos}, \underline{\alpha}^\mathrm{state},
\underline{\alpha}^\mathrm{written}$. 

We explain this formula.
The first four lines of this formula take care that the two binary counters
$\underline{\alpha}^\mathrm{time}$ and
$\underline{\alpha}^\mathrm{pos}$
for the current time and for the current position of the tape head
are incremented at the same time. This is similar to the formula
$\mathrm{counter}_{\sxs,n}$ in Subsection~\ref{subsection:counter-S4XS5}.
Furthermore, also the persistent vectors
$\underline{X}^\mathrm{prevtime}$ (that we will not need otherwise) and
$\underline{X}^\mathrm{prevpos}$ (that we use again in the fifth line of this formula)
get the correct values.
The fifth line ensures that the shared variable vectors
$\underline{\alpha}^\mathrm{prevpos}$, 
$\underline{\alpha}^\mathrm{state}$ and 
$\underline{\alpha}^\mathrm{written}$
of the current point get the correct values.
The shared variable vector
$\underline{\alpha}^\mathrm{read}$ will get the correct value
due to the formulas $\mathit{initial\_symbols}$,
$\mathit{written\_symbols}$, and $\mathit{read\_a\_symbol}$.

\begin{eqnarray*}
	\mathit{compstep}_{\mathrm{right}}(r,\theta)
	&:=& \bigwedge_{k=0}^{N-1} \bigwedge_{l=0}^{N}
	\Biggl( \bigl( \mathrm{rightmost\_zero}(\underline{\alpha}^\mathrm{time},k)
	\wedge \mathrm{rightmost\_zero}(\underline{\alpha}^\mathrm{pos},l) \bigr) \\
	&&      \rightarrow L \biggl(  (\underline{X}^\mathrm{prevtime}=\underline{\alpha}^\mathrm{time})
	\wedge  (\underline{X}^\mathrm{prevpos}=\underline{\alpha}^\mathrm{pos}) \\
	&& \phantom{\rightarrow L \biggl(}  \wedge \Diamond \Bigl(
	(\underline{\alpha}^\mathrm{time}=\underline{X}^\mathrm{prevtime},>k) 
	\wedge  \mathrm{rightmost\_one}(\underline{\alpha}^\mathrm{time},k) \\
	&& \phantom{\rightarrow L \biggl(  \wedge \Diamond \Bigl(}  
	\wedge (\underline{\alpha}^\mathrm{pos}=\underline{X}^\mathrm{prevpos}>l) 
	\wedge \mathrm{rightmost\_one}(\underline{\alpha}^\mathrm{pos},l) \\
	&& \phantom{\rightarrow L \biggl(  \wedge \Diamond \Bigl(}  
	\wedge (\underline{\alpha}^\mathrm{prevpos}=\underline{X}^\mathrm{prevpos}) 
	\wedge \alpha^\mathrm{state}_{r} 
	\wedge \alpha^\mathrm{written}_\theta
	\Bigr) \biggr) \Biggr).
\end{eqnarray*}

If in a computation step the symbol $\theta\in\Gamma$ is written into the current cell,
if the tape head moves to the left, and if the new state after this step is the state $r \in Q$,
then the following formula guarantees the existence of a point
and its cloud with suitable values in the shared variables.
\begin{eqnarray*}
\mathit{compstep}_{\mathrm{left}}(r,\theta)
&:=& \bigwedge_{k=0}^{N-1} \bigwedge_{l=0}^{N}
   \Biggl( \bigl( \mathrm{rightmost\_zero}(\underline{\alpha}^\mathrm{time},k)
           \wedge \mathrm{rightmost\_one}(\underline{\alpha}^\mathrm{pos},l) \bigr) \\
&&      \rightarrow L \biggl(  (\underline{X}^\mathrm{prevtime}=\underline{\alpha}^\mathrm{time})
                             \wedge  (\underline{X}^\mathrm{prevpos}=\underline{\alpha}^\mathrm{pos}) \\
&& \phantom{\rightarrow L \biggl(}  \wedge \Diamond \Bigl(
       (\underline{\alpha}^\mathrm{time}=\underline{X}^\mathrm{prevtime},>k) 
                        \wedge  \mathrm{rightmost\_one}(\underline{\alpha}^\mathrm{time},k) \\
&& \phantom{\rightarrow L \biggl(  \wedge \Diamond \Bigl(}  
       \wedge (\underline{\alpha}^\mathrm{pos}=\underline{X}^\mathrm{prevpos},>l) 
                        \wedge \mathrm{rightmost\_zero}(\underline{\alpha}^\mathrm{pos},l) \\
&& \phantom{\rightarrow L \biggl(  \wedge \Diamond \Bigl(}  
                        \wedge (\underline{\alpha}^\mathrm{prevpos}=\underline{X}^\mathrm{prevpos}) 
                        \wedge \alpha^\mathrm{state}_{r} 
                        \wedge \alpha^\mathrm{written}_\theta 
\Bigr) \biggr) \Biggr).
\end{eqnarray*}
This formula is very similar to the previous one with the exception that here
the binary counter for the position of the tape head is decremented.

The computation is modeled by the following subformula.
Remember that $Q$ is the disjoint union of the sets
$\{q_{\mathrm{accept}}\}$, $\{q_{\mathrm{reject}}\}$, $Q_{\forall}$, $Q_{\exists}$.
\begin{eqnarray*}
\mathit{computation}
&:=& \bigwedge_{q \in Q_\forall} \bigwedge_{\eta \in\Gamma}
   \Biggl(   (\alpha^\mathrm{state}_q \wedge \alpha^\mathrm{read}_\eta) \rightarrow \\
&& \biggl( \bigwedge_{(r,\theta,\mathit{left}) \in \delta(q,\eta)}  \mathit{compstep}_{\mathrm{left}}(r,\theta)
   \wedge   \bigwedge_{(r,\theta,\mathit{right}) \in \delta(q,\eta)}  \mathit{compstep}_{\mathrm{right}}(r,\theta) \biggr) \Biggr) \\
&& \wedge \bigwedge_{q \in Q_\exists} \bigwedge_{\eta \in\Gamma}
   \Biggl(   (\alpha^\mathrm{state}_q \wedge \alpha^\mathrm{read}_\eta) \rightarrow \\
&& \biggl( \bigvee_{(r,\theta,\mathit{left}) \in \delta(q,\eta)}  \mathit{compstep}_{\mathrm{left}}(r,\theta)
   \vee   \bigvee_{(r,\theta,\mathit{right}) \in \delta(q,\eta)}  \mathit{compstep}_{\mathrm{right}}(r,\theta) \biggr) \Biggr) \Biggr) .
\end{eqnarray*}

Finally, the subformula $\mathit{no\_reject}$ is defined as follows.
\[ \mathit{no\_reject} := 
    \neg \alpha^\mathrm{state}_{q_{\mathrm{reject}}} . \]

We have completed the description of the formula $f_\sxs(w)$ for $w\in \Sigma^*$.
It is clear that $f_\sxs(w)$ is a bimodal formula, for any $w \in\Sigma^*$.
We still have to show two claims:
\begin{enumerate}
\item
The function $f_\sxs$ can be computed in logarithmic space.
\item
For any $w \in\Sigma^*$,
\[ w \in L \iff \text{ the bimodal formula $f_\sxs(w)$
is $\sxs$-satisfiable.} \]
\end{enumerate}
The first claim is shown in the following section.
The two directions of the equivalence in the second claim are shown afterwards
in separate sections.

\subsection{LOGSPACE Computability of the Reduction}
\label{subsection:reduction-ATMs-S4XS5-complexity}

We wish to show that the function $f_\sxs$ can be computed in logarithmic space.
This is shown by the same argument as the corresponding claim for the function
$f_\ssl$ in Subsection~\ref{subsection:reduction-ATMs-SSL-complexity}.

\subsection{Construction of a Model}
\label{subsection:model-construction}

In this section we show for any $w \in\Sigma^*$,
if $w \in L$ then the bimodal formula $f_\sxs(w)$
is $\sxs$-satisfiable.
Let us assume $w \in L$. 
We are going to explicitly define
an $\sxs$-product model of $f_\sxs(w)$.

There exists an accepting tree $T=(V,E,c)$ of $M$ on input $w$,
where $V$ is the set of nodes of $T$, where $E \subseteq V\times V$ is the set of edges,
and where the function $c:V\to Q \times \{0,\ldots,2^{N+1}-2\}\times \Gamma^{2^{N+1}-1}$
labels each node with a configuration
(remember the discussion about the description of configurations at the beginning of Subsection~\ref{subsection:definition-ATMs-S4XS5}).
Let $\mathit{root} \in V$ be the root of $T$.
We will now construct an $\sxs$-product model of $f_\sxs(w)$.
We define an $S4$-frame $(W_1,R_\Diamond)$ by
\[ W_1 := V \quad \text{ and } \quad R_\Diamond:= \text{ the reflexive-transitive closure of } E . \]
That means, for any $v,v' \in W_1$ we have $v R_\Diamond v'$ iff in the tree $T$ there is a path from $v$ to $v'$.
We define an $S5$-frame $(W_2,R_L)$ by
\[ W_2 := V  \quad \text{ and } \quad R_L := W_2 \times W_2 . \]
Then the product frame $(W, \stackrel{\Diamond}{\to}, \stackrel{L}{\to})$
with $W:=W_1\times W_2$ and 
with $\stackrel{\Diamond}{\to}$ and $\stackrel{L}{\to}$ defined as in
\cite[Definition 2.3]{HK2019-1} is an $\sxs$-product frame.
We still need to define a suitable valuation $\sigma$.
We define the valuation $\sigma$ as follows.
\begin{eqnarray*}
\sigma(A^\mathrm{time}_k) &:=& \{(v,x) ~:~ v,x \in V
       \text{ and } k \in \mathrm{Ones}(\mathit{time}(v)) \}, \\
\sigma(X^\mathrm{prevtime}_k) &:=& \{(v,x) ~:~ v\in V, x \in V\setminus\{\mathit{root}\}
           \text{ and } k \in \mathrm{Ones}(\mathit{time}(x)-1) \}, \\
\sigma(X^\text{time-apv}_k) &:=& \{(v,x) ~:~ v\in V,x \in V \setminus\{\mathit{root}\}
   \text{ and there exists a node } y \neq \\
&& \mathit{root}  \text{ on the path from } \mathit{root} \text{ to } x \text{ such that } 
    \mathit{pos}(\mathit{pred}(y)) = \\
&&\mathit{pos}(x) \text{ and } k \in \mathrm{Ones}(\mathit{time}(y)) \text{ and, for all nodes $z$ on the} \\
&& \text{path from $y$ to $x$ with $z\neq y$}, 
   \mathit{pos}(\mathit{pred}(z)) \neq \mathit{pos}(x) \} , 
\end{eqnarray*}
for $k \in \{0,\ldots,N-1\}$,
\begin{eqnarray*}
\sigma(A^\mathrm{pos}_k) &:=& \{(v,x) ~:~ v \in V, x\in V \text{ and } k \in \mathrm{Ones}(\mathit{pos}(v)) \}, \\
\sigma(X^\mathrm{pos}_k) &:=& \{(v,x) ~:~ v \in V, x\in V \text{ and } k \in \mathrm{Ones}(\mathit{pos}(x)) \}, \\
\sigma(A^\mathrm{prevpos}_k) &:=& \{(v,x) ~:~  v\in V\setminus\{\mathit{root}\}, x \in V 
         \text{ and } k \in \mathrm{Ones}(\mathit{pos}(\mathit{pred}(v))) \} \\
\sigma(X^\mathrm{prevpos}_k) &:=& \{(v,x) ~:~  v\in V, x\in V\setminus\{\mathit{root}\}
        \text{ and } k \in \mathrm{Ones}(\mathit{pos}(\mathit{pred}(x))) \}, 
\end{eqnarray*}
for $k \in \{0,\ldots,N\}$,
\[ \sigma(A^\mathrm{state}_q) := \{(v,x) ~:~ v\in V, x \in V \text{ and } q = \mathit{state}(v)) \}
\]
for $q\in Q$,
\begin{eqnarray*}
\sigma(A^\mathrm{read}_\eta) &:=& \{(v,x) ~:~ v,x \in V \text{ and } \eta = \mathit{read}(v) \}, \\
\sigma(A^\mathrm{written}_\eta) &:=& \{(v,x) ~:~ (v \in V \setminus\{\mathit{root}\},
            x \in V \text{ and } \eta = \mathit{written}(v) ) \\
   && \phantom{(\{(v,x) ~:~}  \text{or } (v=\mathit{root} \text{ and } x \in V \text{ and } \eta=\#)\}, \\
\sigma(X^\mathrm{read}_\eta)
   &:=& \{(v,x) ~:~ v,x \in V \text{ and } \eta = \mathit{read}(x)) \} , 
\end{eqnarray*}
for $\eta \in\Gamma$,
\begin{eqnarray*}
\sigma(B^\mathrm{active})
  &:= &\{(v,x) ~:~ v, x \in V 
  		\text{ and $v$ is an element of the path }
  	\text{from $\mathit{root}$        to $x$} \}.
\end{eqnarray*}
We have defined an $\sxs$-product model
$(W_1\times W_2,\stackrel{\Diamond}{\to},\stackrel{L}{\to})$.
We claim that in this model $(\mathit{root},\mathit{root})\models f_\sxs(w)$.
For an illustration of an important detail of the model see Figure~\ref{figure:model}.
\begin{figure}
\begin{center}
\includegraphics[width=0.9\linewidth]{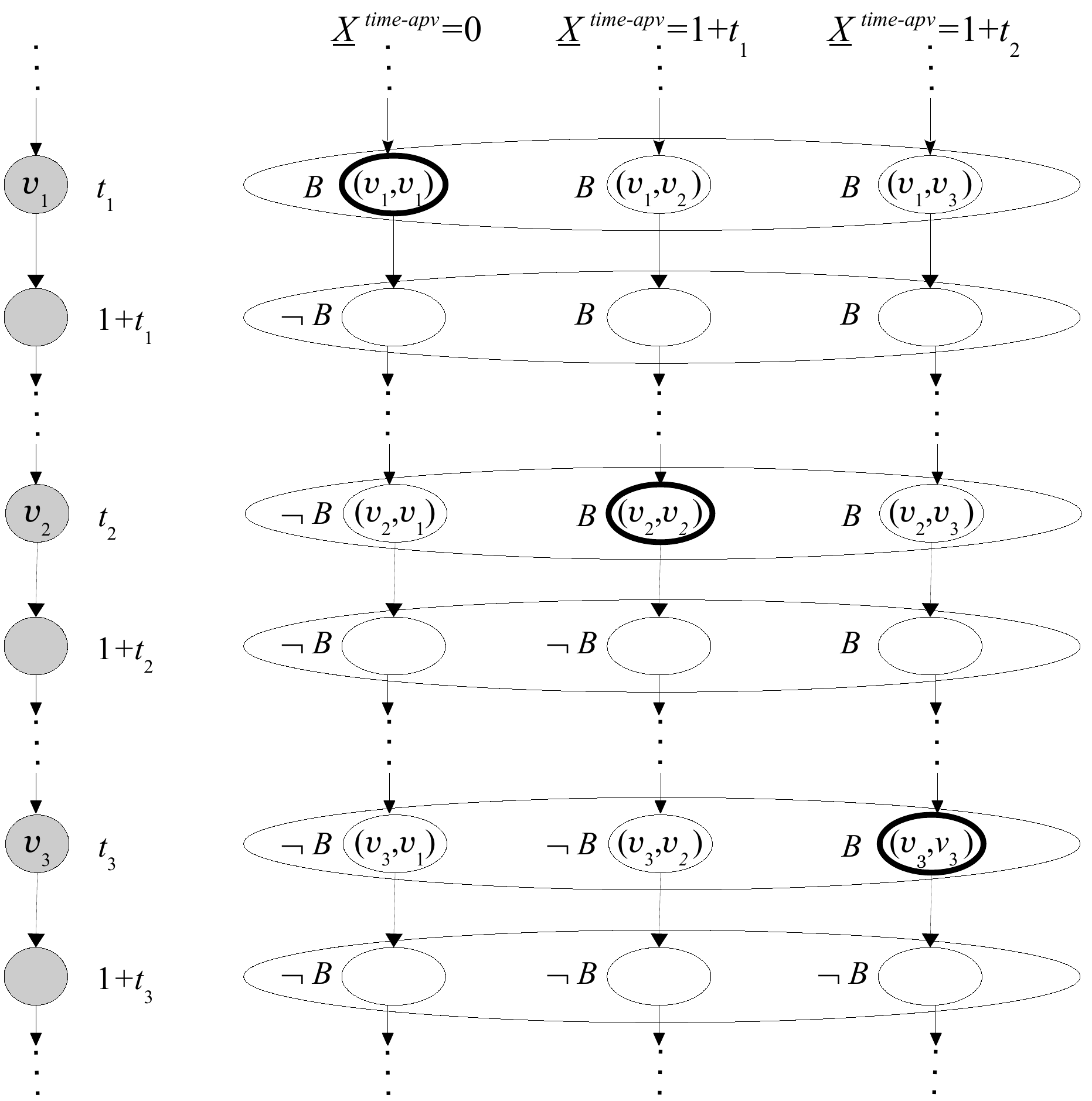}
\end{center}
\caption{A possible detail of an $\sxs$-product model of the formula $f_\sxs(w)$.
Consider a certain cell and let us assume that $v_1, v_2, v_3$ are the first three computation nodes
on some computation path in which this cell is visited. Let
$t_i:=\mathit{time}(v_i)$. The diagram on the left shows a part of the computation path.
The diagram on the right shows the corresponding part of the $\sxs$-product model.
Here $B$ stands for $B^\mathrm{active}$.
}
\label{figure:model}

\end{figure}

First, we observe that for
$\mathit{string} \in \{\text{prevtime},\text{prevpos}, \text{pos}, \text{time-apv}, \text{written}\}$
and natural numbers $i$ as well as for $(v,x) \in V \times V$ the truth value of 
the variable $X^\mathit{string}_i$ in the point $(v,x)$ does not depend on $v$.
Hence, all these variables are persistent. So,
\[ (\mathit{root},\mathit{root})\models \mathit{persistence} .\]
Similarly, for $\mathit{string} \in \{\mathrm{time}, \mathrm{pos},
       \mathrm{state}, \mathrm{read}, \mathrm{prevpos}, \mathrm{written}\}$
and natural numbers $i$ as well as for $(v,x) \in V \times V$ the truth value of the variable
$A^\mathit{string}_i$ in the point $(v,x)$ does not depend on $x$. 
Note that for $v\in V$ the set
\[ \mathrm{Cloud}(v) := \{(v,x) ~:~ v \in V, x \in V\} \]
is the $\stackrel{L}{\to}$-equivalence class of $(v,v)$.
All of the vectors of shared variables have the expected values:
for $v\in V$ and $s \in \mathrm{Cloud}(v)$. We see
\begin{eqnarray*}
s &\models& (\underline{\alpha}^\mathrm{time} = \mathrm{bin}_N( \mathit{time}(v))) , \\
s &\models& (\underline{\alpha}^\mathrm{pos} = \mathrm{bin}_{N+1}( \mathit{pos}(v))) , \\
( s &\models& \alpha^\mathrm{state}_q) \iff q = \mathit{state}(v), \text{ for } q \in Q, \\
( s &\models& \alpha^\mathrm{read}_\eta) \iff \eta = \mathit{read}(v), \text{ for } \eta \in \Gamma,
\end{eqnarray*}
and for $v\in V\setminus\{\mathit{root}\}$ and $s \in \mathrm{Cloud}(v)$ we see
\begin{eqnarray*}
s &\models& (\underline{\alpha}^\mathrm{prevpos} = \mathrm{bin}_{N+1}( \mathit{pos}(\mathit{pred}(v)))) , \\
( s &\models& \alpha^\mathrm{written}_\eta) \iff  \eta = \mathit{written}(v)  , \text{ for } \eta \in \Gamma.
\end{eqnarray*}
Similarly, for the vectors of persistent variables we see,
for $v,x\in V$:
\begin{eqnarray*}\textbf{}
(v,x) &\models& (\underline{X}^\mathrm{prevtime} = \mathrm{bin}_N(t)) ,
\text{ where } \\
&& t = \begin{cases} 
   0, &  \text{ if } x=\mathit{root} \\
   \mathit{time}(\mathit{pred}(x)), & \text{ otherwise},
\end{cases} \\  
(v,x) &\models& (\underline{X}^\mathrm{prevpos} = \mathrm{bin}_{N+1}(p)) ,
\text{ where } \\
&& p = \begin{cases} 
   0, &  \text{ if } x=\mathit{root} \\
    \mathit{pos}(\mathit{pred}(x)), & \text{ otherwise},
\end{cases} \\  
(v,x) &\models& (\underline{X}^\mathrm{pos} = \mathrm{bin}_{N+1}(\mathit{pos}(x))) , \\
(v,x) &\models& (\underline{X}^\text{time-apv} = \mathrm{bin}_N(t)) , \text{ where } \\
&& \!\!\!\!\!\! t = \begin{cases} 
   0, \quad  \text{if the cell $\mathit{pos}(x)$ has not been visited before on the computation} \\
   \qquad \text{path from $\mathit{root}$ to $x$ (this is in particular true in the case $x=\mathit{root}$)}, \\
   1 + \text{the time of the previous visit of the cell $\mathit{pos}(x)$,} \\
   \qquad \text{otherwise}.
\end{cases} \\   
( (v,x) &\models& X^\mathrm{read}_\eta) \iff \eta = \mathit{read}(x), \text{ for } \eta \in\Gamma,  
\end{eqnarray*}
It is straightforward to check that for all $(v,x) \in W$
\[ (v,x) \models  \mathit{uniqueness}  \]
(in fact, in order to achieve $(\mathit{root},x) \models \mathrm{unique}(\underline{\alpha}^\mathrm{written})$
we made a somewhat arbitrary choice for the value of $A^\mathrm{written}_\eta$ in $(\mathit{root},x)$,
for $x\in V$ and $\eta\in\Gamma$).
Hence, we have
\[ (\mathit{root},\mathit{root}) \models K \Box \mathit{uniqueness} . \]
It is clear as well that
\[ (\mathit{root},\mathit{root}) \models \mathit{start} . \]
As none of the nodes in the accepting tree $T$ is labeled with a configuration with the state $q_\mathrm{reject}$
we have $(v,x) \models \mathit{no\_reject}$, for all $(v,x) \in W$, hence
\[ (\mathit{root},\mathit{root}) \models K \Box \mathit{no\_reject} . \]
We still need to show that
\begin{eqnarray*}
 (\mathit{root},\mathit{root}) &\models&
   K \Box \mathit{initial\_symbols}
  \wedge  K \Box \mathit{written\_symbols} \\
  && \wedge K \Box \mathit{read\_a\_symbol} 
  \wedge K \Box \mathit{computation} . 
\end{eqnarray*}
It is sufficient to show that, for all $(v,x) \in W$,
\[ (v,x) \models \mathit{initial\_symbols} \wedge  \mathit{written\_symbols} 
           \wedge \mathit{read\_a\_symbol} \wedge \mathit{computation} . \]

First, let us consider the formula $\mathit{initial\_symbols}$.
We wish to show that
$$(v,x) \models \mathit{initial\_symbols},$$ 
for all $(v,x) \in W$.
Nothing needs to be shown if $(v,x) \models (\underline{X}^\text{time-apv} = \mathrm{bin}_N(0))$ is not true.
So, let us assume that $(v,x) \models (\underline{X}^\text{time-apv} = \mathrm{bin}_N(0))$.
We noted above that
the assumption
$(v,x) \models (\underline{X}^\text{time-apv} = \mathrm{bin}_N(0))$
implies that the cell $\mathit{pos}(x)$ has not been visited before
on the computation path from $\mathit{root}$ to $x$.
Hence, 
$\mathit{read}(x)=\eta$, where
$\eta$ is the initial symbol in the cell $x$.
We obtain
$(v,x) \models X^\mathrm{read}_{\eta}$. 
So, if the number $i:=\mathit{pos}(x)-(2^N-1)$ satisfies
$1\leq i \leq n$ then $(v,x) \models X^\mathrm{read}_{w_i}$,
otherwise $(v,x) \models X^\mathrm{read}_{\#}$.
Remember that
$(v,x) \models (\underline{X}^\mathrm{pos} = \mathrm{bin}_{N+1}(\mathit{pos}(x))$.
We have shown
$(v,x) \models \mathit{initial\_symbols}$.

Next, we consider the formula $\mathit{written\_symbols}$ and show that
$$(v,x) \models \mathit{written\_symbols}$$
for all $(v,x)\in W$.
Let us assume
\[ (v,x) \models ( (\underline{X}^\text{time-apv} > \mathrm{bin}_N(0) )
                               \wedge (\underline{X}^\text{time-apv} = \underline{\alpha}^\mathrm{time})
                               \wedge B^\mathrm{active} ) 
\]
(otherwise, nothing needs to be shown).
The condition $(v,x) \models ( \underline{X}^\text{time-apv} > \mathrm{bin}_N(0))$
implies $x \in V \setminus\{\mathit{root}\}$
and that the binary value $t$ of the vector $\underline{X}^\text{time-apv}$
in the point $(v,x)$ is equal to $1+$ the time of the previous visit of the cell $\mathit{pos}(x)$
on the computation path from $\mathit{root}$ to $x$.
Note that $t\leq \mathit{time}(x)$.
The condition
$(v,x) \models (\underline{X}^\text{time-apv} = \underline{\alpha}^\mathrm{time})$
means $t=\mathit{time}(v)$.
Now the condition $(v,x) \models B^\mathrm{active}$ implies that 
$v$ is an element of the path from $\mathit{root}$ to $x$.
Actually, due to $t = \mathit{time}(v)$
the node $v$ is exactly the computation node on the computation path from $\mathit{root}$
to $x$ after the previous visit of the cell $\mathit{pos}(x)$.
Thus, the symbol written during this visit and described by the value of
$\underline{\alpha}^\mathrm{written}$ at the point $(v,x)$
is just the symbol still contained in the same cell when the computation node $x$ is reached.
Hence, we have $(v,x) \models (\underline{X}^\mathrm{read} = \underline{\alpha}^\mathrm{written})$.
For later purposes we note that for the same reason we have
$(v,x) \models (\underline{X}^\mathrm{pos} = \underline{\alpha}^\mathrm{prevpos})$ as well.
We have shown $(v,x) \models \mathit{written\_symbols}$.

Next, we consider the formula $\mathit{read\_a\_symbol}$.
We wish to show that
$$(v,x) \models \mathit{read\_a\_symbol},$$ 
for all $(v,x) \in W$.
We show this separately for the five subformulas of $\mathit{read\_a\_symbol}$.
First, we show
$$(v,x) \models \mathit{existence\_of\_a\_reading\_point},$$
that is,
\[ 
 (v,x) \models L \bigl(( \underline{X}^\mathrm{pos} = \underline{\alpha}^\mathrm{pos}) 
                        \wedge (\underline{X}^\text{time-apv} \leq \underline{\alpha}^\mathrm{time}) 
                        \wedge B^\mathrm{active}\bigr) .
\]
Indeed, it is clear that
$(v,x) \stackrel{L}{\to} (v,v)$ and that
$(v,v) \models ( \underline{X}^\mathrm{pos} = \underline{\alpha}^\mathrm{pos})$ and
$(v,v) \models B^\mathrm{active}$.
Furthermore, the binary value of $\underline{X}^\text{time-apv}$ in the point 
$(v,v)$ is either $1+$ the time of the previous visit of the cell $\mathit{pos}(v)$
on the path from $\mathit{root}$ to $v$,
if this cell has been visited before $v$ on this path, or $0$, otherwise.
In any case we obtain
$(v,v) \models (\underline{X}^\text{time-apv} \leq \underline{\alpha}^\mathrm{time})$.
Next, we show
$$(v,x) \models \mathit{time\_of\_previous\_visit},$$ 
that is,
\[ (v,x) \models \Bigl( \bigl(( \underline{X}^\text{time-apv} > \mathrm{bin}_N(0) )
\wedge (\underline{X}^\text{time-apv} = \underline{\alpha}^\mathrm{time})
\wedge B^\mathrm{active} \bigr) \rightarrow
( \underline{X}^\mathrm{pos} = \underline{\alpha}^\mathrm{prevpos}) \Bigr) . 
\]
Actually, we have seen this already in the proof of
$(v,x) \models \mathit{written\_symbols}$ above.
Next, we show
$$(v,x) \models \mathit{becoming\_inactive},$$ 
that is,
\[ (v,x) \models \Bigl( \bigl( (\underline{X}^\mathrm{pos}=\underline{\alpha}^\mathrm{prevpos}) 
                     \wedge (\underline{X}^\text{time-apv} < \underline{\alpha}^\mathrm{time}) \bigr) \rightarrow
                     \neg B^\mathrm{active} \Bigr) 
\]
For the sake of a contradiction, let us assume
$(v,x) \models ( (\underline{X}^\mathrm{pos}=\underline{\alpha}^\mathrm{prevpos}) 
                     \wedge (\underline{X}^\text{time-apv} < \underline{\alpha}^\mathrm{time}))$
and $(v,x) \models B^\mathrm{active}$.
Then $v$ is a node on the path from $\mathit{root}$ to $x$.
Due to $(v,x) \models ( \underline{X}^\text{time-apv} < \underline{\alpha}^\mathrm{time})$
we have $\mathit{time}(v)>0$. And due to
$(v,x) \models ( \underline{X}^\mathrm{pos}=\underline{\alpha}^\mathrm{prevpos})$
the node $v$ must have the property
$\mathit{pos}(\mathit{pred}(v))=\mathit{pos}(x)$.
That means that the cell $\mathit{pos}(x)$ has been visited before $x$ on the path from $\mathit{root}$
to $x$. But under these circumstances, the binary value $t$ of
$\underline{X}^\text{time-apv}$ at the point $(v,x)$ is equal to $\mathit{time}(u)$
where $u$ is the last node on the path from $\mathit{root}$ to $x$ with the
property $\mathit{pos}(\mathit{pred}(u))=\mathit{pos}(x)$.
Note that $v$ is a node on the path from $\mathit{root}$ to $x$ with the
property $\mathit{pos}(\mathit{pred}(u))=\mathit{pos}(x)$.
On the other hand, the condition
$(v,x) \models ( \underline{X}^\text{time-apv} < \underline{\alpha}^\mathrm{time})$
implies $t < \mathit{time}(v)$.
That is a contradiction.
We have shown $(v,x) \models \mathit{becoming\_inactive}$.
Next,we show 
$$(v,x) \models \mathit{staying\_inactive},$$
that is,
\[ (v,x) \models \Bigl( \neg B^\mathrm{active} \rightarrow \Box \neg B^\mathrm{active} \Bigr) . 
\]
This is clear from the definition of $\sigma(B^\mathrm{active})$.

We come to the last subformula of $\mathit{read\_a\_symbol}$ and show
$$(v,x) \models \mathit{storing\_the\_read\_symbol},$$
that is,
\[ (v,x) \models           
   \Bigl( \bigl( (\underline{X}^\mathrm{pos}=\underline{\alpha}^\mathrm{pos} )
                     \wedge (\underline{X}^\text{time-apv} \leq \underline{\alpha}^\mathrm{time} )
                     \wedge B^\mathrm{active} \bigr) \rightarrow
                     (\underline{\alpha}^\mathrm{read} = \underline{X}^\mathrm{read}) \Bigr) .
\]
The condition
$(v,x) \models B^\mathrm{active}$ implies that $v$ is a node on the path from $\mathit{root}$ to $x$.
The condition $(v,x) \models ( \underline{X}^\mathrm{pos}=\underline{\alpha}^\mathrm{pos}  )$
says that $\mathit{pos}(x) = \mathit{pos}(v)$.
We claim that $v=x$.
Once we have shown this, of course, we obtain
$(v,x) \models (\underline{\alpha}^\mathrm{read} = \underline{X}^\mathrm{read})$.
For the sake of a contradiction, let us assume $v\neq x$.
Then the cell $\mathit{pos}(x)$
has been visited before $x$ on the path from $\mathit{root}$ to $x$.
Let $u$ be the last node before $x$ on the path from $\mathit{root}$ to $x$ with 
$\mathit{pos}(u)=\mathit{pos}(x)$.
We obtain $\mathit{time}(v) \leq \mathit{time}(u)$.
For the binary value $t$ of $\underline{X}^\text{time-apv}$ in $(v,x)$
we obtain $t= 1 + \mathit{time}(u)$.
Finally, the condition $(v,x) \models (\underline{X}^\text{time-apv} \leq \underline{\alpha}^\mathrm{time})$
implies $t \leq \mathit{time}(v)$.
By putting all this together we arrive at the following contradiction:
\[ \mathit{time}(v) \leq \mathit{time}(u) < 1 +  \mathit{time}(u) = t \leq \mathit{time}(v). \]
We have shown $(v,x) \models \mathit{storing\_the\_read\_symbol}$.

Finally, we have to show that 
\[ (v,x) \models \mathit{computation} , \]
for all $(v,x) \in W$.
We will separately treat the conjunctions over the set $(q,\eta) \in Q_\exists \times \Gamma$
and over the set $Q_\forall \times \Gamma$.
Let us first fix a pair $(q,\eta) \in Q_\exists \times \Gamma$
and let us assume that $(v,x)\in W$ is a point with
$(v,x) \models (\alpha^\mathrm{state}_q \wedge \alpha^\mathrm{read}_\eta)$.
We have to show that there is an element
$(r,\theta,\mathit{left}) \in \delta(q,\eta)$ such that
$(v,x) \models \mathit{compstep}_{\mathrm{left}}(r,\theta)$
or that there is an element
$(r,\theta,\mathit{right}) \in \delta(q,\eta)$ such that
$(v,x) \models \mathit{compstep}_{\mathrm{right}}(r,\theta)$.
As $T$ is an accepting tree and the state $q$ of $c(v)$ is an element of $Q_\exists$,
the node $v$ is an inner node of $T$, hence, it has a successor $v'$.
Let us assume that $((q,\eta),(r,\theta,\mathit{left})) \in \delta$
is the element of the transition relation $\delta$ that leads from $v$ to $v'$
(the case that this element is of the form
$((q,\eta),(r,\theta,\mathit{right}))$ is treated analogously).
We claim that then
\[
(v,x) \models \mathit{compstep}_{\mathrm{left}}(r,\theta) .
\]
In fact, we observe $(v,x) \stackrel{L}{\to} (v,v')$ and
$(v,v') \stackrel{\Diamond}{\to} (v',v')$.
We claim that the two points
$(v,v')$ and $(v',v')$ have the properties formulated in the formula
$\mathit{compstep}_{\mathrm{left}}(r,\theta)$.
Let us check this.
Let us assume that, for some $k\in\{0,\ldots,N-1\}$
and for some $l \in \{0,\ldots,N\}$, 
\[
(v,x) \models 
\bigl( \mathrm{rightmost\_zero}(\underline{\alpha}^\mathrm{time},k)
           \wedge \mathrm{rightmost\_one}(\underline{\alpha}^\mathrm{pos},l) \bigr) .
\]
The number $i:=\mathit{time}(v)$ is an element of $\{0,\ldots,2^N-2\}$
because $v$ is an inner point of the tree $T$ and the length of any computation path is at most $2^N-1$.
The number $j:= \mathit{pos}(v)$ is an element of $\{1,\ldots,2^{N+1}-3\}$
because the computation starts in cell $2^N-1$ and
during each computation step the tape head can move at most one step to the left or to the right.
We obtain $k=\min(\{0,\ldots,N-1\} \setminus\mathrm{Ones}(i))$
and $l=\min\mathrm{Ones}(j)$.
As $v= \mathit{pred}(v')$ we obtain
\[ (v,v') \models (\underline{X}^\mathrm{prevtime}=\underline{\alpha}^\mathrm{time}) \]
and
\[ (v,v') \models (\underline{X}^\mathrm{prevpos}=\underline{\alpha}^\mathrm{pos}) . \]
And as $\mathit{time}(v')=\mathit{time}(v)+1=i+1$ and
$\mathit{pos}(v')=\mathit{pos}(v)-1=j-1$ we conclude 
that
\begin{eqnarray*}
(v',v') & \models & \Bigl(
       (\underline{\alpha}^\mathrm{time}=\underline{X}^\mathrm{prevtime},>k) 
                        \wedge  \mathrm{rightmost\_one}(\underline{\alpha}^\mathrm{time},k) \\
&& \phantom{\Bigl(}  
       \wedge (\underline{\alpha}^\mathrm{pos}=\underline{X}^\mathrm{prevpos},>l) 
                        \wedge \mathrm{rightmost\_zero}(\underline{\alpha}^\mathrm{pos},l) \Bigr) . 
\end{eqnarray*}
Finally, the condition
\[ (v',v') \models (\underline{\alpha}^\mathrm{prevpos}=\underline{X}^\mathrm{prevpos}) \]
is obviously satisfied
and the condition
\[ (v',v') \models \alpha^\mathrm{state}_r 
     \wedge \alpha^\mathrm{written}_\theta  \]
follows from the fact that $((q,\eta),(r,\theta,\mathit{left})) \in \delta$
is the element of the transition relation $\delta$ that leads from $v$ to $v'$.
This ends the treatment of the conjunctions over the set $(q,\eta) \in Q_\exists \times \Gamma$
in the formula $\mathit{computation}$.
Let us now consider a pair $(q,\eta) \in Q_\forall \times \Gamma$.
Let us assume that $(v,x) \in W$ is a point such that
$(v,x) \models (\alpha^\mathrm{state}_q \wedge \alpha^\mathrm{read}_\eta)$.
We have to show that for all elements
$(r,\theta,\mathit{left}) \in \delta(q,\eta)$ we have
$(v,x) \models \mathit{compstep}_{\mathrm{left}}(r,\theta)$
and for all elements
$(r,\theta,\mathit{right}) \in \delta(q,\eta)$ we have
$(v,x) \models \mathit{compstep}_{\mathrm{right}}(r,\theta)$.
Let us consider an arbitrary element $(r,\theta,\mathit{left}) \in \delta(q,\eta)$
(the case of an element $(r,\theta,\mathit{right}) \in \delta(q,\eta)$
is treated analogously).
As $q \in Q_\forall$ and $T$ is an accepting tree,
in $T$ there is a successor $v'$ of $v$ such that
the element $((q,\eta),(r,\theta,\mathit{left}))$ leads from $v$ to $v'$.
Above, we have already seen that this implies
\[ (v,x) \models \mathit{compstep}_{\mathrm{left}}(r,\theta) . \]
Thus, we have shown $(v,x) \models \mathit{computation}$
for all $(v,x) \in W$.
This ends the proof of the claim that in the 
$\sxs$-product model $(W_1\times W_2,\stackrel{\Diamond}{\to},\stackrel{L}{\to},\sigma)$
that we constructed for $w\in L$ 
we have $(\mathit{root},\mathit{root}) \models f_\sxs(w)$.

\subsection{Existence of an Accepting Tree}
\label{subsection:existence-of-an-acceptance-tree}

We come to the other direction. Let $w \in \Sigma^*$.
We wish to show that if $f_\sxs(w)$
is $\sxs$-satisfiable then $w\in L$.
We will show that any $\sxs$-product model
essentially contains an accepting tree of the Alternating Turing Machine $M$
on input $w$. Of course, this implies $w\in L$.

Let us sketch the main idea.
We will consider an $\sxs$-product model of $f_\sxs(w)$.
And we will consider partial trees of $M$ on input $w$
as considered in Subsection~\ref{subsection:ATM}, for any $w\in\Sigma^*$.
First, we will show that a certain very simple partial tree of $M$ on input $w$
``can be mapped to'' the model (later we will give a precise
meaning to ``can be mapped to'').
Then we will show that any partial tree of $M$ on input $w$
that can be mapped to the model 
and that is not an accepting tree of $M$ on input $w$
can be properly extended to a
strictly larger partial tree of $M$ on input $w$ that can be mapped to the model as well.
If there would not exist an accepting tree of $M$ on input $w$ then
we would obtain an infinite strictly increasing sequence
of partial trees of $M$ on input $w$.
But we show that this cannot happen by giving a 
finite upper bound on the size of partial trees of $M$ on input $w$.

Let $w \in \Sigma^*$ be a string such that the formula
$f_\sxs(w)$ is $\sxs$-satisfiable.
We set $n:=|w|$.
Let $(W_1,R_\Diamond)$ be an $S4$-frame, let
$(W_2,R_L)$ be an $S5$-frame, let $\sigma:AT\to \mathcal{P}(W_1 \times W_2)$
be a function such that the quadruple $(W_1\times W_2, \stackrel{\Diamond}{\to},\stackrel{L}{\to},\sigma)$
where $\stackrel{\Diamond}{\to}$ and $\stackrel{L}{\to}$ are defined as in \cite[Definition 2.3]{HK2019-1} is an $\sxs$-product model, and let $(r_1,r_2) \in W_1\times W_2$ be a point with
$(r_1,r_2) \models f_\sxs(w)$.
The quintuple
$$\mathit{Model}:=(W_1\times W_2, \stackrel{\Diamond}{\to},\stackrel{L}{\to},\sigma,(r_1,r_2))$$
will be important in the following.
We claim that we can assume without loss of generality that
$r_1 \stackrel{\Diamond}{\to} x$ for all $x\in W_1$ and
$R_L = W_2\times W_2$. Otherwise, instead of $W_1$ we could consider
the set $W_1^\prime:=\{v\in W_1 \mid r_1 \stackrel{\Diamond}{\to} v\}$ and
instead of $W_2$
we could consider the set $W_2^\prime:=$ the $R_L$-equivalence
class of $r_2$ and the restrictions $\stackrel{L}{\to}^\prime$ resp. $\stackrel{\Diamond}{\to}^\prime$
of $\stackrel{L}{\to}$ resp. $\stackrel{\Diamond}{\to}$ to $W_1^\prime \times W_2^\prime$.
By structural induction one shows that
for any bimodal formula $\varphi$ and for any $(v,x) \in W_1^\prime \times W_2^\prime$, 
\[ (W_1\times W_2, \stackrel{\Diamond}{\to},\stackrel{L}{\to},\sigma),(v,x) \models \varphi
  \iff (W_1^\prime\times W_2^\prime, \stackrel{\Diamond}{\to}^\prime,\stackrel{L}{\to}^\prime,\sigma),(v,x) \models \varphi . \]
So, we shall assume that
$r_1 \stackrel{\Diamond}{\to} x$, for all $v\in W_1$, and $R_L = W_2\times W_2$.
Note that this implies that
if $\varphi$ is a formula with $(r_1,r_2) \models K \Box \varphi$
then, for all $(v,x) \in W_1 \times W_2$, we have $(v,x) \models \varphi$.

For every $v\in W_1$, the set
\[ \mathrm{Cloud}(v) := \{ (v,x) ~:~ v \in W_1, x \in W_2\} \]
is the $\stackrel{L}{\to}$-equivalence class (short: \emph{cloud})
of any element $y \in W_1 \times W_2$ whose first component is $v$. 
Remember that for every $\stackrel{L}{\to}$-equivalence class
and every shared variable 
$\alpha^{\mathit{string}}_i$ for 
$\mathit{string} \in \{\mathrm{time}, \mathrm{pos},
       \mathrm{state}, \mathrm{read}, \mathrm{prevpos}, \mathrm{written}\}$
and natural numbers $i$, the truth value of this shared variable 
is the same in all elements of the $\stackrel{L}{\to}$-equivalence class.

Partial trees of $M$ on input $w$ as introduced in 
Subsection~\ref{subsection:ATM} will play an important role
in the following.
We will write a partial tree of $M$ on input $w$
similarly as in Subsection~\ref{subsection:ATM}
as a triple $T=(V,E,c)$, but with the difference
that we will describe configurations as at the beginning of Subsection~\ref{subsection:definition-ATMs-S4XS5}:
the labeling function $c$ will be a function of the form
$c:V\to Q \times \{0,\ldots,2^{N+1}-2\}\times \Gamma^{2^{N+1}-1}$.
If $T=(V,E,c)$ is a partial tree of $M$ on input $w$
with root $\mathit{root}$
then a function $\pi:V\to W_1$ is called a
\emph{morphism from $T$ to $\mathit{Model}$}
if it satisfies the following four conditions:
\begin{enumerate}
\item
$\pi(\mathit{root}) = r_1$.
\item
$\{(\pi(v),\pi(v')) ~:~ v,v' \in V \text{ and } vEv'\} \subseteq R_\Diamond$.
\item
$\begin{array}[t]{l}
(\forall v \in V\setminus\{\mathit{root}\}) \, (\exists x \in W_2) \\
(\pi(v),x) \models 
\bigl(  (\underline{\alpha}^\mathrm{prevpos} = \mathrm{bin}_{N+1}(\mathit{pos}(\mathit{pred}(v))))
\wedge \alpha^\mathrm{written}_{\mathit{written}(v)} \bigr)
\end{array}$.
\item
$\begin{array}[t]{l}
(\forall v \in V) \, (\exists x \in W_2) \\ 
(\pi(v),x) \models 
 \bigl( (\underline{\alpha}^\mathrm{time} = \mathrm{bin}_N(\mathit{time}(v)))
  \wedge (\underline{\alpha}^\mathrm{pos} = \mathrm{bin}_{N+1}(\mathit{pos}(v)))\\ 
  \qquad\qquad\qquad\wedge \alpha^\mathrm{state}_{\mathit{state}(v)}
  \wedge \alpha^\mathrm{read}_{\mathit{read}(v)} \bigr).
\end{array}$
\end{enumerate}

We say that $T$ \emph{can be mapped to} $\mathit{Model}$ if there exists a morphism
from $T$ to $\mathit{Model}$.
Below we shall prove the following lemma.

\begin{lemma}
\label{lemma:inductionS4S5}
If a partial tree $T=(V,E,c)$ of $M$ on input $w$
can be mapped to $\mathit{Model}$
and is not an accepting tree of $M$ on input $w$
then there exists a partial tree $T=(\widetilde{V},\widetilde{E},\widetilde{c})$ of $M$ on input $w$
that can be mapped to $\mathit{Model}$ and that satisfies
$V \subsetneq \widetilde{V}$.
\end{lemma}

Before we prove this lemma, we deduce the desired assertion from it.
Let 
\[ D:= \max(\{ |\delta(q,\eta)| ~:~ q \in Q, \eta \in\Gamma\}) . \]
Then, due to Condition III in the definition of a
``partial tree of $M$ on input $w$'',
any node in any partial tree of $M$ on input $w$ has
at most $D$ successors.  
As any computation of $M$ on input $w$ stops after
at most $2^N-1$ steps, any partial tree of $M$ on input $w$ has
at most 
\[ \widetilde{D}:=(D^{2^N}-1)/(D-1) \]
nodes.  

We claim that the rooted and labeled tree
\[ T_0:=(\{\mathit{root}\},\emptyset,c) \text{ where }
c(\mathit{root}):=(q_0,2^N-1,\#^{2^N}w\#^{2^N-1-n}) \]
is a partial tree of $M$ on input $w$ and 
that it can be mapped to $\mathit{Model}$.
Indeed, Condition I in the definition of a ``partial tree of $M$ on input $w$'' is
satisfied because the node $\mathit{root}$ is labeled with the initial configuration of $M$ on input $w$.
Conditions II, III, and IV are satisfied because $T_0$ does not have any inner nodes.
Condition $\mathrm{V}^\prime$ is satisfied because $\mathit{state}(\mathit{root}) = q_0$
and, due to $(r_1,r_2) \models \alpha^\mathrm{state}_{q_0}$
(this is a part of $(r_1,r_2) \models \mathit{start}$)
and $(r_1,r_2) \models \neg \alpha^\mathrm{state}_{q_\mathrm{reject}}$
(this follows from $(r_1,r_2) \models K \Box \mathit{no\_reject}$)
we obtain $q_0 \neq q_\mathrm{reject}$.
Thus, $T_0$ is indeed a partial tree of $M$ on input $w$.
Now we show that $T_0$ can be mapped to $\mathit{Model}$.
Of course, we define the function
$\pi:\{\mathit{root}\}\to W_1$ by $\pi(\mathit{root}):=r_1$.
\begin{enumerate}
\item
The condition $\pi(\mathit{root})=r_1$ is true by definition.
\item
The tree $T_0$ does not have any edges, that is, its set $E$ of edges is empty. 
So, the second condition is satisfied.
\item
The third condition does not apply to the tree $T_0$
because $T_0$ has only one node, its root.
\item
On the one hand, we have
$$\mathit{time}(\mathit{root})=0, \qquad \mathit{pos}(\mathit{root})=2^N-1,$$
$$\mathit{state}(\mathit{root})=q_0, \qquad \text{and }\mathit{read}(\mathit{root})=\#.$$
On the other hand, the condition $(r_1,r_2) \models \mathit{start}$ says
\[ (r_1,r_2) \models \bigl( (\underline{\alpha}^\mathrm{time}=\bin_N(0))
   \wedge (\underline{\alpha}^\mathrm{pos}=\bin_{N+1}(2^N-1))
   \wedge \alpha^\mathrm{state}_{q_0} . \]
We still wish to show $(r_1,r_2) \models \alpha^\mathrm{read}_{\#}$.   
We have $(r_1,r_2) \models \mathit{read\_a\_symbol}$.
This implies
$(r_1,r_2) \models \mathit{existence\_of\_a\_reading\_point}$, and this implies
that there exists some $y\in W_2$ with 
\[ (r_1,y) \models \bigl((\underline{X}^\mathrm{pos} = \underline{\alpha}^\mathrm{pos} )
                        \wedge (\underline{X}^\text{time-apv} \leq \underline{\alpha}^\mathrm{time} )
                        \wedge B^\mathrm{active}\bigr) , \]
hence, with
\[ (r_1,y) \models \bigl((\underline{X}^\mathrm{pos} = \mathrm{bin}_{N+1}(2^N-1)) 
                        \wedge (\underline{X}^\text{time-apv} = \mathrm{bin}(0) )
                        \wedge B^\mathrm{active}\bigr) . \]
Now $(r_1,y) \models \mathit{initial\_symbols}$ implies
$(r_1,y) \models X^\mathrm{read}_\#$.
Finally, the condition
\[ (r_1,y) \models \mathit{storing\_the\_read\_symbol} \]
implies $(r_1,y) \models \alpha^\mathrm{read}_\#$,
and this implies $(r_1,r_2) \models \alpha^\mathrm{read}_\#$.
\end{enumerate}

If there would not exist an accepting tree of $M$ on input $w$
then, starting with $T_0$ and using Lemma~\ref{lemma:inductionS4S5},
we could construct an infinite sequence of partial
trees $T_0,T_1,T_2,\ldots$ of $M$ on input $w$ that 
can be mapped to $\mathit{Model}$ 
such that the number of nodes in these trees is strictly increasing.
But we have seen that any partial tree of $M$ input $w$ can have at most $\widetilde{D}$ nodes.
Thus, there exists an accepting tree of $M$ on input $w$.
We have shown $w\in L$.

In order to complete the proof of Theorem~\ref{theorem:S4S5-EXPSPACE-hard}
it remains to prove Lemma~\ref{lemma:inductionS4S5}.

\begin{proof}[Proof of Lemma~\ref{lemma:inductionS4S5}]
Let $T=(V,E,c)$ be a partial tree of $M$ on input $w$
that is not an accepting tree of $M$ on input $w$
and that can be mapped to $\mathit{Model}$.
Then $T$ has a leaf $\widehat{v}$ such that the state
$q:=\mathit{state}(\widehat{v})$ is either an element of $Q_\exists$ or of $Q_\forall$.
First we treat the case that it is an element of $Q_\exists$,
then the case that it is an element of $Q_\forall$.
Let $\pi:V \to W_1$ be a morphism from $T$ to $\mathit{Model}$.

So, let us assume that $q \in Q_\exists$.
We define $\eta:=\mathit{read}(\widehat{v})$.
As $\pi:V \to W_1$ is a morphism from $T$ to $\mathit{Model}$
there exists an $x' \in W_2$ with
\[ (\pi(\widehat{v}),x') \models 
 (\underline{\alpha}^\mathrm{time} = \mathrm{bin}_N(\mathit{time}(\widehat{v})))
  \wedge (\underline{\alpha}^\mathrm{pos} = \mathrm{bin}_{N+1}(\mathit{pos}(\widehat{v}))) 
  \wedge \alpha^\mathrm{state}_q 
   \wedge \alpha^\mathrm{read}_\eta . \]
Hence, due to $(\pi(\widehat{v}),x') \models \mathit{computation}$,
\[ (\pi(\widehat{v}),x') \models 
 \bigvee_{(r,\theta,\mathit{left}) \in \delta(q,\eta)}  \mathit{compstep}_{\mathrm{left}}(r,\theta)
   \vee   \bigvee_{(r,\theta,\mathit{right}) \in \delta(q,\eta)}  \mathit{compstep}_{\mathrm{right}}(r,\theta) . \]
Let us assume there is an element $(r,\theta,\mathit{left}) \in \delta(q,\eta)$
such that $(\pi(\widehat{v}),x') \models  \mathit{compstep}_{\mathrm{left}}(r,\theta)$
(the other case, the case when there is an element $(r,\theta,\mathit{right}) \in \delta(q,\eta)$
such that $(\pi(\widehat{v}),x') \models \mathit{compstep}_{\mathrm{right}}(r,\theta)$, is treated analogously).  
We claim that we can define the new tree $\widetilde{T}=(\widetilde{V},\widetilde{E},\widetilde{c})$
as follows:
{\setlength{\leftmargini}{2.1em}\begin{itemize}
\item
$\widetilde{V} := V \cup\{\widetilde{v}\}$ where $\widetilde{v}$ is a new element (not in $V$),
\item
$\widetilde{E} := E \cup \{(\widehat{v},\widetilde{v})\}$,
\item
$\widetilde{c}(x) := \begin{cases}
          c(x) & \text{for all } x \in V, \\
          c' & \text{for } x=\widetilde{v}, \text{ where $c'$ is the configuration that is reached from $c(\widehat{v})$} \\
             & \text{in the computation step given by } ((q,\eta),(r,\theta,\mathit{left})) \in\delta.
             \end{cases}$
\end{itemize}}
We have to show that $\widetilde{T}$ is a partial tree of $M$ on input $w$.
Condition I in the definition of ``a partial tree of $M$ on input $w$''
is satisfied because $\widetilde{T}$ has the same root as $T$,
and the label of the root does not change.
A node in $\widetilde{T}$ is an internal node of $\widetilde{T}$ if, and only if, it is either an internal node of $T$
or equal to $\widehat{v}$.
For internal nodes of $T$ Conditions II, III, and IV are satisfied by assumption
(and due to the fact that the labels of nodes in $V$ do not change
when moving from $T$ to $\widetilde{T}$).
The new internal node $\widehat{v}$ satisfies Condition II by our definition of $\widetilde{c}(\widetilde{v})$.
Condition III is satisfied for $\widehat{v}$ because $\widehat{v}$ has exactly one successor.
And Condition IV does not apply to $\widehat{v}$ because $\widehat{v}\in Q_\exists$.
A node in $\widetilde{T}$ is a leaf if, and only if, it is either equal to $\widetilde{v}$
or a leaf in $T$ different from $\widehat{v}$.
For the leaves in $T$ different from $\widehat{v}$ Condition $\mathrm{V}^\prime$ is satisfied by assumption
(and due to the fact that the labels of nodes in $V$ do not change).
Finally, we have to show that Condition $\mathrm{V}^\prime$ 
is satisfied for the new leaf $\widetilde{v}$ in $\widetilde{T}$ as well.
We postpone this until after the definition of a morphism from $\widetilde{T}$ to $\mathit{Model}$.

We also have to show that $\widetilde{T}$ can be mapped to $\mathit{Model}$.
Let us define a function $\widetilde{\pi}:\widetilde{V}\to W_1$
that we will show to be a morphism from $\widetilde{T}$ to $\mathit{Model}$.
As $\widehat{v}$ is an element of a partial tree of $M$ on input $w$
with $\mathit{state}(\widehat{v}) \in Q_\exists$, at least one more computation step
can be done. As any computation of $M$ on input $w$ stops after at most $2^N-1$ steps,
we observe that the number $i:=\mathit{time}(\widehat{v})$ 
satisfies $0 \leq i < 2^N-1$. 
Then $\{0,\ldots,N-1\} \setminus\mathrm{Ones}(i) \neq \emptyset$.
We set $k:= \min(\{0,\ldots,N-1\} \setminus\mathrm{Ones}(i))$.
Together with $(\pi(\widehat{v}),x') \models (\underline{\alpha}^\mathrm{time} = \mathrm{bin}_N(i))$
we conclude 
$(\pi(\widehat{v}),x') \models \mathrm{rightmost\_zero}(\underline{\alpha}^\mathrm{time},k)$.
As during each computation step, the tape head can move at most one step to the left or to the right
and as the computation started in position $2^N-1$
the number $j:=\mathit{pos}(\widehat{v})$ satisfies
$0 < j \leq 2^{N+1}-3$.
Then $\mathrm{Ones}(j) \neq \emptyset$.
We set $l:= \min\mathrm{Ones}(j)$.
Together with
$(\pi(\widehat{v}),x') \models (\underline{\alpha}^\mathrm{pos} = \mathrm{bin}_{N+1}(j))$
we conclude
$(\pi(\widehat{v}),x') \models \mathrm{rightmost\_one}(\underline{\alpha}^\mathrm{pos},l)$.
Thus, we have
\[
(\pi(\widehat{v}),x') \models \bigl( \mathrm{rightmost\_zero}(\underline{\alpha}^\mathrm{time},k)
    \wedge \mathrm{rightmost\_one}(\underline{\alpha}^\mathrm{pos},l) \bigr).
\]
Due to $(\pi(\widehat{v}),x') \models  \mathit{compstep}_{\mathrm{left}}(r,\theta)$
there exist an element $y \in W_2$ and an element $x \in W_1$ such that
$\pi(\widehat{v}) R_\Diamond x$
as well as
\[ (\pi(\widehat{v}),y) \models
     \bigl(  (\underline{X}^\mathrm{prevtime}=\underline{\alpha}^\mathrm{time})
                             \wedge  (\underline{X}^\mathrm{prevpos}=\underline{\alpha}^\mathrm{pos})
                        \bigr) \]
and
\begin{eqnarray*}
(x,y) \models
&&  \Bigl(
       (\underline{\alpha}^\mathrm{time}=\underline{X}^\mathrm{prevtime},>k) 
                        \wedge  \mathrm{rightmost\_one}(\underline{\alpha}^\mathrm{time},k) \\
&&  \wedge (\underline{\alpha}^\mathrm{pos}=\underline{X}^\mathrm{prevpos},>l) 
                        \wedge \mathrm{rightmost\_zero}(\underline{\alpha}^\mathrm{pos},l) \\
&&  \wedge ( \underline{\alpha}^\mathrm{prevpos}=\underline{X}^\mathrm{prevpos}) 
                        \wedge \alpha^\mathrm{state}_{r} 
                        \wedge \alpha^\mathrm{written}_\theta 
\Bigr) 
\end{eqnarray*}
We claim that we can define the desired function $\widetilde{\pi}:\widetilde{V} \to W_1$ by
\[ \widetilde{\pi}(v) := \begin{cases}
       \pi(v) & \text{if } v \in V, \\
       x & \text{if } v=\widetilde{v}.
       \end{cases} 
\]
Before we show that $\widetilde{\pi}$ is a morphism from $\widetilde{T}$ to $\mathit{Model}$,
let us complete the proof that $\widetilde{T}$ is a partial tree of $M$ on input $w$.
We still need to show that Condition $\mathrm{V}^\prime$ 
is satisfied for the new leaf $\widetilde{v}$ in $\widetilde{T}$ as well.
It is sufficient to show that the state $r$ of the configuration $c'$ 
is not the rejecting state $q_\mathrm{reject}$.
But this follows from
$(x,y) \models  \alpha^\mathrm{state}_r$ and
$(x,y) \models \neg \alpha^\mathrm{state}_{q_\mathrm{reject}}$
(this follows from $(r_1,r_2) \models K \Box \mathit{no\_reject}$).
We have shown that $\widetilde{T}$ is a partial tree of $M$ on input $w$.

Now we show that $\widetilde{\pi}$ is a morphism from $\widetilde{T}$ to $\mathit{Model}$.
The first condition in the definition of a ``morphism from $\widetilde{T}$ to $\mathit{Model}$''
is satisfied because $\widetilde{\pi}(\mathit{root}) = \pi(\mathit{root})=r_1$.
For the second condition let us consider $v,v' \in \widetilde{V}$ with $v\widetilde{E}v'$.
We have to show $\pi(v) R_\Diamond \pi(v')$.
There are two possible cases.
\begin{itemize}
\item
In the case $v,v' \in V$ we have $vEv'$ and, hence, $\pi(v)R_\Diamond \pi(v')$.
As $\widetilde{\pi}(v)=\pi(v)$ and $\widetilde{\pi}(v')=\pi(v')$ we obtain
$\widetilde{\pi}(v) R_\Diamond \widetilde{\pi}(v')$.
\item
The other possible case is $v=\widehat{v}$ and $v'=\widetilde{v}$.
But we know $\widetilde{\pi}(\widehat{v})=\pi(\widehat{v})$,
$\widetilde{\pi}(\widetilde{v}) = x$, and
$\pi(\widehat{v}) R_\Diamond x$.
\end{itemize}
Next, we verify that the third condition in the definition of a
``morphism from $\widetilde{T}$ to $\mathit{Model}$'' is satisfied.
For $v\in V\setminus\{\mathit{root}\}$ 
it is satisfied by assumption
(and by $\widetilde{\pi}(v)=\pi(v)$ and $\widetilde{c}(v) = c(v)$).
For $\widetilde{v}$ it is sufficient to show that 
\[ 
(x,y) \models \bigl( 
(\underline{\alpha}^\mathrm{prevpos} = \mathrm{bin}_{N+1}(\mathit{pos}(\mathit{pred}(\widetilde{v}))))
\wedge \alpha^\mathrm{written}_{\mathit{written}(\widetilde{v})} \bigr)
\]
(remember that $\widetilde{\pi}(\widetilde{v})=x$).
These are really two conditions. We prove them separately.
\begin{itemize}
\item
In the tree $\widetilde{T}$ we have $\mathit{pos}(\mathit{pred}(\widetilde{v}))=\mathit{pos}(\widehat{v})=j$,
and in $T$ we have $\mathit{pos}(\widehat{v})=j$ as well.
The assumption 
$(\pi(\widehat{v}),x') \models (\underline{\alpha}^\mathrm{pos}= \mathrm{bin}_{N+1}(j))$,
for some $x'\in W_2$,
implies $(\pi(\widehat{v}),y) \models (\underline{\alpha}^\mathrm{pos}= \mathrm{bin}_{N+1}(j))$,
and the conditions $(\pi(\widehat{v}),y) \models  (\underline{X}^\mathrm{prevpos}=\underline{\alpha}^\mathrm{pos})$
and $(x,y) \models (\underline{\alpha}^\mathrm{prevpos}=\underline{X}^\mathrm{prevpos})$
as well as the persistence of $\underline{X}^\mathrm{prevpos}$ imply
$(x,y) \models (\underline{\alpha}^\mathrm{prevpos}=\mathrm{bin}_{N+1}(j))$.
\item
In $\widetilde{T}$ we have $\mathit{written}(\widetilde{v}) = \theta$.
And we have $(x,y) \models \alpha^\mathrm{written}_\theta$,
hence, $(x,y) \models \alpha^\mathrm{written}_{\mathit{written}(\widetilde{v})}$.
\end{itemize}
We come to the fourth condition in the definition of a
``morphism from $\widetilde{T}$ to $\mathit{Model}$''.
It is satisfied for $v\in V$ by assumption
(and due to $\widetilde{\pi}(v)=\pi(v)$ and $\widetilde{c}(v) = c(v)$).
We still need to show that it is satisfied for $v=\widetilde{v}$.
Remember $\widetilde{\pi}(\widetilde{v})=x$.
It is sufficient to show
\[ 
(x,y) \models
(\underline{\alpha}^\mathrm{time} = \mathrm{bin}_N(\mathit{time}(\widetilde{v})))
  \wedge (\underline{\alpha}^\mathrm{pos} = \mathrm{bin}_{N+1}(\mathit{pos}(\widetilde{v}))) 
  \wedge \alpha^\mathrm{state}_{\mathit{state}(\widetilde{v})}
  \wedge \alpha^\mathrm{read}_{\mathit{read}(\widetilde{v})}.
\]
This assertion consists really of four assertions. We treat them one by one.
\begin{itemize}
\item
In the trees $T$ and $\widetilde{T}$ we have $\mathit{time}(\widehat{v})=i$, 
and in the tree $\widetilde{T}$ we have $\mathit{time}(\widetilde{v})=i+1$.
We have already seen that $(\pi(\widehat{v}),x') \models (\underline{\alpha}^\mathrm{time} = \mathrm{bin}_N(i))$
and that
$(\pi(\widehat{v}),x') \models \mathrm{rightmost\_zero}(\underline{\alpha}^\mathrm{time},k)$.
Of course, we get
\[ 
(\pi(\widehat{v}),y) \models ((\underline{\alpha}^\mathrm{time} = \mathrm{bin}_N(i))
\wedge \mathrm{rightmost\_zero}(\underline{\alpha}^\mathrm{time},k)) . 
\]
The conditions 
\begin{eqnarray*}
(\pi(\widehat{v}),y) &\models& (\underline{X}^\mathrm{prevtime}=\underline{\alpha}^\mathrm{time}), \\
(x,y) &\models&
       (\underline{\alpha}^\mathrm{time}=\underline{X}^\mathrm{prevtime},>k) 
                        \wedge  \mathrm{rightmost\_one}(\underline{\alpha}^\mathrm{time},k) 
\end{eqnarray*}
and the persistence of $\underline{X}^\mathrm{prevtime}$
imply $(x,y) \models (\underline{\alpha}^\mathrm{time} = \mathrm{bin}_N(i+1))$.
\item
In the trees $T$ and $\widetilde{T}$ we have $\mathit{pos}(\widehat{v})=j$, 
and in the tree $\widetilde{T}$ we have $\mathit{pos}(\widetilde{v})=j-1$.
We have already seen
$(\pi(\widehat{v}),x') \models (\underline{\alpha}^\mathrm{pos} = \mathrm{bin}_{N+1}(j))$,
and
$(\pi(\widehat{v}),x') \models \mathit{rightmost\_one}(\underline{\alpha}^\mathrm{pos},l)$.
Of course, we get
\[ 
(\pi(\widehat{v}),y) \models ((\underline{\alpha}^\mathrm{pos} = \mathrm{bin}_{N+1}(j))
\wedge \mathrm{rightmost\_one}(\underline{\alpha}^\mathrm{pos},l)) . 
\]
The conditions 
\begin{eqnarray*}
(\pi(\widehat{v}),y) &\models& (\underline{X}^\mathrm{prevpos}=\underline{\alpha}^\mathrm{pos}), \\
(x,y) &\models&
       (\underline{\alpha}^\mathrm{pos}=\underline{X}^\mathrm{prevpos},>l) 
                        \wedge  \mathrm{rightmost\_zero}(\underline{\alpha}^\mathrm{pos},l) 
\end{eqnarray*}
and the persistence of $\underline{X}^\mathrm{prevpos}$
imply $(x,y) \models (\underline{\alpha}^\mathrm{pos} = \mathrm{bin}_{N+1}(j-1))$.
\item
We have $\mathit{state}(\widetilde{v}) = r$. And we have $(x,y) \models \alpha^\mathit{state}_r$.
\item
Let $\gamma:=\mathit{read}(\widetilde{v})$ in $\widetilde{T}$.
We wish to show 
$(x,y) \models \alpha^\mathrm{read}_\gamma$.
We remark that the proof is a formal version of the informal explanation after the definition of the
formula $\mathit{read\_a\_symbol}$.
It is sufficient to show 
that there is some $z \in W_2$ with $(x,z) \models \alpha^\mathrm{read}_\gamma$.
The condition
$(x,y) \models \mathit{existence\_of\_a\_reading\_point}$ implies that there exists
some $z \in W_2$ such that
\[
(x,z) \models \bigl( ( \underline{X}^\mathrm{pos} = \underline{\alpha}^\mathrm{pos} )
                        \wedge (\underline{X}^\text{time-apv} \leq \underline{\alpha}^\mathrm{time} )
                        \wedge B^\mathrm{active}\bigr).
\]
Remember that the binary value of $\underline{\alpha}^\mathrm{time}$
in $(x,y)$ and, hence, also in $(x,z)$, is equal to $i+1$ and that the binary value of $\underline{\alpha}^\mathrm{pos}$
in $(x,y)$ and, hence, also in $(x,z)$, is equal to $j-1$.
Hence, the binary value of $\underline{X}^\mathrm{pos}$ in $(x,z)$ is equal to $j-1$.
Let $t$ be the unique number in $\{0,\ldots,i+1\}$ with
$(x,z) \models (\underline{X}^\text{time-apv} = \bin_N(t))$.
Let $v_t$ be the unique node in the computation path in $\widetilde{T}$
from $\mathit{root}$ to $\widetilde{v}$ with $\mathit{time}(v_t)=t$.

First, we claim that
$(\widetilde{\pi}(u),z) \models B^\mathrm{active}$
for all nodes $u$ on the path from $\mathit{root}$ to $\widetilde{v}$.
Any such $u$ satisfies
$uE^*\widetilde{v}$. We obtain
$\widetilde{\pi}(u) R_\Diamond \widetilde{\pi}(\widetilde{v})$, hence,
$\widetilde{\pi}(u) R_\Diamond x$.
If $(\widetilde{\pi}(u),z) \models \neg B^\mathrm{active}$
then due to 
$(\widetilde{\pi}(u),z) \models \mathit{staying\_inactive}$,
we would obtain $(x,z) \models \neg B^\mathrm{active}$.
But this is a contradiction to the condition
$(x,z) \models B^\mathrm{active}$ with which we started.
Thus, for all nodes on the path from $\mathit{root}$ to $\widetilde{v}$
we have $(\widetilde{\pi}(u),z) \models B^\mathrm{active}$.

Can there be a node $u\neq v_t$ in the path from $v_t$ to $\widetilde{v}$
with $\mathit{pos}(\mathit{pred}(u))=j-1=\mathit{pos}(\widetilde{v})$? 
We claim that this is not the case.
Indeed, if there were such a $u$ then for this node $u$ we would have
$(\widetilde{\pi}(u),z) \models ( \underline{X}^\mathrm{pos}=\underline{\alpha}^\mathrm{prevpos})
  \wedge (\underline{X}^\text{time-apv} < \underline{\alpha}^\mathrm{time} )$.
But then $(\widetilde{\pi}(u),z) \models \mathit{becoming\_inactive}$
would imply $(\widetilde{\pi}(u),z) \models \neg B^\mathrm{active}$ in contradiction
to what we have just shown.
Hence, we have shown that there is no node $u\neq v_t$ in the path from $v_t$ to $\widetilde{v}$
with $\mathit{pos}(\mathit{pred}(u))=j-1=\mathit{pos}(\widetilde{v})$.
Let us now distinguish the two cases $t=0$ and $t>0$.

First we treat the case $t=0$.
Then $v_t=\mathit{root}$. We have just seen that the cell $j-1$ has not been visited
before $\widetilde{v}$ on the path from $\mathit{root}$ to $\widetilde{v}$.
Hence, the initial symbol in the cell $j-1$ is still the symbol in this cell when the node 
$\widetilde{v}$ is reached. Thus $\gamma$ is the initial symbol in this cell.
Due to
$(x,z) \models \mathit{initial\_symbols}$ we obtain
$(x,z) \models X^\mathrm{read}_\gamma$.
Due to
$(x,z) \models \mathit{storing\_the\_read\_symbol}$
we obtain $(x,z) \models \alpha^\mathrm{read}_\gamma$.

Finally, we treat the case $t>0$.
Then $v_t \neq \mathit{root}$. And we have just seen that the cell $j-1$ has not been visited
before $\widetilde{v}$ on the path from $\mathit{v_t}$ to $\widetilde{v}$.
But we claim that it has been visited in the predecessor of $v_t$.
Indeed, we have
$(\widetilde{\pi}(v_t),z) \models 
\bigl(( \underline{X}^\text{time-apv} > \mathrm{bin}_N(0) )
                               \wedge (\underline{X}^\text{time-apv} = \underline{\alpha}^\mathrm{time})
                               \wedge B^\mathrm{active} \bigr)$.
Hence, due to
$(\widetilde{\pi}(v_t),z) \models \mathit{time\_of\_previous\_visit}$, we obtain
$(\widetilde{\pi}(v_t),z) \models ( \underline{X}^\mathrm{pos} = \underline{\alpha}^\mathrm{prevpos})$.
Above we have seen that the binary value of $\underline{X}^\mathrm{pos}$ in $(x,z)$ is $j-1$.
As $\underline{X}^\mathrm{pos}$ is persistent, the binary value of $\underline{X}^\mathrm{pos}$ in 
$(\widetilde{\pi}(v_t),z)$ is $j-1$ as well. Hence, the binary value of $\underline{\alpha}^\mathrm{prevpos}$ in
$(\widetilde{\pi}(v_t),z)$ is $j-1$ as well.
That implies $\mathit{pos}(\mathit{pred}(v_t))=j-1$.
We have shown that the predecessor of the node $v_t$ is the last node before
$\widetilde{v}$ in which the cell $j-1=\mathit{pos}(\widetilde{v})$ has been visited.
Hence, the symbol $\gamma$ that is read when the node $\widetilde{v}$ is reached,
has been written in the computation step from $\mathit{pred}(v_t)$ to $v_t$.
Hence, we have $(\widetilde{\pi}(v_t),z)\models \alpha^\mathrm{written}_\gamma$.
Due to 
$(\widetilde{\pi}(v_t),z) \models \mathit{written\_symbols}$ we obtain
$(\widetilde{\pi}(v_t),z) \models X^\mathrm{written}_\gamma$.
As above in the other case, due to
$(x,z) \models \mathit{storing\_the\_read\_symbol}$,
we finally obtain $(x,z) \models \alpha^\mathrm{read}_\gamma$.
\end{itemize}
Thus, $\widetilde{T}$ is not only a partial tree of $M$ on input $w$ but can
also be mapped to $\mathit{Model}$.
This ends the treatment of the case $q \in Q_\exists$.

Now we consider the other case, the case $q \in Q_\forall$.
We define $\eta:=\mathit{read}(\widehat{v})$.
Let 
\[ (r_1,\theta_1,\mathit{dir}_1),\ldots,(r_d,\theta_d,\mathit{dir}_d) \]
be the elements of
$\delta(q,\eta)$ where $d\geq 1$ and $\mathit{dir}_m \in \{\mathit{left},\mathit{right}\}$, for $m=1,\ldots,d$.
We claim that we can define the new tree $\widetilde{T}=(\widetilde{V},\widetilde{E},\widetilde{c})$
as follows:
\begin{itemize}
\item
$\widetilde{V} := V \cup\{\widetilde{v}_1,\ldots,\widetilde{v}_d\}$ where $\widetilde{v}_1,\ldots,\widetilde{v}_d$ 
are new (not in $V$) pairwise different elements,
\item
$\widetilde{E} := E \cup \{(\widehat{v},\widetilde{v}_1),\ldots,(\widehat{v},\widetilde{v}_d)\}$,
\item
$\widetilde{c}(x) := \begin{cases}
          c(x) & \text{for all } x \in V, \\
          c'_m & \text{for } x=\widetilde{v}_m, \text{ where $c'_m$ is the configuration that is reached from } c(\widehat{v})\\
         &  \text{in the computation step given by } ((q,\eta),(r_m,\theta_m,\mathit{dir}_m)) \in\delta.
             \end{cases}$
\end{itemize}
Before we show that $\widetilde{T}$
is a partial tree of $M$ on input $w$,
we define a function $\widetilde{\pi}:\widetilde{V}\to W_1$
that we will show to be a morphism from $\widetilde{T}$ to $\mathit{Model}$.
Since $T$ can be mapped to $\mathit{Model}$, we have
\[ (\pi(\widehat{v}),x') \models 
 (\underline{\alpha}^\mathrm{time} = \mathrm{bin}_N(\mathit{time}(\widehat{v})))
  \wedge (\underline{\alpha}^\mathrm{pos} = \mathrm{bin}_{N+1}(\mathit{pos}(\widehat{v}))) 
  \wedge \alpha^\mathrm{state}_q 
   \wedge \alpha^\mathrm{read}_\eta . \]
for some $x' \in W_2$,
hence, due to $(\pi(\widehat{v}),x') \models \mathit{computation}$,
\[ (\pi(\widehat{v}),x') \models 
 \bigwedge_{(r,\theta,\mathit{left}) \in \delta(q,\eta)}  \mathit{compstep}_{\mathrm{left}}(r,\theta)
   \wedge  \bigwedge_{(r,\theta,\mathit{right}) \in \delta(q,\eta)}  \mathit{compstep}_{\mathrm{right}}(r,\theta) . \]
As in the case $q \in Q_\exists$ one shows that the numbers
$i:=\mathit{time}(\widehat{v})$
and
 $j:=\mathit{pos}(\widehat{v})$
satisfy $0 \leq i < 2^N-1$
and
$0 < j \leq 2^{N+1}-3$, and one defines
$$\begin{array}{lll}
k	&:=	& \min(\{0,\ldots,N-1\} \setminus\mathrm{Ones}(i)),\\
l_\mathrm{left}	&:= &\min\mathrm{Ones}(j), \\ 
l_\mathrm{right}	&:= 	&\min(\{0,\ldots,N\} \setminus\mathrm{Ones}(j)).
\end{array}$$
As in the case $q \in Q_\exists$ one obtains
\begin{eqnarray*}
(\pi(\widehat{v}),x') &\models& \bigl( \mathrm{rightmost\_zero}(\underline{\alpha}^\mathrm{time},k) \\
   && \wedge \mathrm{rightmost\_one}(\underline{\alpha}^\mathrm{pos},l_\mathrm{left})
    \wedge \mathrm{rightmost\_zero}(\underline{\alpha}^\mathrm{pos},l_\mathrm{right}) \bigr).
\end{eqnarray*}
Let us consider some $m\in\{1\ldots,d\}$.
If $\mathit{dir}_m=\mathit{left}$ then,
due to $(\pi(\widehat{v}),x') \models  \mathit{compstep}_{\mathrm{left}}(r,\theta)$,
there exist an element $y_m \in W_2$ and an element $x_m \in W_1$ such that
$\pi(\widehat{v}) R_\Diamond x_m$
as well as
\[ (\pi(\widehat{v}),y_m) \models
   ((\underline{X}^\mathrm{prevtime}=\underline{\alpha}^\mathrm{time})
      \wedge  (\underline{X}^\mathrm{prevpos}=\underline{\alpha}^\mathrm{pos})) \]
and
\begin{eqnarray*}
(x_m,y_m) &\models&
\Bigl( (\underline{\alpha}^\mathrm{time}=\underline{X}^\mathrm{prevtime},>k) 
                        \wedge  \mathrm{rightmost\_one}(\underline{\alpha}^\mathrm{time},k) \\
&& \wedge (\underline{\alpha}^\mathrm{pos}=\underline{X}^\mathrm{prevpos},>l_\mathrm{left}) 
                        \wedge \mathrm{rightmost\_zero}(\underline{\alpha}^\mathrm{pos},l_\mathrm{left}) \\
&& \wedge (\underline{\alpha}^\mathrm{prevpos}=\underline{X}^\mathrm{prevpos}) 
                        \wedge \alpha^\mathrm{state}_{r} 
                        \wedge \alpha^\mathrm{written}_\theta  \Bigr) .
\end{eqnarray*}
Similarly, if $\mathit{dir}_m=\mathit{right}$ then,
due to $(\pi(\widehat{v}),x') \models  \mathit{compstep}_{\mathrm{right}}(r,\theta)$,
there exist an element $y_m \in W_2$ and an element $x_m \in W_1$ such that
$\pi(\widehat{v}) R_\Diamond x_m$
as well as
\[ (\pi(\widehat{v}),y_m) \models
   ((\underline{X}^\mathrm{prevtime}=\underline{\alpha}^\mathrm{time})
      \wedge  (\underline{X}^\mathrm{prevpos}=\underline{\alpha}^\mathrm{pos})) \]
and
\begin{eqnarray*}
(x_m,y_m) &\models&
\Bigl( (\underline{\alpha}^\mathrm{time}=\underline{X}^\mathrm{prevtime},>k) 
                        \wedge  \mathrm{rightmost\_one}(\underline{\alpha}^\mathrm{time},k) \\
&& \wedge (\underline{\alpha}^\mathrm{pos}=\underline{X}^\mathrm{prevpos},>l_\mathrm{right}) 
                        \wedge \mathrm{rightmost\_one}(\underline{\alpha}^\mathrm{pos},l_\mathrm{right}) \\
&& \wedge (\underline{\alpha}^\mathrm{prevpos}=\underline{X}^\mathrm{prevpos}) 
                        \wedge \alpha^\mathrm{state}_{r} 
                        \wedge \alpha^\mathrm{written}_\theta  \Bigr) .
\end{eqnarray*}
We claim that we can define the desired function $\widetilde{\pi}:\widetilde{V} \to W_1$ by
\[ \widetilde{\pi}(v) := \begin{cases}
       \pi(v) & \text{if } v \in V, \\
       x_m & \text{if } v=\widetilde{v}_m, \text{ for some } m \in \{1,\ldots,d\}.
       \end{cases} 
\]
Similarly as in the case $q \in Q_\exists$ one shows
that $\widetilde{T}$ is a partial tree of $M$ on input $w$.
Note that also Condition IV is satisfied for $\widehat{v}$.
Finally, similarly as in the case $q \in Q_\exists$ one shows
that $\widetilde{\pi}$ is a morphism from $\widetilde{T}$ to $\mathit{Model}$.
This ends the treatment of the case $q \in Q_\forall$.
We have proved Lemma~\ref{lemma:inductionS4S5}.
\end{proof}

\end{document}